\documentclass[mnsc,nonblindrev]{informs3} 

\OneAndAHalfSpacedXI 
 
\RequirePackage{endnotes} \usepackage{endnotes}

\usepackage{natbib}
\usepackage{bibunits}
\bibpunct[, ]{(}{)}{,}{a}{}{,}%
\def\bibfont{\small}%
\usepackage{xspace} 
\usepackage{endnotes} \usepackage{enumerate} \usepackage{rotating} \usepackage{subfigure}
\usepackage{float} \usepackage{mathtools} \usepackage{booktabs} \usepackage{tabularx} \usepackage{multirow} \usepackage{caption} \usepackage{amsmath} \usepackage{graphicx} \usepackage{enumitem} \usepackage{multirow} \usepackage{makecell} \usepackage{comment} \usepackage{lipsum} \usepackage{threeparttable} \usepackage{hhline} \usepackage{algorithm} \usepackage{algpseudocode} \usepackage{mathrsfs} \usepackage{physics} \usepackage{bbm} \usepackage{svg} \usepackage{bm} \usepackage{mathdots} \usepackage{tikz}
\graphicspath{{./Fig/}} \usetikzlibrary{arrows.meta,positioning,fit,calc} \newcommand{\qedsymbol}{$\square$} \allowdisplaybreaks \newcommand{\supp}{\operatorname{supp}}

\definecolor{DarkBlue}{rgb}{0,0.08,0.45} \definecolor{Navy}{rgb}{0,0,0.5} \definecolor{Salmon}{rgb}{0.95,0.38,0.32} \definecolor{Teal}{rgb}{0,0.45,0.45} \usepackage[colorlinks=true,breaklinks=true,bookmarks=true,urlcolor=DarkBlue,citecolor=DarkBlue,linkcolor=DarkBlue,bookmarksopen=false,draft=false]{hyperref} \usepackage[noabbrev, nameinlink, capitalize]{cleveref} \crefname{assumption}{Assumption}{Assumptions} \crefname{lemma}{Lemma}{Lemmas} \crefname{theorem}{Theorem}{Theorems} \crefname{corollary}{Corollary}{Corollaries} \crefname{proposition}{Proposition}{Propositions} \crefname{claim}{Claim}{Claims} \crefname{algorithm}{Algorithm}{Algorithms} \crefname{figure}{Figure}{Figures} \crefname{remark}{Remark}{Remarks} \crefname{section}{Section}{Sections} \crefname{procedure}{Procedure}{Procedures} \crefname{definition}{Definition}{Definitions} \crefname{example}{Example}{Examples} \crefname{table}{Table}{Tables} \crefname{equation}{}{}

\TheoremsNumberedThrough \ECRepeatTheorems \EquationsNumberedThrough

\newcommand{\E}{\mathbb{E}} \newcommand{\Pbb}{\mathbb{P}}            \newcommand{\OPT}{\mathrm{OPT}} \newcommand{\Reg}{\mathrm{Reg}} \newcommand{\ALG}{\mathrm{ALG}} 

\newcommand{\policyname}[1]{#1\xspace}
\newcommand{\polStationary}{\policyname{Stationary}}
\newcommand{\polBayesian}{\policyname{Bayesian}}
\newcommand{\polRobust}{\policyname{Robust}}
\newcommand{\polTrendAware}{\policyname{Trend-aware}}
\newcommand{\polBuyLast}{\policyname{Buy Last}}
\newcommand{\polBuyNow}{\policyname{Buy Now}}
\newcommand{\polHistMin}{\policyname{Historical Min}}
\newcommand{\polHistMean}{\policyname{Historical Mean}}
\newcommand{\polHistMed}{\policyname{Historical Median}}
\newcommand{\polSelector}{\policyname{Heuristic Selector}}

\MANUSCRIPTNO{}

\DeclareFontShape{OT1}{cmr}{bx}{sc}{<->ssub * cmr/bx/n}{} \emergencystretch=2em \hbadness=10000 \hfuzz=30pt

\makeatletter
\newcommand{\ECBibHyperlinks}{%
  \def\hyper@natlinkstart##1{%
    \Hy@backout{EC.##1}%
    \hyper@linkstart{cite}{cite.EC.##1}%
    \def\hyper@nat@current{EC.##1}%
  }%
  \def\hyper@natlinkbreak##1##2{%
    \hyper@linkend##1\hyper@linkstart{cite}{cite.EC.##2}%
  }%
  \def\hyper@natanchorstart##1{%
    \Hy@raisedlink{\hyper@anchorstart{cite.EC.##1}}%
  }%
}
\makeatother

\defaultbibliographystyle{informs2014}
\defaultbibliography{reference}
\newcommand{\PutSingleSpacedBib}{%
  \begingroup
  \OneAndAHalfSpacedXI
  \putbib
  \endgroup
}

\begin{document}
\begin{bibunit}


\RUNTITLE{Strategic Buying Agents}

\TITLE{\Large Strategic Buying Agents}

 \ARTICLEAUTHORS{%
 \AUTHOR{Mingyang Fu} \AFF{Rotman School of Management, University of Toronto, \EMAIL{mingyang.fu@utoronto.ca}} \AUTHOR{Ming Hu} \AFF{Rotman School of Management, University of Toronto, \EMAIL{ming.hu@rotman.utoronto.ca}} }

\ABSTRACT{%
The emergence of agentic AI is shifting online shopping from search and recommendation toward delegated purchasing by autonomous buying agents that monitor markets, reason about uncertainty, and make purchase decisions on behalf of consumers. We study the design of strategic buying agents that decide when to purchase an item within a finite shopping window. The central challenge is to translate price observations, the shopping window, and information about future price changes into a purchase policy. We formulate this problem across three information regimes: stationary, Bayesian, and robust, and use the resulting optimal policies as a policy menu for implementation. In the stationary regime, price adjustments follow an exogenously specified Poisson arrival process, and post-adjustment prices are drawn from a known stationary distribution. We show that the optimal policy is a dynamic purchase-threshold policy, with the threshold characterized by an ordinary differential equation. In the Bayesian regime, the adjustment intensity is assumed known, but the price-adjustment distribution is uncertain. We show that the optimal rule remains threshold-based, with the threshold depending on posterior beliefs. We also bound the value of knowing the true price-adjustment distribution. In the robust regime, the agent relies only on price bounds and seeks worst-case protection. We show that randomized threshold policies yield optimal guarantees for both competitive ratio and minimax regret. Finally, we evaluate our proposed policies using Amazon product price histories tracked by Keepa, comprising 367 items and 48,933 time-stamped price observations. We also examine how these policies can be incorporated into language-model buying agents. On our test instances, the stationary and Bayesian policies perform competitively in terms of mean normalized consumer surplus despite their stylized assumptions, while the robust policy performs best at the 10th percentile of the normalized surplus distribution. The results also suggest that language models are better suited to choosing among the three information regimes and selecting price samples for calibration than to making buy-or-wait decisions directly.
}%


\maketitle

\section{Introduction}\label{sec:intro_buying}
Online shopping platforms have historically treated human consumers as the primary decision makers. Consumers search, compare products, evaluate prices, and complete checkout, while platforms provide search tools, recommendations, and payment systems. Recent developments in agentic commerce suggest a different mode of interaction, in which consumers may delegate parts of the shopping process to AI agents that monitor products, compare prices, and execute transactions subject to user-specified conditions \citep{OddLotsCollison2026, OpenAIACP2025, StripeACP2025, StripeACS2025, StripeLinkAgents2026, StripeGuide2026}. More broadly, deployed web-agent systems show that agents can already navigate interfaces and execute user-directed actions \citep{openai2025operator,openai2025chatgptagent}.

Together, these AI technologies make delegated shopping feasible. A broad literature on online commerce has generated rich insights into search, recommendation, and consumer choice. Recent work on LLM-enabled operational decision support further shows how language models can structure decision inputs, invoke analytical tools, and communicate recommendations \citep{simchilevi2025llmsupplychaindecisions,baek2026ai}. These advances provide important functional ingredients for shopping agents: an agent may be able to browse product pages, monitor price adjustments, interpret a consumer's request, update its assessment of the shopping task, and complete checkout through the payment infrastructure. Yet they do not by themselves determine how the agent should act after observing the current price. We study this buying-policy problem for an automated buyer facing uncertainty about future price changes.


We focus on a basic delegated-purchase decision: when to buy a single item within a finite shopping window. This decision arises naturally in e-commerce. A consumer may need to make a purchase before a trip, a gift occasion, or a sale expiration date. During the shopping window, the posted price can fluctuate. Thus, the agent’s problem is to decide whether the current price is sufficiently attractive relative to the option value of waiting. Figure~\ref{fig:intro_speaker_example} illustrates the type of shopping instance that motivates our problem. In the example, the consumer asks for a speaker, states a budget and a deadline in natural language, and the agent observes the speaker's product page together with its price history.

\begin{figure}[!htbp]
\centering
\caption{Illustrative shopping instance with product information, price history, and consumer request.}
\label{fig:intro_speaker_example}
\includegraphics[width=0.85\linewidth]{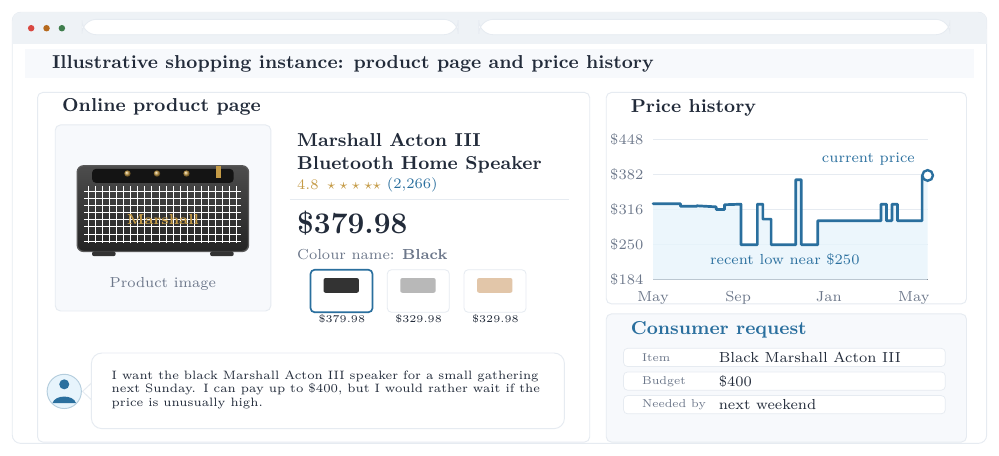}
\end{figure}

Seller-side theoretical models typically treat buyers' strategic behavior as given, with sellers responding under stylized informational assumptions. In many such models, for instance, the seller commits to a deterministic, preannounced price trajectory that buyers treat as known. A buyer-side agent operating in a real marketplace instead faces an online decision problem. It observes historical prices and incoming price adjustments during the shopping window, but does not know the future price path. The design problem is to construct a purchase policy. At each decision time, the agent maps the consumer's valuation, the current price, the remaining time, and its information or belief about future price changes into a buy-or-wait decision. Figure~\ref{fig:intro_agent_workflow} illustrates how we decompose this design problem into two layers: formulation and optimization. The formulation layer converts the shopping task and observed market information into a specific decision problem, while the optimization layer solves this problem to obtain an implementable buy-or-wait rule. 

\begin{figure}[!htbp]
\centering
\caption{Purchase-timing decision process for a buying agent. A shopping task provides the item, its value, and a finite window. The decision process maps price information, time remaining, and regime-specific information or belief into a buy-or-wait decision, with new observations updating the relevant state before subsequent decisions.}
\label{fig:intro_agent_workflow}
\includegraphics[width=0.75\linewidth,height=0.45\textheight,keepaspectratio]{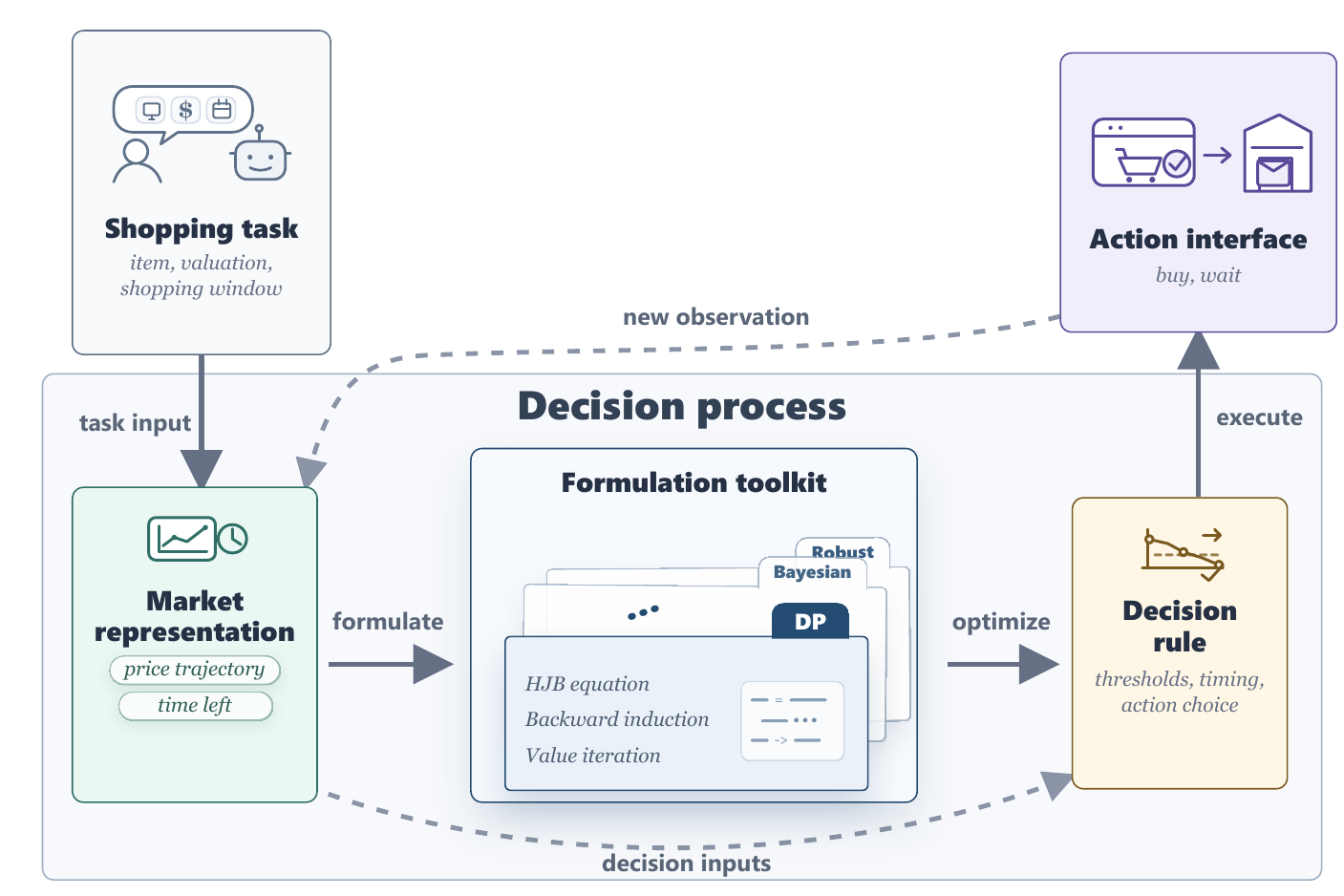}
\end{figure}

We construct the formulation layer of the agent-design problem based on the information the agent can rely on about future price changes. Rather than imposing a single model, we study three representative formulations under distinct informational assumptions. In implementation, we use a heuristic selector or LLM to decide which model to use given the available information. The first regime establishes a benchmark in which the agent believes the seller's pricing behavior is stable and can be calibrated using historical data. Specifically, the agent trusts that price adjustments arrive at a constant rate and that each post-adjustment price is drawn from a stationary distribution fully known in advance. In this environment, the optimal buying rule is a dynamic threshold: the agent purchases when the current price falls to or below a threshold that depends on the remaining time. We show that the purchase threshold is discontinuous at the deadline: immediately before the deadline, it can lie below the consumer's valuation, whereas at the deadline, the agent purchases whenever the posted price does not exceed that valuation. The threshold characterization also identifies the policy's operational levers: longer windows and more frequent adjustments make the agent more selective, while less attractive price-adjustment distributions and higher consumer valuations make earlier purchases more likely. A simple extension in \Cref{subsec:trend_stationary} allows the post-adjustment price distribution to shift along a known deterministic trend. The optimal policy remains a threshold rule, with the trend raising or lowering the purchase threshold through the expected future price adjustments.

The second regime considers settings in which the agent still anticipates random price adjustments at a constant arrival rate and believes the prices are drawn from a stable but unknown distribution. Facing this uncertainty, the agent begins with a prior using historical data and updates its belief as prices change during the shopping window. The optimal rule remains threshold-based, but the threshold now depends on the posterior belief and the time left until the deadline. Each price adjustment, therefore, plays two roles: it creates a new purchase opportunity and signals the seller's latent pricing regime. A high observed price reduces the current surplus and makes future price opportunities appear less favorable, thereby raising the subsequent purchase threshold through belief updating. Beyond the form of the rule, the Bayesian regime asks when distributional uncertainty materially changes the buying decision. Very short windows leave little opportunity for price adjustments and belief updating, and long windows give the agent enough time to learn the latent regime. 
The difficult cases are intermediate windows, where the shopping window is too short for the agent to identify the underlying price-adjustment distribution, yet long enough for stochastic price fluctuations to noticeably affect the optimal surplus.

The third regime considers settings in which the agent cannot rely on a stable probabilistic model of price changes. Instead, the agent knows only bounds on feasible prices, which can be observed from historical data, and seeks worst-case protection. The worst-case price path resembles a flash sale: the price can decline gradually, making delay attractive, and then jump back to a high level when the sale ends at an unknown time. This tension points to the use of randomized threshold policies as a robustly optimal response. The agent draws a purchase threshold at the beginning of the window and makes a purchase the first time the observed price falls to or below the threshold, with the threshold distribution calibrated to protect against adverse paths. We characterize optimal guarantees for both the competitive ratio and minimax regret. Minimax regret is useful when the initial price yields no positive surplus, because the competitive ratio collapses to zero in this case. 

Table~\ref{tab:model-comparison} summarizes the three information regimes. Throughout the paper, we refer to the stopping rules induced by our three formulations as operations research (OR) policies. This terminology separates these OR policies from purchasing policies that may also generate purchasing decisions but are not derived from solving an explicit OR model.

\begin{table}[!htbp]
\centering
\caption{Information Regimes and Purchase Policies}
\label{tab:model-comparison}
\small
\begin{tabularx}{0.96\textwidth}{@{}>{\raggedright\arraybackslash}p{0.29\textwidth}>{\raggedright\arraybackslash}X>{\raggedright\arraybackslash}p{0.22\textwidth}@{}}
\toprule
Model & Information regime & Purchase policy \\
\midrule
Stationary & \makecell[l]{Known adjustment intensity;\\ known price-adjustment distribution} & Time-dependent threshold \\
\midrule
Bayesian & \makecell[l]{Known adjustment intensity;\\ unknown price-adjustment distribution} & Belief-dependent threshold \\
\midrule
Robust & \makecell[l]{Arbitrary adjustments;\\ price bounds only} & Randomized threshold \\
\bottomrule
\end{tabularx}
\end{table}

Finally, we evaluate our stopping rules using tracked Amazon product price histories covering 367 items and 48,933 time-stamped posted-price observations. The evaluation compares several classes of methods. Fixed OR policies apply a given stopping rule based on our models with parameters estimated from pre-window price history. A heuristic selection method, referred to as the Heuristic Selector policy, chooses among these OR policies using diagnostic screens computed from the pre-window price history and then applies the selected OR policy to make the stopping decision. The LLM-OR agent uses the language model to select the policy and primitive inputs, while the selected OR policy makes the stopping decision. The remaining methods do not use OR policies. The LLM-baseline specifies whether to buy the item now or at a target purchase time, and this query is repeated after each price adjustment detected by an automated monitoring program. Other simple baselines include purchase policies that buy the item immediately or only at the deadline, and statistical-threshold rules that make a purchase before the deadline if the price falls to or below the pre-window minimum, mean, or median, with terminal purchases when the price is no more than the consumer's valuation. The evaluation shows that stationary and Bayesian policies perform favorably relative to simple baselines in terms of mean normalized consumer surplus,  while the robust policy performs best at the 10th percentile of the normalized surplus distribution. Moreover, language models are more effective at selecting, calibrating, and explaining our OR policies than at making purchase-timing decisions directly.

\section{Literature Review}

First, our paper relates to the literature on digital marketplaces and automated buying agents. Early work on agent-mediated e-commerce, such as \cite{maes1999agents} and \cite{he2003agent}, envisioned software agents that represent buyers and sellers, automate stages of the consumer buying process, and negotiate subject to user-specified constraints such as reservation prices and deadlines. Digital marketplaces and comparison shopping technologies reduce buyer search costs \citep{bakos1997reducing}, affect price competition, and reshape the information environment in which consumers compare products, sellers, and prices \citep{brynjolfsson2000frictionless,kephart2002shopbot}. Recent language models and autonomous agent systems extend the idea of delegated online shopping from information retrieval toward agents that can navigate interfaces \citep{yao2022webshop}, reason over instructions, and take actions on behalf of users \citep{yao2023react}. \citet{weber2020algorithms} analyze business models for highly autonomous consumer buying agents, emphasizing how delegation can reshape the roles of consumers, platforms, and intermediaries. \citet{allouah2025aiagentbuying} evaluate AI shopping agents in controlled e-commerce environments and document choice homogeneity and strong model dependence in product selection. Complementary work studies the preference-acquisition side of agentic shopping: \citet{cao2026solicit} study a problem of first soliciting consumer preferences through conversation and then recommending products based on learned preferences, while \citet{kumar2026conversational} examine how much a conversational recommender should converse when elicitation improves matching but imposes communication costs. We build on this stream by focusing on a consumer-side operational decision that arises after observing the price path: the buyer-side agent's intertemporal purchase policy.

Our paper also belongs to the research stream on forward-looking consumers and dynamic purchase behavior. In marketing and economics, consumers may time purchases in response to sales and inventory considerations \citep{hendel2006measuring}, as well as current and future prices and product availability \citep{nair2007intertemporal}. \citet{assuncao1993rational} derive a rational purchase and consumption policy under uncertain future promotions, showing how price expectations and inventory considerations shape intertemporal purchasing. In operations management models of strategic consumers, sellers anticipate that buyers may purchase the item immediately, wait for lower prices, or exit the market. \citet{su2007intertemporal} provides a formulation of seller pricing with strategic consumers, showing how heterogeneity in consumer valuations and patience shapes the seller's optimal price path. Subsequent research examines more specific seller-side levers such as announced discount commitments \citep{aviv2008optimal}, capacity rationing \citep{liu2008strategic}, and quick response \citep{cachon2009purchasing}. More recent work emphasizes uncertainty and learning in these dynamic pricing environments. \citet{moon2017randomized} study randomized markdowns in e-commerce environments with heterogeneous consumer costs of monitoring prices and availability. A learning-oriented counterpart is \citet{birge2025markdown}, who analyze markdown policies when the seller learns demand while facing forward-looking customers. Other seller-side studies examine markdown pricing with an unknown fraction of strategic consumers \citep{mersereau2012markdown} and randomized promotions as an intertemporal price-discrimination mechanism \citep{chen2023randomized}. We share with this literature the premise that consumers are forward-looking, but shift the focus from seller-side policies to consumer-side decision policies. 

Our modeling of posted prices builds on the literature on price stickiness and dynamic price adjustment. Macroeconomic models of sticky prices \citep{golosov2007menu} and empirical work using retail price microdata \citep{nakamura2008five, cavallo2018scraped} have documented that posted prices often remain fixed for a while before being adjusted. Online retailing also features algorithmic repricing and heterogeneous price-adjustment behavior among sellers \citep{chen2016empirical}. Related dynamic-pricing settings study pricing policies with limited opportunities for adjustment \citep{chen2016realtime} and the welfare effects of price adjustments in airline markets \citep{williams2022dynamic}. In the stationary and Bayesian information regimes, our modeling approach is closest in spirit to \citet{calvo1983staggered}: prices remain fixed until random adjustment opportunities arise, after which a new posted price is generated. We use this representation from the buyer's perspective as an exogenous price environment rather than as a structural model of seller behavior. 

Methodologically, our OR policies are close to the literature on optimal stopping. One classical branch studies timing decisions when the payoff-relevant state follows a Markov process; the decision maker stops when the process crosses an endogenous threshold. This branch underlies real-options and investment-timing models \citep{mcdonald1986value, peskir2006optimal}. A second branch examines sequential search using reservation rules to evaluate randomly arriving alternatives, commonly applied to job search \citep{mccall1970economics, lippman1976economicsii}, price search \citep{rothschild1974searching}, and asset selling \citep{rosenfield1983optimal}. Our OR policies adopt this threshold logic in a buyer-side posted-price setting. Our focus is on how the purchase threshold depends on the agent’s information, with its form varying according to the agent’s model of future price changes and the resulting decision problem.

Finally, the paper relates to learning and robust online decision-making. Our Bayesian policy relates to learning in dynamic pricing. 
Much of this literature studies how sellers set prices while learning an unknown demand function or demand model \citep{besbes2009dynamic, keskin2014dynamic}. Bayesian formulations model demand uncertainty through priors over demand environments and update beliefs as sales observations arrive \citep{harrison2012bayesian}. The literature also considers richer intertemporal features, including discounted objectives \citep{feng2024dynamic} and cyclic pricing with patient customers \citep{zhang2022online}. For a broader review of dynamic pricing and learning, see \citet{denboer2015dynamic}.
Our Bayesian model instead studies passive learning by a buying agent who observes market-generated price adjustments. Our robust policy relates to online algorithms and competitive analysis \citep{borodin1998online}, in which policies are evaluated on uncertain or adversarial input sequences \citep{karp1990optimal,mehta2007adwords}. In the context of online algorithms, \citet{elyaniv2001optimal} is methodologically close to this paper. They study a one-way trading problem in which a trader observes prices sequentially and must act without knowing future realizations. This framework is close to our robust setting, where the agent relies on price bounds rather than a full probabilistic model. 
\section{Purchase-Timing Environment}\label{sec:environment}
Before introducing particular models of price evolution, we specify the common operating environment for the buying agent. The purchase-timing problem takes the item, valuation, and deadline as exogenous task inputs. The models in \Cref{sec:stationary,sec:bayesian,sec:robust} differ in the information the agent can use about future price changes, but share the same operating interface.

The agent has monitoring and transaction capabilities. It observes the posted price process sequentially throughout the shopping window and, at any time before the deadline, can either purchase at the currently posted price or continue monitoring. We abstract from monitoring and checkout delays, so a purchase decision is implemented at the price observed when the decision is made. A purchase terminates the task, and the agent's action does not affect the future price process. Thus, an implementable buying policy must be adapted to the price history, elapsed time, and the information available under the relevant regime.

\subsubsection*{Decision primitives.} 
The shopping window has a length \(T\). We use elapsed time \(t\in[0,T]\), where \(t=0\) is the beginning of the shopping window and \(t=T\) is the deadline. The consumer has valuation \(v\) for the item. If the agent purchases at price \(p\), the consumer receives surplus $(v-p)^+$. If the agent exits the shopping window without making a purchase, the payoff is zero. This formulation isolates the timing decision. The agent's task is to decide when the current price is good enough to buy.

At any decision point, the relevant state consists of the current price, the time remaining until the deadline, and the regime-specific information or belief about future prices. The current price determines the payoff from buying immediately, the remaining time determines the opportunity to wait, and the information regime determines how the agent evaluates that waiting option. A buying policy maps this state into a buy-or-wait decision.

\subsubsection*{Information regimes.} 
We organize the design problem around three informational regimes. Each regime specifies a different information structure that the agent can use to model future price changes. First, in the stationary price-adjustment regime, the agent knows both the adjustment intensity and the price-adjustment distribution. This regime serves as the most structured benchmark for computing the expected continuation value of waiting. 

Second, in the Bayesian learning regime, the agent retains the assumption of the Poisson arrival for price adjustments but does not know the price-adjustment distribution. The agent begins with a prior over possible price environments and updates its belief as price changes are observed. 

Third, in the robust regime, the agent does not rely on a probabilistic model of price evolution. Instead, it knows only bounds on feasible prices and seeks protection against worst-case paths. This regime is appropriate when historical data are sparse, price dynamics are unstable, or model calibration is not credible. 

These regimes form a hierarchy of informational assumptions. The stationary regime represents the most structured case in which the agent can evaluate future opportunities through a specified stochastic model. The Bayesian model relaxes this structure by treating the price-adjustment distribution as unknown, so new price observations update the agent's belief about the price environment. The robust model imposes the least structure and replaces expected optimality with worst-case performance guarantees.


\section{Buying in a Stationary Price-Adjustment Environment}\label{sec:stationary}
 
This section studies a stationary stochastic environment for seller price adjustments. The agent observes the seller's current posted price, while future prices remain stochastic and outside the agent's control. The current price remains available until the seller revises it, and the agent must decide whether to purchase now or continue monitoring for a more favorable price before the deadline. This OR model reduces the monitoring problem to computing a threshold policy once the adjustment intensity, price-adjustment distribution, valuation, and deadline are specified.

\subsection{Stationary Price-Adjustment Model}
The central primitive is a stationary price process. The posted price remains fixed until an adjustment opportunity arrives. Adjustments arrive according to a homogeneous Poisson process with rate $\lambda$, and each adjustment redraws the posted price from a stationary distribution $H$. This follows the Calvo formulation of price stickiness, where price adjustment opportunities arise randomly and are memoryless \citep{calvo1983staggered}.

Let $\mathcal P\subseteq\mathbb R_+$ denote the feasible price space equipped with its Borel sigma-field $\mathcal B(\mathcal P)$. Write \(P_t\) for the posted price at elapsed time \(t\), so \(P_0\) is the window-start price and \(P_T\) is the deadline price.

\begin{assumption}\label{asm:model}
The price path is generated by the following primitives.
\begin{enumerate}[label=(\roman*)]
  \item\label{asm:poisson}
\textbf{Poisson adjustment.} Adjustment opportunities arrive over elapsed time according to a homogeneous Poisson process $\{N_r:r\ge0\}$ with rate $\lambda\in(0,\infty)$. 

  \item\label{asm:kernel}
\textbf{Stationary price-adjustment distribution.} At each adjustment time, the post-adjustment price is drawn independently from a time-invariant distribution $H$ on $\mathcal P$. 

  \item\label{asm:path}
\textbf{Price stickiness.} Given a window-start price $P_0$, elapsed adjustment times $\{r_n\}_{n\ge1}$, and post-adjustment prices $\{Y_n\}_{n\ge1}$, the posted price is piecewise constant:
    \[
        P_t =
        \begin{cases}
            P_0, & 0\le t<r_1,\\
            Y_n, & r_n\le t<r_{n+1},\quad n\ge1.
        \end{cases}
    \]
\end{enumerate}
\end{assumption}

The distribution $H$ is the agent's predictive distribution for the next post-adjustment price; throughout the paper, we refer to $H$ as the price-adjustment distribution. This framework allows the agent to aggregate heterogeneous adjustment triggers, such as competitor promotions and stochastic demand shocks, without requiring the agent to identify the specific drivers behind each price change. Formally, we represent each unique category of adjustment triggers as an event $e\in E$ and assume events arrive according to a memoryless point process with relative intensities $\mu(\dd e)$. Conditioned on the occurrence of event $e$, the seller's response is captured by an event-specific price-adjustment distribution $H_e$. The composite predictive distribution is $H(A)=\int_E H_e(A)\,\mu(\dd e)$, $A\in\mathcal B(\mathcal P)$. Since the agent's stopping problem depends on future adjustments only through the aggregate arrival rate and the predictive price-adjustment distribution, we adopt $(\lambda, H)$ as the environment's primitives. Throughout this section, we restrict our focus to nondegenerate buying environments in which the price-adjustment distribution can generate a price no higher than the consumer's valuation, i.e., $\inf\supp(H)\le v$. Otherwise, all future post-adjustment prices exceed the valuation, so waiting is never beneficial, and the problem is degenerate.

\subsection{Buying-Agent Problem}

At elapsed time $t\in[0,T]$, suppose the current posted price is $P_t=p$. Let $\mathcal{T}_t$ denote the set of admissible stopping times taking values in $[t,T]$. Since both the price-adjustment arrivals and the price-adjustment distribution are stationary, elapsed time affects the continuation problem only through the time left until the deadline, $T-t$. Thus, the agent's value function can be written as 
\begin{equation}\label{eq:value}
  V(T-t,p)
  :=
  \sup_{\tau\in\mathcal{T}_t}
  \mathbb{E}\left[(v-P_{\tau})^+\,\Big|\,P_t=p\right],
\end{equation}
where $P_{\tau}$ denotes the posted price at the purchase time. If the agent does not stop before the deadline, the terminal payoff is $(v-P_{T})^+$.

\subsubsection{Optimal Buying Policy.}\label{sec:optimal}

Denote the time-to-go $T-t$ by $s$ for convenience. The value function~\eqref{eq:value} satisfies the finite-horizon optimal stopping variational inequality
\begin{equation}\label{eq:VI}
  \max\left\{
    (v-p)^+ - V(s,p),\;\;
    -\partial_s V(s,p) + \lambda\,\mathbb{E}_H\left[V(s,P)\right] - \lambda\,V(s,p)
  \right\}
  = 0,
  \qquad s\in(0,T],
\end{equation}
with terminal condition $V(0,p)=(v-p)^+$. The stopping set is $\mathcal{S}:=\bigl\{(s,p): V(s,p)=(v-p)^+\bigr\}$ and the continuation set is $\mathcal{W}:=\bigl\{(s,p): V(s,p)>(v-p)^+\bigr\}$.

\begin{lemma}\label{lemma:mono_V}
Under \Cref{asm:model}, the value function $V(s,p)$ is non-increasing and 1-Lipschitz in $p$ and non-decreasing and continuous in $s$.
\end{lemma}

The price monotonicity and Lipschitz property imply that the relative option value of waiting, \(W(s,p):=V(s,p)-(v-p)\), is non-decreasing in \(p\) on \([0,v]\). Hence, for each fixed \(s\), there exists a threshold $b(s)$ such that the stopping set takes the form \([0,b(s)]\).

\begin{proposition}[{\sc Optimal Purchase-Threshold Policy}]\label{prop:threshold}
Under \Cref{asm:model}, the following hold.
  \begin{enumerate}[label=(\roman*)]
    \item\label{prop:threshold_form}
There exists an optimal threshold function $b:[0,T]\to[0,v]$ such that, at each elapsed time $t\leq T$, the agent makes a purchase whenever the current price satisfies $P_t\le b(T-t)$ and waits otherwise.
    \item\label{prop:ODE}
Let $m(s):=\mathbb{E}_{P\sim H}\left[V(s,P)\right]$. For every $s>0$, the threshold satisfies $b(s)=v-m(s)$. Consequently, $b$ solves the ODE
      \begin{equation}\label{eq:ODE}
        b'(s)
        = -\lambda\,\mathbb{E}_{P\sim H}\left[(b(s)-P)^+\right],
      \end{equation}
for almost every $s>0$, with initial condition $\displaystyle\lim_{s\downarrow 0}b(s)=\mathbb{E}_H\left[\min\{v,P\}\right]$.
    \item\label{prop:mono_convex}
The purchase threshold $b(s)$ is non-increasing and convex on $(0,T]$.
  \end{enumerate}
\end{proposition}

At the deadline, the agent purchases if and only if $P_T\le v$, so the terminal threshold is $v$. Immediately before the deadline, however, the dynamic threshold has right limit $b(0^+)=\mathbb{E}_H[\min\{v,P\}]$. Hence, the optimal purchasing threshold is discontinuous at the deadline whenever $H$ places positive mass below~$v$. The jump arises because waiting briefly does not immediately forfeit the current price: if no adjustment occurs, the agent can still purchase at the unchanged price; if an adjustment occurs, the agent observes a new post-adjustment price draw. Thus, even an arbitrarily short positive horizon gives the agent a one-sided opportunity to obtain a lower price, lowering the pre-deadline threshold from $v$ to $b(0^+)$.

The limiting thresholds also clarify how simple static rules can be interpreted. A mean-price rule, which purchases when the current price is no larger than the average post-adjustment price $\mathbb E_H[P]$, reflects a short remaining time logic: it coincides with the pre-deadline limit when the valuation is high enough, $v\ge \sup\supp(H)$, in which case $b(0^+)=\mathbb E_H[P]$. At the other extreme, the ODE implies $\lim_{s\to\infty}b(s)=\inf\supp(H)$, so the lower-support rule captures the long remaining time logic: with many future adjustment opportunities, the agent can afford to wait for unusually favorable prices.


\subsubsection{Comparative Statics and Illustration.}

Having characterized the optimal policy for a fixed buying environment, we next ask how it changes when the primitives vary. This comparative-statics question is important for both implementation and interpretation: the arrival rate, the price-adjustment distribution, and the consumer valuation each enter the stopping rule through a distinct economic channel. To compare price-adjustment distributions without imposing a parametric family, we order them by first-order stochastic dominance (FSD, also known as the usual stochastic order): we write $H_1\preceq_{\mathrm{FSD}}H_2$ if $H_2$ first-order stochastically dominates $H_1$ \citep[cf.][]{hadar1969rules,bawa1975optimal,shaked2007stochastic}.


This order provides a nonparametric way to compare future price opportunities. A stochastically larger price-adjustment distribution makes low future prices less likely, while changes in $\lambda$ and $v$ alter the frequency of opportunities and the payoff from purchase. The next result demonstrates how these three primitives impact the optimal purchase threshold.

\begin{proposition}[{\sc Stationary Threshold Monotonicity}]\label{prop:stationary_monotonicity}
For every $s>0$, the optimal threshold $b(s; \lambda, H, v)$ satisfies the following monotonicity properties.
\begin{enumerate}[label=(\roman*)]
    \item $b(s; \lambda, H, v)$ is non-increasing in $\lambda$: if $\lambda_1 \le \lambda_2$, then $b(s; \lambda_2, H, v) \le b(s; \lambda_1, H, v)$.
    \item $b(s; \lambda, H, v)$ is non-decreasing (in the sense of $\preceq_{\mathrm{FSD}}$) in $H$: if $H_1 \preceq_{\mathrm{FSD}} H_2$, then $b(s; \lambda, H_1, v) \le b(s; \lambda, H_2, v)$.
    \item $b(s; \lambda, H, v)$ is non-decreasing in $v$: if $v_1 \le v_2$, then $b(s; \lambda, H, v_1) \le b(s; \lambda, H, v_2)$.
\end{enumerate}
\end{proposition}

The comparative statics highlight three distinct effects. First, a larger $\lambda$ makes the next price adjustment arrive sooner, reducing the cost of waiting. The agent, therefore, becomes more selective and stops only at a lower posted price. Second, a stochastically higher price-adjustment distribution $H$ makes waiting less attractive. The stopping threshold, therefore, increases. Third, a higher consumer valuation $v$ raises the stopping threshold. A larger $v$ weakly increases the value of future opportunities, but this effect is realized only in future states where the price is low enough to generate a positive surplus. For a current purchase opportunity that already yields a nonnegative surplus, a higher $v$ directly raises the surplus from stopping one-for-one, so the agent is willing to accept a higher price.

\begin{figure}[!htbp]
  \centering
  \caption{Time-dependent purchase threshold $b(T-t)$ as a function of the time $t$, together with two illustrative sample paths of the posted price.}
  \label{fig:threshold}
  \includegraphics[width=0.7\linewidth]{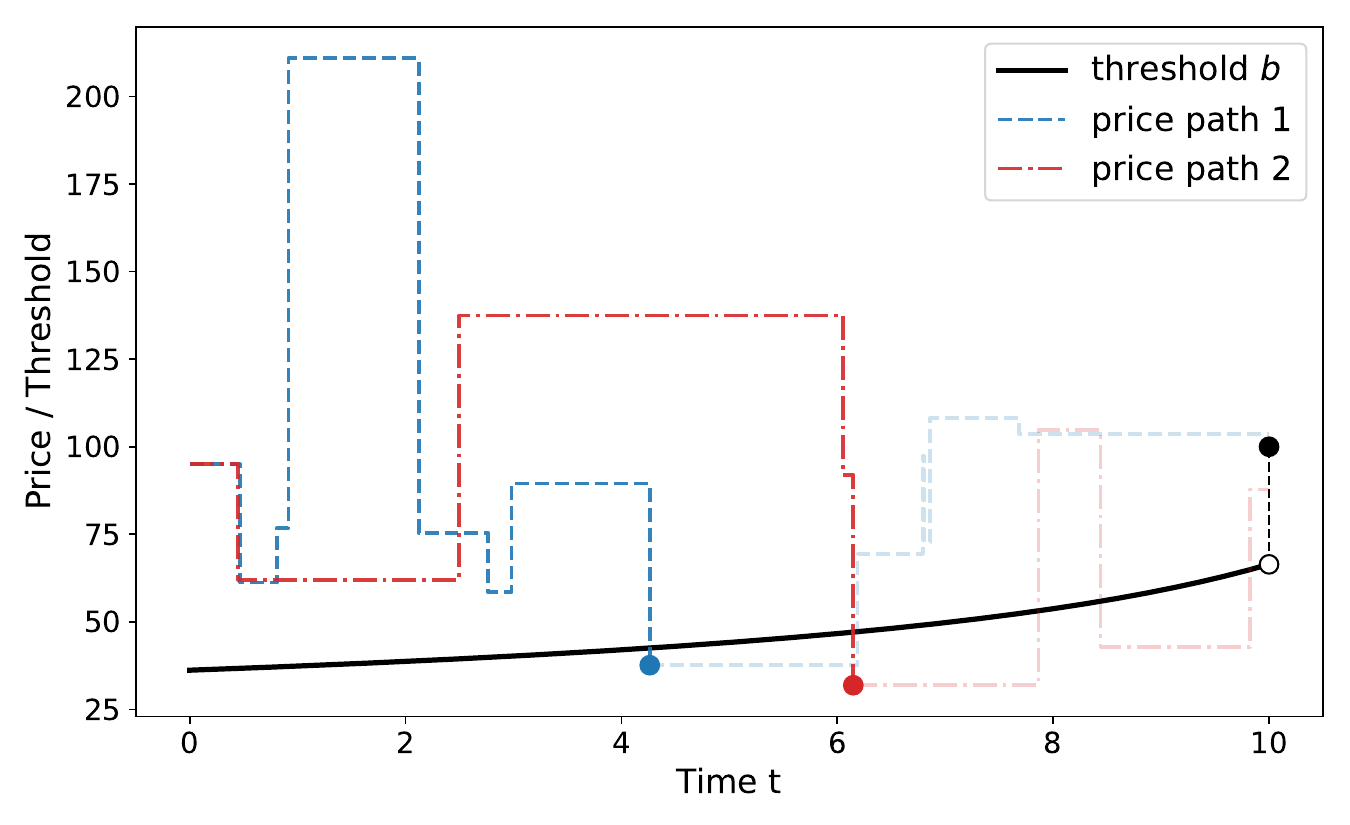}
\end{figure}

The following example illustrates the threshold dynamics in closed form.
\begin{example}[Uniform distribution]\label{example_1}
To illustrate the structure of the optimal policy, consider the case in which the effective price-adjustment distribution is uniform: $H=\mathrm{Unif}[p_-,p_+]$ with $0\le p_-<p_+$ and $p_-<v$. For a positive time-to-go $s$, on the region where $b(s)\in[p_-,p_+]$, the ODE~\eqref{eq:ODE} becomes
\[
  b'(s)
  = -\frac{\lambda}{p_+-p_-}\int_{p_-}^{b(s)}(b(s)-p)\,\dd p
  = -\frac{\lambda\,(b(s)-p_-)^2}{2(p_+-p_-)}.
\]
This is a separable ODE with a solution
\[
b(s)=p_-+\frac{2(p_+-p_-)\bigl(b(0^+)-p_-\bigr)}{2(p_+-p_-)+\lambda s\,\bigl(b(0^+)-p_-\bigr)},
\]
where $b(0^+)=\mathbb{E}_H[\min\{v,P\}]$. Since $b(0^+)\le p_+$ and $b$ is non-increasing with $\lim_{s\to\infty}b(s)=p_-$, the threshold remains in $[p_-,b(0^+)]\subseteq[p_-,p_+]$ for all $s>0$, so the closed form above applies globally on the positive-horizon branch. For $v\ge p_+$, this simplifies to $b(0^+)=(p_-+p_+)/2$ and
\[
  b(s)
  = p_- + \frac{2(p_+-p_-)}{4+\lambda s}.
\]
The right limit of the purchase threshold at short positive horizons is the midpoint of the support, and the threshold decays hyperbolically toward the lower endpoint~$p_-$ as the time-to-go $s$ grows. This closed-form solution illustrates the general properties established in \Cref{prop:threshold}: for $s>0$, $b$ is strictly decreasing and strictly convex, with $\lim_{s\to\infty}b(s)=p_-=\inf\operatorname{supp}(H)$.
\end{example}

\subsection{Extension: Time-Varying Price-Adjustment Distribution}
\label{subsec:trend_stationary}

The stationary benchmark above assumes that the price-adjustment distribution is time invariant. In some applications, however, part of the seller's price movement over the shopping window is predictable. Examples include life-cycle markdown schedules, pre-stockout price run-ups, and calendar-based promotional campaigns. A simple extension is therefore to let the price-adjustment distribution shift with elapsed time through a known deterministic trend.

In the trend-aware extension of the stationary benchmark, whenever a price adjustment occurs at elapsed time $t$, the post-adjustment price is \(g(t)+X\), $X\sim H_0$, where $g:[0,T]\to\mathbb{R}$ is a given deterministic trend and $H_0$ is a baseline price-adjustment distribution. Thus the effective price-adjustment distribution becomes time dependent: \(H_t(A)=H_0\bigl(\{x:g(t)+x\in A\}\bigr)\), $A\in\mathcal{B}(\mathcal{P})$. For the dynamic program, write $s:=T-t$. The buying problem remains an optimal stopping problem, but now future adjustment opportunities depend on the current elapsed time, as reflected in $H_{T-s}$.

The dynamic programming logic of \Cref{prop:threshold} still holds. Denoting by $V^{\mathrm{tr}}(s,p)$ the value function under the trend-aware dynamics, we have the variational inequality as follows:
\begin{equation}\label{eq:VI_trend}
\max\!\left\{
\begin{aligned}
& (v-p)^+ - V^{\mathrm{tr}}(s,p), \\
& -\partial_s V^{\mathrm{tr}}(s,p) + \lambda\,\mathbb{E}_{X\sim H_0}\!\bigl[V^{\mathrm{tr}}(s,\, g(T-s)+X)\bigr] - \lambda V^{\mathrm{tr}}(s,p)
\end{aligned}
\right\}=0,
\quad s\in(0,T],
\end{equation}
with terminal condition $V^{\mathrm{tr}}(0,p)=(v-p)^+$. Specializing to $g\equiv 0$ recovers~\eqref{eq:VI}. 

The deterministic trend changes the price-adjustment distribution over time but does not make those post-adjustment prices depend on the current posted price \(p\). Hence the monotonicity argument of \Cref{lemma:mono_V} carries over: for each fixed time-to-go \(s\), the value function \(V^{\mathrm{tr}}(s,p)\) is non-increasing in \(p\). Therefore, the optimal purchase policy is represented by a trend-aware threshold $b^{\mathrm{tr}}(s)$. At elapsed time $t\leq T$, the agent purchases whenever $P_t\le b^{\mathrm{tr}}(T-t)$.

Operationally, the trend component shifts the stopping threshold through the expected level of future price adjustments. A downward trend makes waiting more attractive by lowering future price draws; an upward trend has the opposite effect. \Cref{fig:trend_aware_multi} illustrates how the purchase threshold changes across several deterministic trend environments, together with sample price paths under each specification.

\begin{figure}[!htbp]
  \centering
  \caption{Trend-aware purchase thresholds under deterministic linear trend specifications with slope $\beta$, together with two illustrative sample price paths in each case.}
  \label{fig:trend_aware_multi}
  \includegraphics[width=\linewidth]{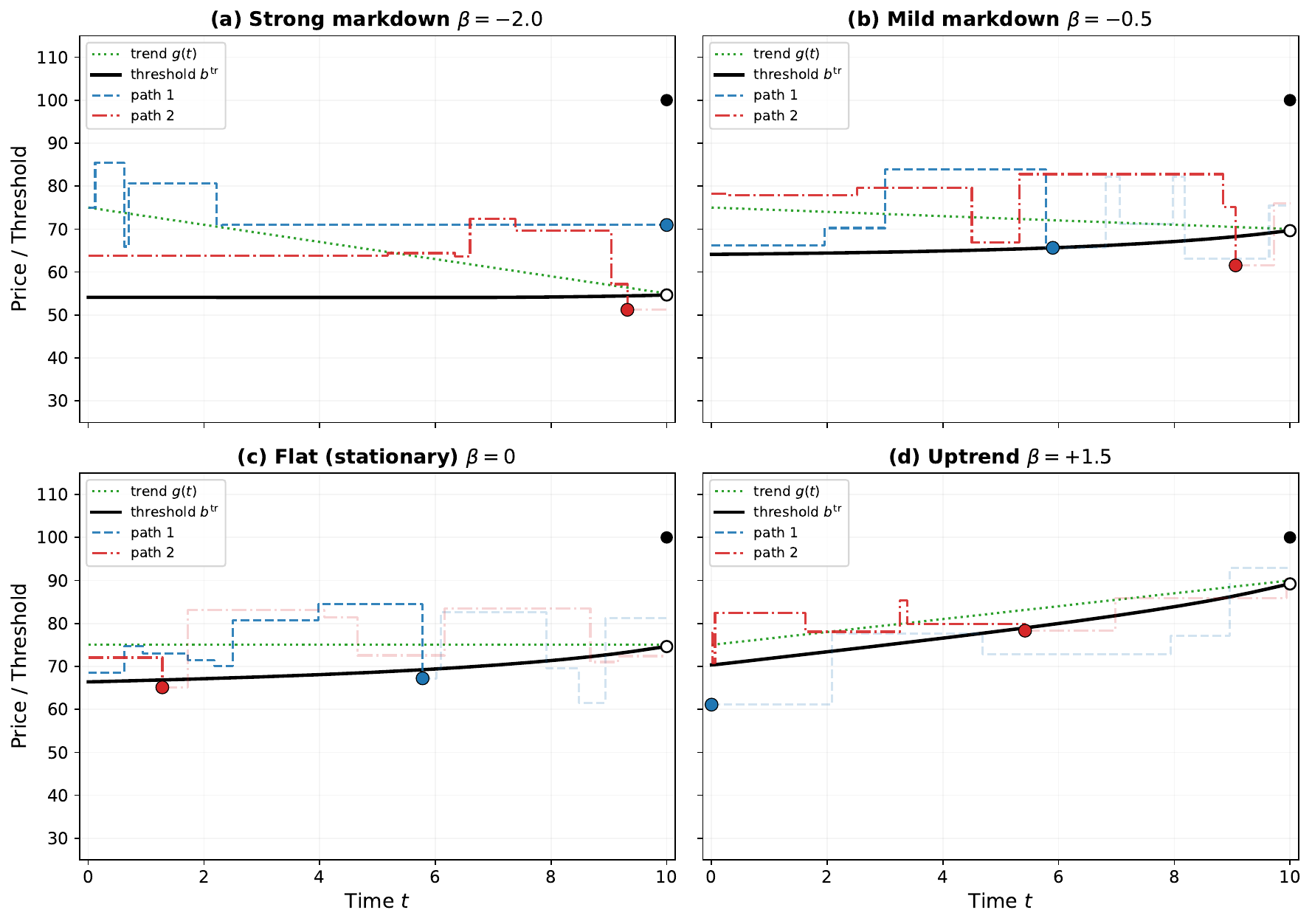}
\end{figure}

\section{Buying Under Bayesian Learning}\label{sec:bayesian}

The previous section treats both the adjustment intensity $\lambda$ and the price-adjustment distribution $H$ as known. We keep the Poisson adjustment process from \Cref{asm:model}, but replace the known distribution with a latent one. The true price-adjustment distribution is $H_{\tilde\theta}$, where $\tilde\theta\in\Theta$ is fixed over the shopping window and drawn from a prior belief $\hat\pi$. Conditional on $\tilde\theta$, the price process satisfies \Cref{asm:model} with $H=H_{\tilde\theta}$. The agent observes realized post-adjustment prices and updates its posterior belief. 

\subsection{Bayesian Learning of the Price-Adjustment Distribution}

We model uncertainty about future price changes through a family of candidate distributions $\{H_\theta:\theta\in\Theta\}$, where $(\Theta,\mathcal B(\Theta))$ is a measurable parameter space. A latent pricing environment $\tilde{\theta}\in\Theta$ is drawn once from a prior belief $\hat\pi\in\mathcal P(\Theta)$ and remains fixed throughout the shopping window. Conditional on $\tilde{\theta}=\theta$, adjustment times still arrive according to the Poisson process with rate $\lambda$, and each post-adjustment price is drawn i.i.d. from $H_\theta$. The latent environment can encode one or several persistent features of the post-adjustment price distribution, such as markdown aggressiveness or dispersion in price changes. The buying agent does not observe $\tilde{\theta}$ directly. Instead, it updates beliefs from the realized sequence of posted prices observed at adjustment times.

In empirical implementations, the prior $\hat\pi$ can be built from contextual predictors, including comparable-item prices, seller characteristics, and platform-side covariates. We treat the window-start price $P_0$ as the first in-horizon Bayesian signal rather than as part of prior learning data. This convention yields a unified formulation for products with and without historical price records. Let 
\[
  \pi_t(\cdot):=\mathbb P(\theta\in\cdot\mid \text{price history observed up to elapsed time } t)
\]
be the posterior belief at elapsed time $t$. Because the adjustment intensity is treated as known, learning in this section comes entirely from price observations. Hence, the posterior remains constant between consecutive adjustment times and changes only when a new posted price is observed.

When each distribution $H_\theta$ admits a density $h_\theta$ with respect to a reference measure, Bayes' rule implies that, upon observing a post-adjustment price $p$, the posterior updates according to
\begin{equation*}
\Phi(\pi,p)(d\vartheta)
:=
\frac{h_\vartheta(p)\pi(d\vartheta)}
{\int_\Theta h_u(p)\pi(du)}.
\end{equation*}
We apply Bayes' rule only to price histories that can arise under the model; for those histories, the denominator is positive and finite. Prices outside the predictive support are irrelevant for the dynamic program. Given a belief $\pi$, the agent's predictive distribution of the next post-adjustment price is
\begin{equation*}
  \bar{H}_\pi(\cdot)
  :=
  \int_\Theta H_\theta(\cdot)\,\pi(d\theta).
\end{equation*}
Accordingly, the agent's state is summarized by the pair $(p,\pi)$: the current posted price and the current belief about the distribution of the next post-adjustment price.

\subsection{Bayesian Buying-Agent Problem}

At time $t$, the buying agent observes the current posted price $P_t=p$ and holds a posterior belief $\pi_t=\pi$. Let $\mathcal T_t$ denote the set of admissible stopping rules taking values in $[t,T]$. The Bayesian value function can be written as
\[
  V^{B}(T-t,p,\pi)
  :=
  \sup_{\tau\in\mathcal T_t}
  \mathbb E\left[(v-P_{\tau})^+ \mid P_t=p,\ \pi_t=\pi\right].
\]
The terminal condition is \(V^B(0,p,\pi)=(v-p)^+\). Because the posterior changes only at adjustment times, the Bayesian continuation problem is Markovian in the state $(s,p,\pi)$, where $s=T-t$ denotes the time-to-go. The corresponding variational inequality is
\begin{equation*}
  \max\left\{
    \begin{aligned}
      &(v-p)^+-V^B(s,p,\pi),\\
      &-\partial_s V^B(s,p,\pi)
      +\lambda \E_{P\sim \bar{H}_\pi}\left[
        V^B\bigl(s,P,\Phi(\pi,P)\bigr)
      \right]
      -\lambda V^B(s,p,\pi)
  \end{aligned}\right\}
  =0,
  \qquad s\in(0,T].
\end{equation*}

\subsubsection{Belief-Dependent Threshold Policy.}

The next proposition shows that the optimal policy retains a threshold structure, but the threshold now depends on both the time-to-go and posterior belief.
\begin{proposition}[{\sc Belief-Dependent Purchase Threshold}]\label{prop:bayes_threshold}
Under the Bayesian learning model, for each fixed $(s,\pi)$, the mapping $p\mapsto V^B(s,p,\pi)$ is non-increasing. Moreover, there exists a belief-dependent purchase threshold $b(s,\pi)\in[0,v]$ such that an agent with the current state $(p,\pi)$ stops whenever $p\le b(s,\pi)$ and continues whenever $p>b(s,\pi)$. Equivalently, at elapsed time $t\leq T$, the optimal policy purchases whenever $P_t\leq b(T-t,\pi_t)$.
\end{proposition}

\Cref{prop:bayes_threshold} is the Bayesian analog of \Cref{prop:threshold}. Learning enlarges the state from the current price to the pair $(p,\pi)$, but conditional on the belief, the current price affects the pre-adjustment decision only through the surplus from buying now. An adjustment both changes the price and updates the posterior through $\Phi$, so an observed price adjustment can shift the threshold used for subsequent decisions. \Cref{fig:bayes_threshold_paths} illustrates this belief-dependent updating: the purchase threshold is re-evaluated as the posterior state changes, rather than following a fixed threshold path.

\begin{figure}[!htbp]
  \centering
  \caption{Two belief-dependent dynamic purchase thresholds as functions of time, together with corresponding illustrative sample paths of the posted price.}
  \label{fig:bayes_threshold_paths}
  \includegraphics[width=0.7\linewidth]{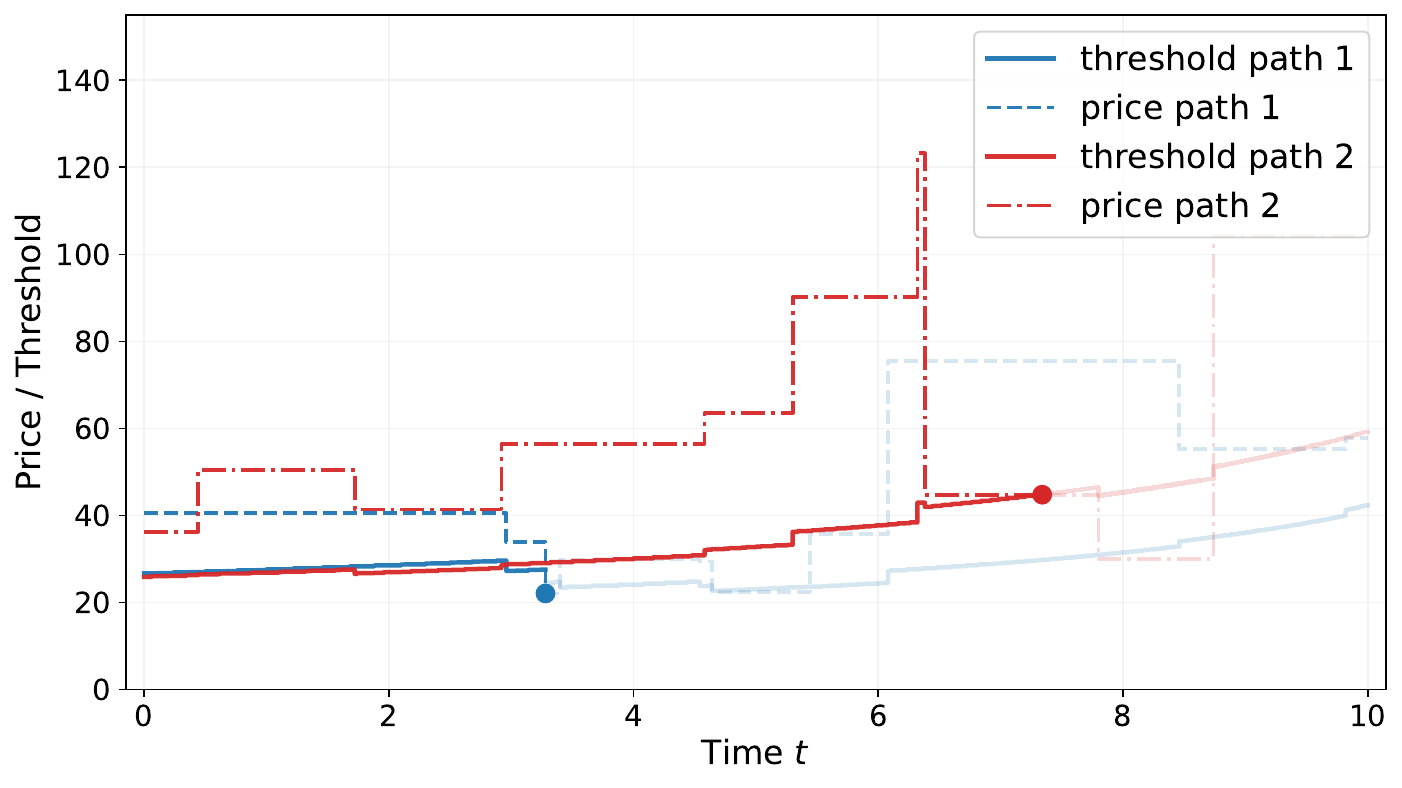}
\end{figure}

\subsubsection{Comparative Statics.}

The threshold structure also makes the comparative statics possible. Holding the current belief fixed, the arrival rate still governs how often the agent expects new opportunities to arrive, while the valuation scales the payoff from purchasing. The next result demonstrates the corresponding monotonicity of the Bayesian threshold.

\begin{proposition}[{\sc Bayesian Threshold Monotonicity}]\label{prop:bayes_param_monotonicity}
The optimal belief-dependent threshold $b(s,\pi; \lambda, v)$ satisfies the following monotonicity properties.
\begin{enumerate}[label=(\roman*)]
    \item $b(s,\pi; \lambda, v)$ is non-increasing in $\lambda$:
if $\lambda_1 \le \lambda_2$, then $b(s,\pi; \lambda_2, v) \le b(s,\pi; \lambda_1, v)$.
    \item  $b(s,\pi; \lambda, v)$ is non-decreasing in $v$: if
$v_1 \le v_2$, then $b(s,\pi; \lambda, v_1) \le b(s,\pi; \lambda, v_2)$.
\end{enumerate}
\end{proposition}

We now consider how the threshold changes as the belief becomes more pessimistic. To formalize this comparison, we need an order on posterior beliefs that is compatible with both predictive price distributions and Bayesian updating. The following monotone-information condition provides such an order for one-dimensional distributions.

\begin{definition}[Monotone Likelihood-Ratio Order]\label{def:mlr}
Let $F_1$ and $F_2$ be distributions on an ordered set $\mathcal X\subseteq\mathbb R$, and suppose they admit densities $f_1$ and $f_2$ with respect to a common dominating measure. We write $F_1\preceq_{\mathrm{MLR}}F_2$ if, for every $x_1<x_2$, \( f_{2}(x_2) f_{1}(x_1) \geq f_{2}(x_1) f_{1}(x_2) \). We use the same notation for price distributions on $\mathcal P$ and for beliefs on the ordered parameter space~$\Theta$.
\end{definition}
\begin{assumption}[{\sc Monotone Likelihood-Ratio}]\label{asm:mlrp}
The parameter space is an ordered subset $\Theta\subseteq\mathbb R$. Each $H_\theta$ admits a density $h_\theta$ with respect to a common dominating measure on $\mathcal P$. For every $\theta_1<\theta_2$, $H_{\theta_1}\preceq_{\mathrm{MLR}}H_{\theta_2}$.
\end{assumption}

\Cref{asm:mlrp} is a natural monotone-information condition. It states that higher values of the latent parameter $\theta$ shift the price distribution toward higher prices in a likelihood-ratio sense. Since MLR order is stronger than first-order stochastic dominance \citep{milgrom1981good}, a higher $\theta$ makes higher price realizations more likely. In economic terms, $\theta$ indexes the ``costliness'' of the seller's pricing regime. A higher $\theta$ corresponds to a seller that tends to post higher prices, perhaps due to stronger demand conditions, lower competitive pressure, or more conservative markdown policies.
 
The MLR structure also ensures that observed prices are informative about $\theta$ in an order-preserving way: a high observed price is a signal that $\theta$ is likely high, and vice versa. This is the informational regularity that makes the posterior behave monotonically, as the next lemma establishes.

\begin{lemma}\label{lem:belief_order}
Under \Cref{asm:mlrp}, the following properties hold:
  \begin{enumerate}[label=(\roman*)]
    \item\label{lem:bo_predictive}
If $\pi_1\preceq_{\mathrm{FSD}}\pi_2$, then $\bar{H}_{\pi_1}\preceq_{\mathrm{FSD}} \bar{H}_{\pi_2}$.
    \item\label{lem:bo_update}
For every observed price~$p$ for which both posteriors are defined and $\pi_1\preceq_{\mathrm{MLR}}\pi_2$, $\Phi(\pi_1,p)\preceq_{\mathrm{MLR}}\Phi(\pi_2,p)$.
    \item\label{lem:bo_signal}
More generally, if $\pi_1\preceq_{\mathrm{MLR}}\pi_2$ and $p_1\le p_2$, then
\(
  \Phi(\pi_1,p_1)\preceq_{\mathrm{MLR}}\Phi(\pi_2,p_2)
\)
whenever both posteriors are defined. In particular, for every prior~$\pi$ and prices $p_1\le p_2$, $\Phi(\pi,p_1)\preceq_{\mathrm{MLR}}\Phi(\pi,p_2)$ whenever both posteriors are defined.
  \end{enumerate}
\end{lemma}

\Cref{lem:belief_order} establishes the monotonicity properties used in the Bayesian comparative statics. Part~\ref{lem:bo_predictive} links beliefs to predictive price distributions: if $\pi_2$ places relatively more weight on high-$\theta$ regimes than $\pi_1$, then the predictive distribution under $\pi_2$ first-order stochastically dominates that under $\pi_1$. Part~\ref{lem:bo_update} shows that this belief ordering is preserved under Bayesian updating, so a more pessimistic prior leads to a more pessimistic posterior after the same price observation. Part~\ref{lem:bo_signal} adds a cross-monotonicity property: a more pessimistic prior combined with a higher observed price leads to a more pessimistic posterior. Together, these properties ensure that the ordering on $\mathcal{P}(\Theta)$ is compatible with the dynamic program, which enables the comparative statics in \Cref{prop:bayes_belief_monotonicity}.

\begin{proposition}[{\sc Belief Monotonicity of the Bayesian Threshold}]
\label{prop:bayes_belief_monotonicity}
Under the Bayesian learning model and \Cref{asm:mlrp}, if $\pi_1\preceq_{\mathrm{MLR}}\pi_2$, then
  \begin{enumerate}[label=(\roman*)]
    \item\label{prop:bayes_belief_value}
The Bayesian value satisfies $V^B(s,p,\pi_1)\ge V^B(s,p,\pi_2)$ for all $(s,p)$.
    \item\label{prop:bayes_belief_threshold}
The belief-dependent threshold satisfies $b(s,\pi_1)\le b(s,\pi_2)$ for all $s$.
  \end{enumerate}
\end{proposition}

\Cref{prop:bayes_belief_monotonicity} identifies how posterior pessimism changes stopping incentives. In this buying problem, a pessimistic belief lowers the option value of future adjustments, making buying at the current price relatively more attractive and raising the purchase threshold. A realized price adjustment affects the buy-or-wait decision in two ways. First, the newly observed price determines the immediate surplus, which is then compared with the threshold. Second, as a signal, the same price updates the posterior belief that determines the threshold. Under the MLR condition, a higher observed price shifts the posterior belief toward higher values of \(\theta\), induces a stochastically higher predictive distribution for future prices, and thereby raises the threshold.

\subsubsection{Information Gap.}\label{sec:info_gap}

We now quantify the payoff loss due to uncertainty about which price-adjustment distribution governs future price adjustments. Let $V^\theta(s,p)$ denote the oracle value function when the true price-adjustment distribution~$H_\theta$ is known, and let $V^{\mathrm B}(s,p,\pi)$ denote the Bayesian value function under posterior belief~$\pi$. The \emph{information gap}
\begin{equation*}
  \mathcal G(T,\hat\pi)
  :=
  \E_{\tilde{\theta}\sim\hat\pi}\left[
    \E_{P\sim H_{\tilde{\theta}}}\left[
      V^{\tilde{\theta}}(T,P)-V^{\mathrm B}(T,P,\Phi(\hat\pi,P))
    \right]
  \right]
\end{equation*}
measures the expected loss from acting under posterior uncertainty rather than full knowledge.

\begin{lemma}[{\sc Short-Window Information Gap}]\label{lem:short_window_gap}
Under the Bayesian learning model, for every prior $\hat\pi$ and every $T\ge0$,
\(
  0\le \mathcal G(T,\hat\pi)
  \le
  v(1-e^{-\lambda T})
  \le
  v\lambda T
\).
\end{lemma}

The intuition behind Lemma \ref{lem:short_window_gap} is that when the shopping window is short, a price adjustment is unlikely to occur, so even an oracle with full knowledge of the price-adjustment distribution has little opportunity to improve on immediate purchase when the valuation is no less than the current price.

For long shopping windows, we impose the following regularity assumption.

\begin{assumption}\label{ass:regular_family}
The parameter space $\Theta= [\underline\theta,\bar\theta]\subset\mathbb R$ is a compact interval. Each $H_\theta$ admits a density $h_\theta$, and the following conditions hold.
\begin{enumerate}[label=(\roman*)]
\item \emph{Lower-tail support.}
There exist common constants $c>0$, $\alpha>0$, and $\varepsilon>0$, and a lower endpoint $p_L^{\theta}\geq 0$ for each $\theta$, such that $v>p_L^{\theta}$, $H_\theta([0,p_L^{\theta}))=0$, and $H_\theta([p_L^{\theta},p_L^{\theta}{+}x])\ge c\,x^\alpha$ for all $x\in[0,\varepsilon]$.

    \item \emph{Distributional Lipschitz continuity.}
There exists $L<\infty$ such that for all $\theta,\theta'\in\Theta$, \( W_\infty(H_\theta,H_{\theta'}) \le L |\theta-\theta'|. \) For one-dimensional price distributions, we use $W_\infty$ to denote the infinity-Wasserstein distance. 

\item \emph{Posterior regularity.}
There exist constants $a_0,a_1>0$ satisfying the following conditions. First, the prior $\hat\pi$ has a strictly positive $C^2$ density on the relative interior of $\Theta$ and satisfies $-\partial_\theta^2\log\hat\pi(\theta)\ge a_0$ there. Second, for every price $x$, the effective parameter set $\Theta(x):=\{\theta\in\Theta:h_\theta(x)>0\}$ is an interval. On the relative interior of $\Theta(x)$, the map $\theta\mapsto\log h_\theta(x)$ is twice differentiable and satisfies $-\partial_\theta^2\log h_\theta(x)\ge a_1$.
  \end{enumerate}
\end{assumption}

\Cref{ass:regular_family} supports the long-window bound in three ways. First, the lower-tail condition ensures that, with many adjustment opportunities left, the oracle threshold is close to the lowest attainable price, so setting aside time for learning is not too costly. Second, the \(W_\infty\)-Lipschitz condition makes the oracle threshold stable under small changes in \(\theta\). Third, the likelihood and prior are sufficiently curved in the parameter, so observed price adjustments concentrate the posterior at the usual parametric rate. To gain intuition, in a Gaussian location family \(Y_i\mid\theta\sim N(\theta,\sigma^2)\), each post-adjustment price adds \(1/\sigma^2\) units of precision, so the posterior standard deviation shrinks at rate \(n^{-1/2}\). Assumption~\ref{ass:regular_family}(iii) imposes the analogous curvature requirement for the general family.

Next, we present a long-window bound under a \emph{learn-then-act} policy. Because this policy is feasible, its payoff gives a lower bound on the Bayesian value $V^B$. The policy reserves a portion of the window for passive learning, during which it observes post-adjustment prices, updates the posterior, and forms a confidence set for the unknown price-adjustment distribution. It then switches to an acting phase and follows an upper-envelope threshold: the largest oracle threshold over parameters that remain statistically plausible. With high probability, the true parameter lies in the confidence set, and the upper-envelope threshold is weakly above the true oracle threshold at every remaining time, so the policy stops no later than the oracle along every realized continuation path. The analysis then separates the information gap into two terms: the opportunity cost of reserving time for learning and the acting-phase loss from using a confidence-envelope threshold rather than the true oracle threshold. The regularity assumptions control these terms through lower-tail mass, Lipschitz continuity of the oracle thresholds, and posterior concentration. 

\begin{proposition}[{\sc Long-Window Information Gap}]\label{prop:info_gap_confidence}
Suppose \Cref{ass:regular_family} holds. Then there exists a constant $C<\infty$ such that, for every $T\ge1$,
\[
  \mathcal G(T,\hat\pi)
  \le
  C\left(
    T^{-1/\alpha}
    +
    \sqrt{\frac{\log(eT)}{T}}
  \right),
\]
and hence $\mathcal G(T,\hat\pi)\to 0$ as $T\to\infty$.
\end{proposition}

The bound reflects two forces. A longer learning phase provides the agent with price adjustments and reduces posterior uncertainty, thereby reducing the loss from using an estimated threshold during the acting phase. At the same time, delaying action is costly because the agent may miss buying opportunities. The learn-then-act policy balances these effects: the lower-tail condition controls the cost of postponing purchase, while posterior concentration controls the error in the acting phase. Thus, the information gap is bounded by a learning-delay term $T^{-1/\alpha}$ and a posterior-uncertainty term $\sqrt{\log(eT)/T}$.

The short- and long-window bounds identifies the horizons in which model uncertainty can matter most. In short windows, adjustments are unlikely to occur, so the value of knowing the price-adjustment distribution is limited. In long windows, adjustments make the latent regime learnable while leaving enough time to use the learned policy. Between these two extremes, model uncertainty can be more consequential: the agent may observe enough price adjustments for beliefs to affect the stopping rule, but not enough to resolve the latent state before the opportunity expires.


\subsection{Extension: Learning the Trend}
\label{subsec:trend_bayes}

The same learning logic can also accommodate deterministic trends in price adjustments. Concretely, suppose that when a price adjustment occurs at elapsed time \(t\), the post-adjustment price has the form \(g_\eta(t)+X_\theta\), where \(g_\eta\) is a parametric trend function indexed by \(\eta\in\mathcal H\) and \(X_\theta\sim H_\theta\) is a residual price component indexed by \(\theta\in\Theta\). The joint latent state is therefore \(\xi:=(\eta,\theta)\in\Xi:=\mathcal H\times\Theta\), with prior belief \(\pi\in\mathcal P(\Xi)\). Price adjustment times still follow the Poisson process with arrival rate $\lambda$, while post-adjustment prices are drawn from the shifted distribution \(g_\eta(t)+H_\theta\).

Posterior updating is likewise joint. Let \(h_\theta\) denote the density associated with \(H_\theta\). If the agent observes a post-adjustment price $p$ at elapsed time~$t$, then each candidate pair $(\eta,\theta)$ assigns likelihood $h_\theta\bigl(p-g_\eta(t)\bigr)$, so Bayes' rule updates the joint belief over trend and residual regime simultaneously. Intuitively, the agent uses the timing pattern of observed prices to learn whether prices are drifting upward or downward over the window, and uses detrended residuals to learn whether the seller's pricing regime is systematically aggressive or conservative.

The threshold structure carries over by the same monotonicity argument as in \Cref{prop:bayes_threshold}. For each time-to-go and joint posterior, the buy region in the current price remains a lower interval. Trend learning changes the predictive price-adjustment distribution and the belief state, but not the threshold form of the stopping rule.


This joint-learning formulation remains conceptually close to the baseline Bayesian model because, after writing $s=T-t$, the dynamic-programming state remains $(s,p,\pi)$, with elapsed time entering the distribution of the post-adjustment price through $T-s$. The belief's domain expands from $\Theta$ to $\Xi$. The main cost is computational rather than conceptual: the posterior now evolves over a higher-dimensional parameter space, and the purchase threshold depends on beliefs about both the trend and the residual noise. For this reason, the extension is most attractive when the trend family is low-dimensional. A fully nonparametric trend specification would substantially enlarge the state space and would likely require a separate approximation method.
 

\section{Robust Buying Under Minimal Information}\label{sec:robust}

The previous two sections rely on a probabilistic adjustment model of future prices. The Bayesian model relaxes assumptions about the price-adjustment distribution, but it still assumes that the timing of price changes follows a Poisson process. In some applications, however, price histories may be short, market conditions may shift abruptly, or seller behavior may be too unstable to support such a structure. These considerations motivate a robust OR model that drops the Poisson, independence, and stationarity assumptions from the adjustment models above and asks what guarantees remain when the agent trusts only coarse price information.

In this section, we work directly with realized price paths in elapsed time \(t\in[0,T]\). The analysis is model-free: no stochastic model is imposed on the price path's evolution. Prices may change irregularly or even in an adversarial manner. The only assumption we impose is that the realized price path \(p:[0,T]\to[p_L,p_U]\) is bounded, with \(0\le p_L\le p_0\le p_U\). 

The buying agent observes the path sequentially and must decide online whether and when to purchase. As in the earlier sections, the agent has valuation $v$, but now the only trusted market information consists of the observed price history, the known bounds $(p_L,p_U)$, and the initial price $p_0$. Throughout this section, we focus on the nondegenerate case \(v>p_L\). If \(v\le p_L\), every feasible price is weakly above the valuation, so \(\mathrm{OPT}(p)=0\) for all paths and the competitive-ratio problem is trivial.

\subsection{Online Policy and Benchmark}

An online policy $\mathcal A$ chooses a purchase time \(\tau^{\mathcal A}\in[0,T]\cup\{\infty\}\), expressed in elapsed-time units, based only on the observed portion of the price path. If the agent purchases by the deadline, its payoff under path $p(\cdot)$ is
\[
  U^{\mathcal A}(p)
  =
  \bigl(v-p(\tau^{\mathcal A})\bigr)^+.
\]
If no purchase occurs by the deadline, the payoff is zero.

To evaluate the quality of an online buying policy, we compare it with an offline benchmark that observes the entire realized price path in advance. The offline value is
\[
  \mathrm{OPT}(p)
  =
  \sup_{t\in[0,T]} \bigl(v-p(t)\bigr)^+.
\]
This benchmark represents the best possible purchase timing decision. The performance criterion is the competitive ratio
\[
  \mathrm{CR}(\mathcal A)
  =
  \inf_{p(\cdot)}
  \frac{\mathbb E[U^{\mathcal A}(p)]}{\mathrm{OPT}(p)},
\]
with the usual convention that the ratio equals $1$ when $\mathrm{OPT}(p)=0$. The expectation is over the internal randomization of the policy, if randomization is allowed. This multiplicative criterion asks what fraction of the offline surplus can be guaranteed uniformly over all feasible price paths. We later complement it with minimax regret, which evaluates the absolute surplus loss relative to the same offline benchmark.

\subsection{Worst-Case Price Trajectories}
To derive performance limits, it is enough to consider a simple subclass of feasible paths. The adversary can let the price decline gradually, making waiting attractive, while retaining the ability to end the markdown and revert to a high price at an unknown time. This creates the central online tension: buying early sacrifices possible future savings, but waiting long risks missing a sale.

\begin{definition}[Flash-Sale Process]\label{def:flash-sale}
A \emph{flash-sale process} is characterized by a pair $(\hat{p},\sigma)$ consisting of:
\begin{enumerate}
\item[(i)] A \emph{skimming trajectory} $\hat{p}:[t_1,t_2]\to[p_L,p_U]$, which is a continuous linear\footnote{A linear trajectory is not essential; rather, any continuous decreasing trajectory is admissible. We adopt a linear function for analytical convenience.} function with $\hat{p}(t_1)=p_0$ and $\hat{p}(t_2)=p_L$, representing the posted price while the sale is active;
\item[(ii)] An \emph{end-sale time} $\sigma$, which is a random variable supported on $[t_1,t_2]$, representing the random instant at which the sale terminates;
\end{enumerate}

Conditional on a realization of $\sigma$, the realized price path is
\begin{equation*}
  p(t)
  \;=\;
  \begin{cases}
  p_0, & t\in[0,t_1),\\
  \hat{p}(t)=p_0-\frac{p_0-p_L}{t_2-t_1}(t-t_1), & t\in[t_1,\sigma),\\
  p_U,  & t\in[\sigma,T].
  \end{cases}
\end{equation*}
The agent observes $p(t)$ sequentially. At each instant $t<\sigma$, the observed price $p(t)=\hat p(t)$ is consistent with the sale still being active, but the agent does not know the realization of $\sigma$ before the sale ends and cannot distinguish ``the sale continues'' from ``the sale will end soon.'' Once the price jumps to the \emph{reversion price} $p_U$ at time $\sigma$, the agent learns that the sale has ended.
\end{definition}

For each realized end-sale time, the offline optimum under the corresponding price path is
\[
  \mathrm{OPT}(p) = \bigl(v - \hat p(\sigma^-)\bigr)^+ = \bigl(v - \hat p(\sigma)\bigr)^+,
\]
since $\hat p$ is non-increasing and the best achievable price is the last sale price $\hat p(\sigma^-)$ just before expiry. An online policy that has not made a purchase by time $\sigma$ faces the reversion price $p_U$ for the remainder of the horizon.

\subsection{Randomized Threshold Policies}
\begin{definition}[Random Purchase Threshold Policy]\label{def:random-threshold}
A \emph{random purchase threshold policy} is an online policy that draws a random threshold $b$ from a distribution on $[p_L,p_U]$ at time $0$, before observing the future price path, and then keeps this threshold fixed until the deadline. Given the realized threshold $b$, the policy uses the static threshold before the deadline and defines
\[
\tau_b=\inf\{\,t\in[0,T): p(t)\le b\,\},
\]
with the convention that \(\tau_b=\infty\) if the set is empty. If \(\tau_b<\infty\), the policy purchases at \(\tau_b\). If no such crossing occurs before the deadline, the policy applies the terminal valuation rule at time $T$: it purchases if and only if \(p(T)\le v\).
\end{definition}

Equivalently, the policy is fully characterized by a distribution over thresholds together with the terminal action. Conditional on the draw of $b$, the agent follows the corresponding deterministic static threshold rule on $[0,T)$.

\begin{figure}[!htbp]
  \centering
  \caption{Two sampled randomized thresholds and price trajectories.}
  \label{fig:flash_sale_paths}
  \includegraphics[width=0.7\linewidth]{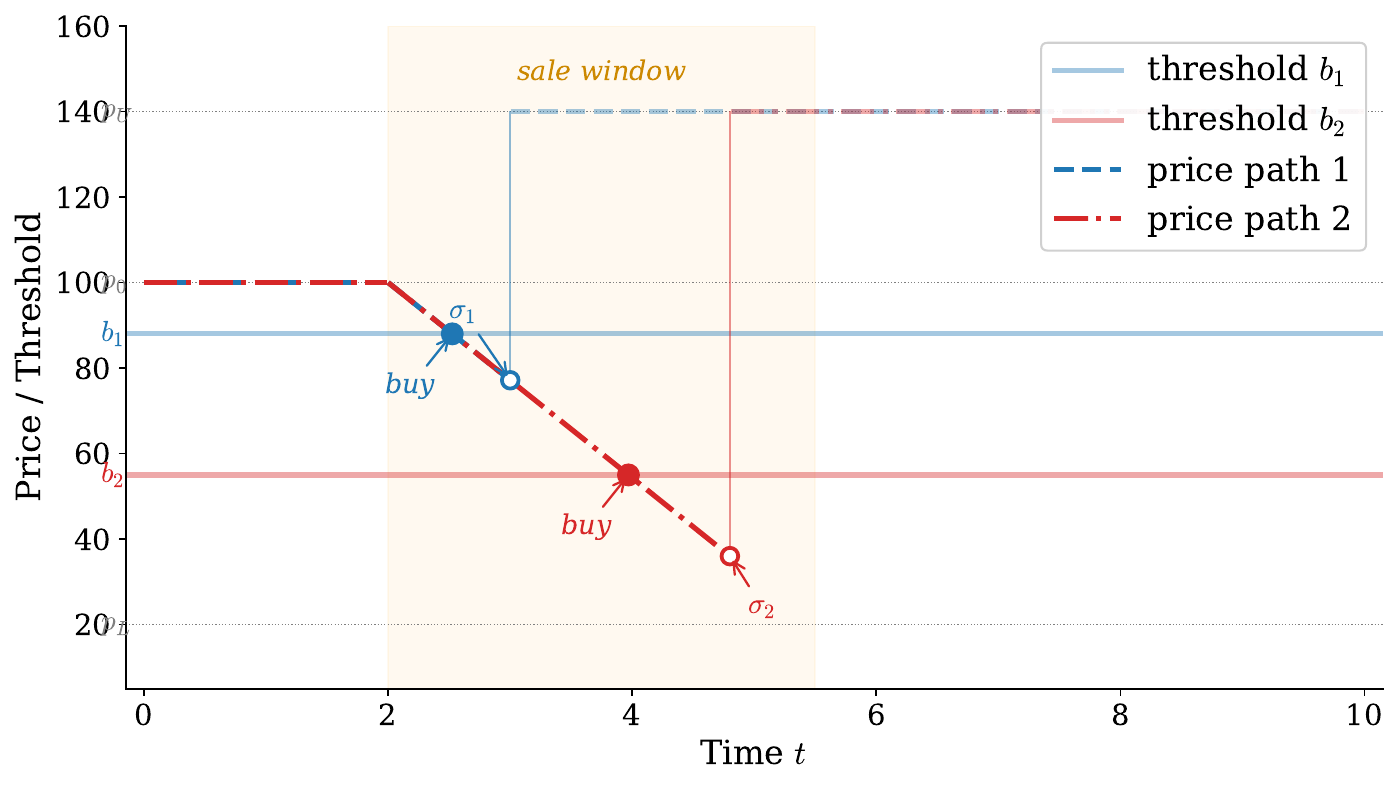}
\end{figure}

\subsection{Competitive Ratio}

We first characterize the optimal competitive ratio \(\mathrm{CR}^*\). The expression depends on how the consumer's valuation compares with the initial price and the upper price bound, because these comparisons determine both the value of buying immediately and the consequence of missing a temporary markdown. To characterize \(\mathrm{CR}^*\), we define \(\rho(v)\) as follows. 

If \(p_L<v\le p_0\), set \(\rho(v)=0\). If \(p_0<v\le p_U\), set
\[
\rho(v)
:=
\frac{1}{1+\log\left(\frac{v-p_L}{\,v-p_0\,}\right)}.
\]
If \(v>p_U\), let \(x_0\) denote the unique solution to \( \frac{x}{v-p_U}+\log x=\log(p_U-p_L)-1 \) and set
\[
\rho(v)
:=
\begin{cases}
\displaystyle
\frac{v-p_U}{x_0+v-p_U},
& x_0\ge p_U-p_0,\\[1.2em]
\displaystyle
\frac{v-p_0}{\,v-p_0 +(p_U-p_0)\log\left(\frac{p_U-p_L}{\,p_U-p_0\,}\right)},
& x_0<p_U-p_0.
\end{cases}
\]

\begin{proposition}[{\sc Optimal Competitive Ratio}]\label{prop:optimal_cr}
For \(v>p_L\), the optimal competitive ratio is \(\mathrm{CR}^*=\rho(v)\). Moreover, for \(v>p_0\), the ratio is achieved by a randomized threshold policy.
\end{proposition}

The proof proceeds by establishing matching upper and lower bounds. The upper bound constructs flash-sale processes under which no online policy can earn more than a \(\rho(v)\) fraction of the offline payoff. The lower bound constructs randomized threshold policies that guarantee \(\rho(v)\) against every feasible path. The upper-bound construction is related to the equalizing distributions used in online booking \citep{ball2009toward}, but the mathematical primitives differ. In the online-booking model, the adversary distribution is governed by a single fare bound. In our buying model, three price primitives matter: the initial price \(p_0\), the lower price bound \(p_L\), and the upper price bound \(p_U\). Moreover, missing the markdown has different consequences depending on whether \(v\le p_U\) or \(v>p_U\), which is why \(\rho(v)\) has valuation-dependent regimes rather than reducing directly to the competitive ratio in \cite{ball2009toward}.

\begin{lemma}[{\sc Flash-Sale Upper Bound}]\label{lem:flash_sale_upper_main}
For every online policy \(\mathcal A\), \(\mathrm{CR}(\mathcal A)\le \rho(v)\). Specifically, the upper bound is witnessed by the following flash-sale process. For the nondegenerate case \(p_0>p_L\), fix a reference time \(\tilde t\) on the sale trajectory and write \(q:=\hat p(\tilde t)\). The sale remains active until \(\tilde t\), and the random end-sale time after \(\tilde t\) has distribution
\begin{equation}\label{eq:flash_sale_end}
\Pbb \left[\sigma\le t\right] =
\begin{cases}
0, & t_1\le t<\tilde t,\\[0.8ex]
\dfrac{(p_0-p_L)(t-\tilde t)}
{A(t_2-t_1)+(p_0-p_L)(t-\tilde t)}, & \tilde t\le t<t_2,\\[2ex]
1, & t=t_2.
\end{cases}
\end{equation}
The regime-specific choices are:
\begin{enumerate}[label=\textnormal{(\roman*)}]
    \item If \(p_L<v\le p_0\), take \(q=v-\varepsilon\) and \(A=\varepsilon\); the resulting upper bound tends to zero as \(\varepsilon\downarrow0\).
    \item If \(p_0<v\le p_U\), take \(q=p_0\) and \(A=v-p_0\).
    \item If \(v>p_U\) and \(x_0<p_U-p_0\), take \(q=p_0\) and \(A=p_U-p_0\).
    \item If \(v>p_U\) and \(x_0\ge p_U-p_0\), take \(q=p_U-x_0\) and \(A=x_0\).
\end{enumerate}
\end{lemma}

Under these distributions, any deterministic online rule can be viewed as choosing a target markdown. The induced markdown-depth distribution equalizes the online payoff across target markdowns, while the offline oracle captures the final sale price just before expiry. The proof applies Yao's principle and computes the resulting expected payoff ratios.

The matching lower bound is constructive: a randomized static threshold, with distribution chosen by valuation regime, balances the potential gain from waiting for lower prices against the risk of missing the eventual best price.

\begin{lemma}[{\sc Robustly Optimal Randomized Threshold Policies}]
\label{lem:matching_threshold_policy}
There exists a randomized threshold policy whose competitive ratio is at least \(\rho(v)\). The policy draws a random threshold $b$ before the price path is revealed, and purchases at the first pre-deadline time when the observed price falls to or below $b$. If the threshold is not reached before the deadline, the policy applies the terminal valuation rule at \(T\), purchasing if and only if \(p(T)\le v\).

For \(p_L<v\le p_0\), the claim is immediate because \(\rho(v)=0\). For \(v>p_0\), write \(a:=\min\{v,p_U\}\) and define
\[
c(v):=
\begin{cases}
p_0, & p_0<v\le p_U,\\
p_0, & v>p_U\ \text{and}\ x_0<p_U-p_0,\\
p_U-x_0, & v>p_U\ \text{and}\ x_0\ge p_U-p_0.
\end{cases}
\]
Draw \(b\) from the distribution
\[
F_b(z):=
\begin{cases}
0, & z<p_L,\\[0.6ex]
\rho(v)\log\!\left(\dfrac{a-p_L}{a-z}\right), & p_L\le z<c(v),\\[1.6ex]
1, & z\ge c(v).
\end{cases}
\]
This policy achieves a competitive ratio of \(\rho(v)\).
\end{lemma}

Together, \Cref{lem:flash_sale_upper_main,lem:matching_threshold_policy} imply \(\mathrm{CR}^*\le\rho(v)\) and \(\mathrm{CR}^*\ge\rho(v)\), respectively, and therefore prove \Cref{prop:optimal_cr}. The zero guarantee for \(p_L<v\le p_0\) reflects the multiplicative nature of competitive analysis: the initial price yields no positive surplus, and an adversary can make any attractive markdown arbitrarily small. The payoff ratio can therefore collapse even when the absolute loss remains meaningful. This motivates the minimax regret criterion below.

\subsection{Minimax Regret}
In this subsection, we introduce and characterize minimax regret under the same primitives. This additive criterion compares the online payoff to the same offline benchmark in terms of a surplus value, rather than as a payoff ratio. 

For a price path $p(\cdot)$ and an online policy $\mathcal A$ with purchase time $\tau^{\mathcal A}$, the regret of policy $\mathcal A$ against path $p$ is
\[
\Reg(\mathcal A,p)=\OPT(p)-\mathbb E\left[U^{\mathcal A}(p)\right].
\]
The associated minimax regret is
\[
\Reg^*=\inf_{\mathcal A}\sup_{p(\cdot)} \Reg(\mathcal A,p),
\]
where the infimum is over all possibly randomized online policies and the supremum is over all feasible price paths in $[p_L,p_U]$.

For \(v>p_L\), write \(a:=\min\{v,p_U\}\), \(\gamma:=\min\{p_0,\; a-(a-p_L)/e\}\), and \(\Delta:=a-\gamma\).
We will show that the minimax regret is
\[
\bar R(v)
:=
\Delta\log\frac{a-p_L}{\Delta}.
\]

\begin{proposition}[{\sc Minimax Regret}]
\label{prop:minimax-regret}
For $v>p_L$, the minimax regret is $\Reg^*=\bar R(v)$. Moreover, the robustly optimal policy is a randomized threshold policy.
\end{proposition}

The proof again proceeds by establishing matching upper and lower bounds. The lower bound constructs flash-sale distributions under which every online policy incurs regret at least \(\bar R(v)\). The upper bound constructs randomized threshold policies whose worst-case regret is at most \(\bar R(v)\).

\begin{lemma}[{\sc Adversarial Regret Lower Bound}]
\label{lem:regret_lower_bound}
For every online policy $\mathcal A$,
\[
\sup_{p(\cdot)}\Reg(\mathcal A,p)\ge \bar R(v).
\]
The lower bound is witnessed by the following flash-sale process. With \(a\), \(\gamma\), and \(\Delta\) as defined above, fix a skimming trajectory $\hat p$ and a reference time $\tilde t_R$ such that $\hat p(\tilde t_R)=\gamma$. The sale remains active until $\tilde t_R$, and the end-sale time after $\tilde t_R$ is chosen according to
\[
\Pbb \left[\sigma\le t\right] =
\begin{cases}
0, & t_1\le t<\tilde t_R,\\[0.8ex]
\dfrac{(p_0-p_L)(t-\tilde t_R)}
{\Delta(t_2-t_1)+(p_0-p_L)(t-\tilde t_R)}, & \tilde t_R\le t<t_2,\\[2ex]
1, & t=t_2.
\end{cases}
\]
The realized path follows $\hat p$ until $\sigma$ and then jumps to $p_U$.
\end{lemma}

\begin{lemma}[{\sc Achievable Regret Bound}]
\label{lem:regret_achievable}
There exists a randomized static threshold policy $\mathcal A$ such that
\[
\sup_{p(\cdot)}\Reg(\mathcal A,p)\le \bar R(v).
\]
The policy is constructed as follows. If $p_L<v\le p_0$, set $\bar p:=v-(v-p_L)/e$ and draw $b$ with density $1/(v-b)$ on $[p_L,\bar p]$.

If $v>p_0$, let $L_a:=\log((a-p_L)/(a-p_0))$, with the convention that \(L_a=+\infty\) when \(a=p_0\). When $L_a\le1$, draw $b$ with density $1/(a-b)$ on $[p_L,p_0)$ and put an atom $1-L_a$ at $p_0$. When $L_a>1$, set $\bar p_a:=a-(a-p_L)/e$ and draw $b$ with density $1/(a-b)$ on $[p_L,\bar p_a]$. In all regimes, the drawn threshold governs purchases only on $[0,T)$; if it is not reached before the deadline, the policy buys at \(T\) if and only if \(p(T)\le v\).
\end{lemma}

Together, \Cref{lem:regret_lower_bound,lem:regret_achievable} imply \(\Reg^*\ge\bar R(v)\) and \(\Reg^*\le\bar R(v)\), respectively, and therefore prove \Cref{prop:minimax-regret}.

\subsection{Extension: Decreasing Price Trend}
\label{subsec:robust_trend_envelopes}

The robust construction also applies when the trusted price information takes the form of decreasing price envelopes rather than constant bounds. Suppose the agent knows two non-increasing functions \(p_L(\cdot)\) and \(p_U(\cdot)\) such that every feasible price path satisfies \(p_L(t)\le p(t)\le p_U(t),  t\in[0,T]\), with \(p(0) = p_0\). This specification captures settings in which prices are expected to move downward over the shopping window, but the realized path may still fluctuate within a time-varying band. The terminal lower envelope \(p_L(T)\) is the lowest price that the adversary can make available, while the terminal upper envelope \(p_U(T)\) is the worst price the agent must face if it waits until the deadline.

Thus, the decreasing-trend problem reduces to the constant-bound problem with the effective primitives
$
p_L^{\mathrm{eff}}:=p_L(T),p_U^{\mathrm{eff}}:=p_U(T),p_0^{\mathrm{eff}}:=\min\{p_0,p_U(T)\}.
$
The effective initial price reflects the fact that the agent can always wait until the deadline and face a price no higher than \(p_U(T)\). The randomized threshold policies and the competitive-ratio and minimax-regret expressions above therefore carry over after replacing \((p_L,p_U,p_0)\) by \((p_L^{\mathrm{eff}},p_U^{\mathrm{eff}},p_0^{\mathrm{eff}})\). Operationally, the agent draws a static price threshold from the same distribution as in the constant-bound model, calibrated to the terminal band, and makes a purchase the first time the realized price falls to or below that threshold before the deadline. If the threshold is not reached before the deadline, the agent applies the terminal valuation rule at \(T\), buying if and only if \(p(T)\le v\). The guarantee uses \(p_U^{\mathrm{eff}}\) to bound the missed-threshold continuation payoff. \Cref{ec:robust_trend_envelopes} explains how the adversarial and achievability arguments map to this envelope formulation.

\section{Empirical Evaluation of Buying-Agent Policies}\label{sec:numerical}

This section evaluates a set of buying policies, including our OR policies and several comparison policies, using real price histories. We construct 1,000 shopping instances from posted-price histories of 367 items on Amazon. We use prior observations to estimate parameters and choose among candidate models, and reserve the subsequent price path for sequential evaluation. In the remainder of this section, we first describe the construction of the evaluation instances, then introduce the evaluated policies and selection mechanisms, report their surplus performance, and finally illustrate the agent workflow on a common shopping instance.

\subsection{Instance Construction}\label{subsec:numerical-design}

We construct the evaluation instances from time-stamped posted-price histories from Keepa, an Amazon price-tracking service. The final evaluation set contains 1,000 shopping instances spanning 367 items, with underlying records that include 48,933 time-stamped price observations. To form the instance pool, we randomly sample 3 shopping windows per item and 3 valuations per window, discarding windows with no in-window price change, and then randomly select 1,000 instances.

For each instance, we split the price path into a pre-window history and a held-out 30-day shopping window. The pre-window history is the only price information available for calibration, model selection, and language-model prompts at the beginning of the shopping window. During evaluation, each policy observes shopping-window prices only sequentially as they become available. The full held-out path is reserved for ex-post computation of the offline benchmark.

We require each selected instance to have at least five pre-window price observations for parameter calibration. We also require at least two shopping-window observations, counting the window-start price, to ensure a non-trivial sequential decision problem. When the price is not observed exactly at the window start, the window-start price is taken to be the last posted price before the window begins. The consumer valuation is generated as a multiple of the window-start price, with the multiplier drawn uniformly from $[1.01,1.50]$.

\subsection{Candidate Policies and Selection}\label{subsec:deployment-mechanisms}

This subsection defines the decision mechanisms evaluated in the numerical study. 

\paragraph{Fixed OR policies.}
The Fixed OR policies apply a common stopping-rule specification to all instances, but calibrate the model inputs separately using each instance's pre-window price history. We evaluate four such policies: \polStationary, \polBayesian, \polRobust, and \polTrendAware. \polStationary is an implementation of the stationary model in \Cref{sec:stationary}: it estimates the adjustment intensity and price-adjustment distribution from the pre-window history and solves the model's threshold equation. \polBayesian implements the Bayesian learning model in \Cref{sec:bayesian}, using the same adjustment-intensity estimate and a discretized Bayesian predictive model for uncertainty in the price-adjustment distribution; prior counts are initialized from pre-window prices and updated as shopping-window prices are observed. \polRobust implements the robust buying model in \Cref{sec:robust}: it constructs feasible lower and upper price bounds from the pre-window history and the window-start price, and applies the model's randomized threshold policy. \polTrendAware implements the time-varying price-adjustment extension in \Cref{subsec:trend_stationary}: it fits a deterministic price trend on the pre-window history and applies the stopping computation to detrended residual prices. \Cref{Ec:policy_calibration} reports the calibration details.

\paragraph{Heuristic Selector.} 

We also evaluate a Heuristic Selector, which implements model selection through a sequence of pre-specified diagnostic screens and serves as a simple rule-based benchmark for automated policy selection. The selector uses only the pre-window history to assign each instance to one of the candidate OR policies: \polStationary, \polBayesian, \polRobust, or \polTrendAware. It applies the screens sequentially. The selector first checks the stability of the adjustment intensity and routes histories with unreliable timing to \polRobust. Then it applies a multiple-break screen for level shifts in pre-window prices, in the spirit of \cite{bai1998estimating,bai2003computation}, and a linear-trend screen. Among non-trending histories, unexplained level shifts trigger \polRobust; otherwise, sufficiently long histories are routed to \polStationary, and shorter histories to \polBayesian. Among trending histories, the selector detrends the prices and then reapplies the level-shift screen to the residuals; stable residuals yield \polTrendAware, while unstable residuals yield \polRobust. The screening details are reported in \cref{Ec:heuristic_selector}, and the diagnostic cutoffs used in the reported run are fixed across instances. 

\paragraph{Language-model agents.}
The third benchmark group consists of two language-model agents: LLM-baseline and LLM-OR. In the reported implementation, both agents call the DeepSeek API with the \texttt{deepseek-v4-flash} model. LLM-baseline is a direct-prompting agent. At the beginning of each shopping instance, the LLM receives the decision context, including the pre-window price history, valuation, and deadline. In a live deployment, the agent operates alongside an automatic price-monitoring program that calls the LLM at the start of the shopping window and after every price change. When queried at time $t$, the LLM receives the conversation history of this shopping task, the current price $P_t$, and the remaining time $T-t$. It then chooses either to purchase immediately or to wait and schedule a future purchase time $\xi\in(t,T]$. If the price remains unchanged before $\xi$, the purchase is triggered. Otherwise, if a new price $P_{\tilde t}$ arrives at $\tilde t < \xi$, the LLM is re-queried with the updated price $P_{\tilde t}$ and the new remaining time $T-\tilde t$ appended to the ongoing conversation history.\footnote{We do not query the LLM continuously. Suppose each query independently carries a positive probability of recommending a purchase; then, highly frequent querying would increase the chance that a purchase is triggered early. The scheduled purchase time thus serves as a commitment device for more disciplined execution.} In our numerical study, rather than using a live price monitor, we emulate this process on held-out historical price paths. The simulation directly compares the scheduled purchase time $\xi$ with the next price-adjustment time $\tilde t$ to decide whether to purchase or to query the LLM again. LLM-OR uses the language model at the beginning of the shopping window to select an OR model and the calibration samples, that is, the subset of pre-window price observations used to estimate the model inputs. Specifically, the LLM chooses a lookback horizon and a maximum number of recent price observations to use. The selected OR model then computes the threshold policy for making buy-or-wait decisions along the realized price path.

\paragraph{Simple baselines.}
We also include several simple policies that do not use our OR policies. The first group consists of two naive rules: \polBuyNow, which buys immediately when the initial price is no greater than the valuation, and \polBuyLast, which waits until the deadline and buys only if the last price is no greater than the valuation. The second group consists of three historical-price heuristics: \polHistMean, \polHistMed, and \polHistMin, which buy once the current price falls to or below the corresponding historical benchmark. If no purchase is triggered before the deadline, these waiting policies make a terminal purchase only when the last observed price is at most the valuation.

\subsection{Performance Measures and Results}\label{subsec:numerical-results}

Consistent with the payoff convention in \Cref{sec:environment}, we measure performance by realized consumer surplus: \((v-p_\tau)^+\) for a purchase at time \(\tau\) with price \(p_\tau\), and zero if no purchase is made. To compare outcomes across products with different price and valuation scales, we use normalized surplus, defined as realized surplus divided by the ex-post optimal surplus in the same shopping window whenever the latter is positive. When the offline best surplus is $0$, we simply set the normalized surplus to $1$. The offline best price is used only for ex-post evaluation. All methods are evaluated independently on the same set of shopping instances for comparability.


Table~\ref{tab:tie_aware_win_frequency} reports the two rate-based measures. The win rate is tie-aware and measures relative performance: any method that achieves the highest normalized surplus among all evaluated methods on an instance earns a win (i.e., if multiple methods tie for the best performance in the evaluated pool, all receive a win, hence the sum of win rates may exceed one). The optimal-hit rate counts instances on which a method attains the ex-post optimal surplus (i.e., buying at the lowest price along the whole shopping window, rather than merely outperforming other methods). The table shows a tight leading group on these measures: LLM-OR, \polBayesian, \polStationary, and \polHistMin all attain win rates above 48\% and optimal-hit rates of at least 43\%. LLM-OR belongs to this group, with a win rate of 48.9\% and an optimal-hit rate of 44.1\%, indicating that language-model-assisted model selection can recover competitive stopping decisions when the selected OR policy controls purchase timing. The \polStationary policy remains competitive despite its restrictive price-adjustment assumption, suggesting that a calibrated threshold rule captures much of the value of waiting in these instances.

\begin{table}[ht]
\centering
\caption{Win Rates Across 1,000 Instances. Win rate is tie-aware and counts all methods attaining the instance-level maximum surplus. Opt. hit rate counts instances with zero regret relative to the offline best surplus.}
\label{tab:tie_aware_win_frequency}
\begin{tabular*}{0.85\textwidth}{@{\extracolsep{\fill}}llcc}
\toprule
Method & Policy type & Win rate & Opt. hit rate \\
\midrule
LLM-OR & LLM-agent & 48.9\% & 44.1\% \\
Stationary & Fixed OR & 48.3\% & 43.2\% \\
Bayesian & Fixed OR & 48.3\% & 43.1\% \\
Historical Min & Simple baseline & 48.1\% & 43.0\% \\
LLM-baseline & LLM-agent & 46.8\% & 40.4\% \\
Buy Last & Simple baseline & 44.9\% & 40.9\% \\
Trend-aware & Fixed OR & 43.9\% & 38.8\% \\
Historical Median & Simple baseline & 43.6\% & 38.5\% \\
Historical Mean & Simple baseline & 43.4\% & 38.6\% \\
Buy Now & Simple baseline & 30.4\% & 25.8\% \\
Heuristic Selector & Heuristic selector & 26.5\% & 21.6\% \\
Robust & Fixed OR & 21.9\% & 17.6\% \\
\bottomrule
\end{tabular*}
\end{table}

Figure~\ref{fig:normalized_surplus_quantile_ranges} complements the rate-based comparison in Table~\ref{tab:tie_aware_win_frequency} by showing the cross-instance distribution of normalized surplus. Because normalized surplus is measured relative to the ex-post optimal surplus in the same instance, values near one indicate decisions that nearly match the offline benchmark. The leading methods are close over much of the distribution, but their quantile performance differs more substantially. For example, \polHistMin and \polBuyLast score well on optimal-hit rates, a pattern consistent with instances in which the lowest observed price occurs near an endpoint of the shopping window, but their lower tails are weaker. By contrast, \polRobust and \polSelector have lower win rates but stronger 10th-percentile normalized surplus than the methods with the highest win and optimal-hit rates, reflecting more conservative downside protection. These contrasts suggest that policy quality depends not only on how often a rule attains an ex-post hit, but also on how it trades off hit frequency against tail risk. The gap between \polSelector and the strongest fixed policies suggests room for improved selection rules.

\begin{figure}[tbp]
\centering
\caption{Distribution of normalized surplus across the 1,000 instances. For each method, the light segment spans the 10th to 90th percentiles, the dark segment spans the interquartile range, the black dot marks the median, and the orange diamond marks the mean. Methods are sorted by mean normalized surplus.}
\includegraphics[width=0.9\textwidth]{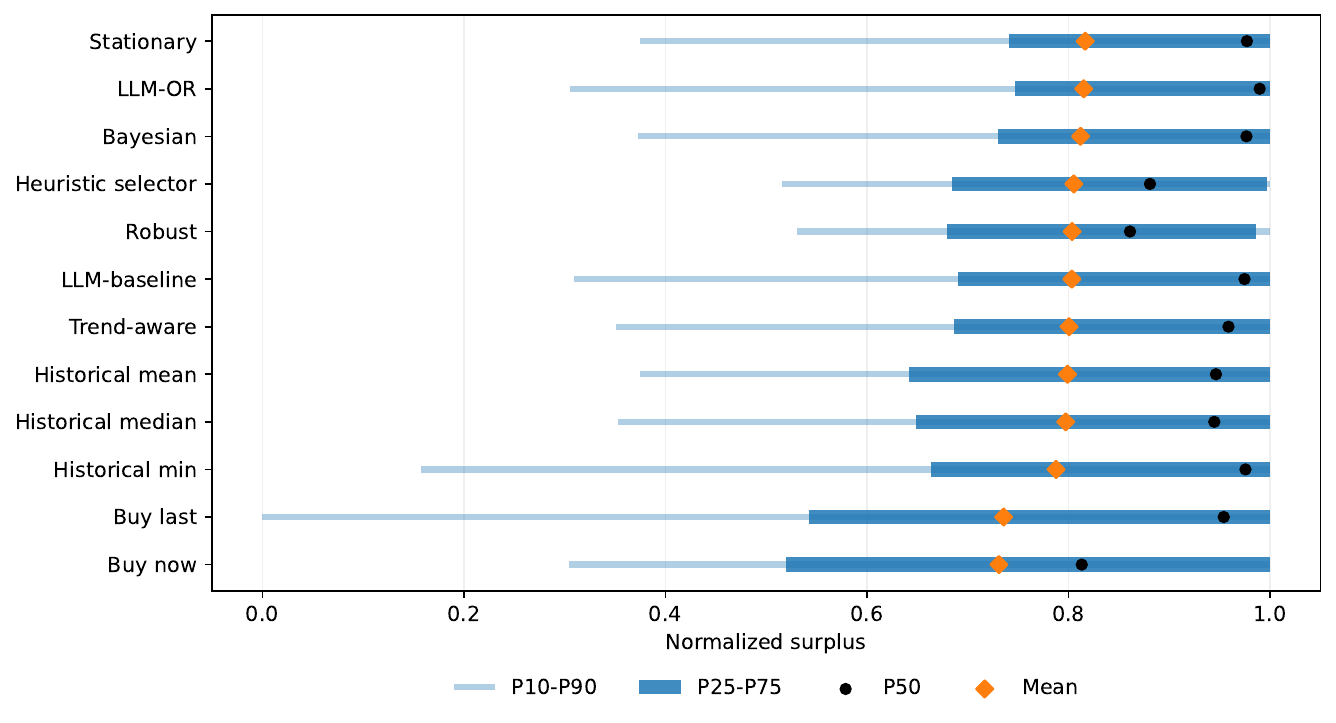}
\label{fig:normalized_surplus_quantile_ranges}
\end{figure}

\subsection{Illustration of Buying-Agent Trace}\label{subsec:agent-illustration}

We conclude the numerical study with an illustrative trace comparing the LLM baseline and LLM-OR on the same shopping instance. As described in \Cref{subsec:deployment-mechanisms}, LLM-baseline is queried at the window start and after each price change to make purchase-timing decisions, with any scheduled action superseded by a new price adjustment. LLM-OR uses the language model to select the OR model, choose the calibration history, and generate user-facing communication, while the selected OR policy controls the sequential purchasing decisions along the realized price path.


\begin{figure}[tbp]
  \centering
  \caption{Decision trace for one evaluation instance. Price panels show pre-window history, shopping-window prices, and realized purchase outcomes; the LLM-OR panel also highlights the truncated price data chosen by the LLM and used by the OR policy. Message strips report the corresponding summaries and purchase explanations.}
  \label{fig:buying_agent_decision_trace}
  \includegraphics[width=0.88\linewidth,height=0.84\textheight,keepaspectratio]{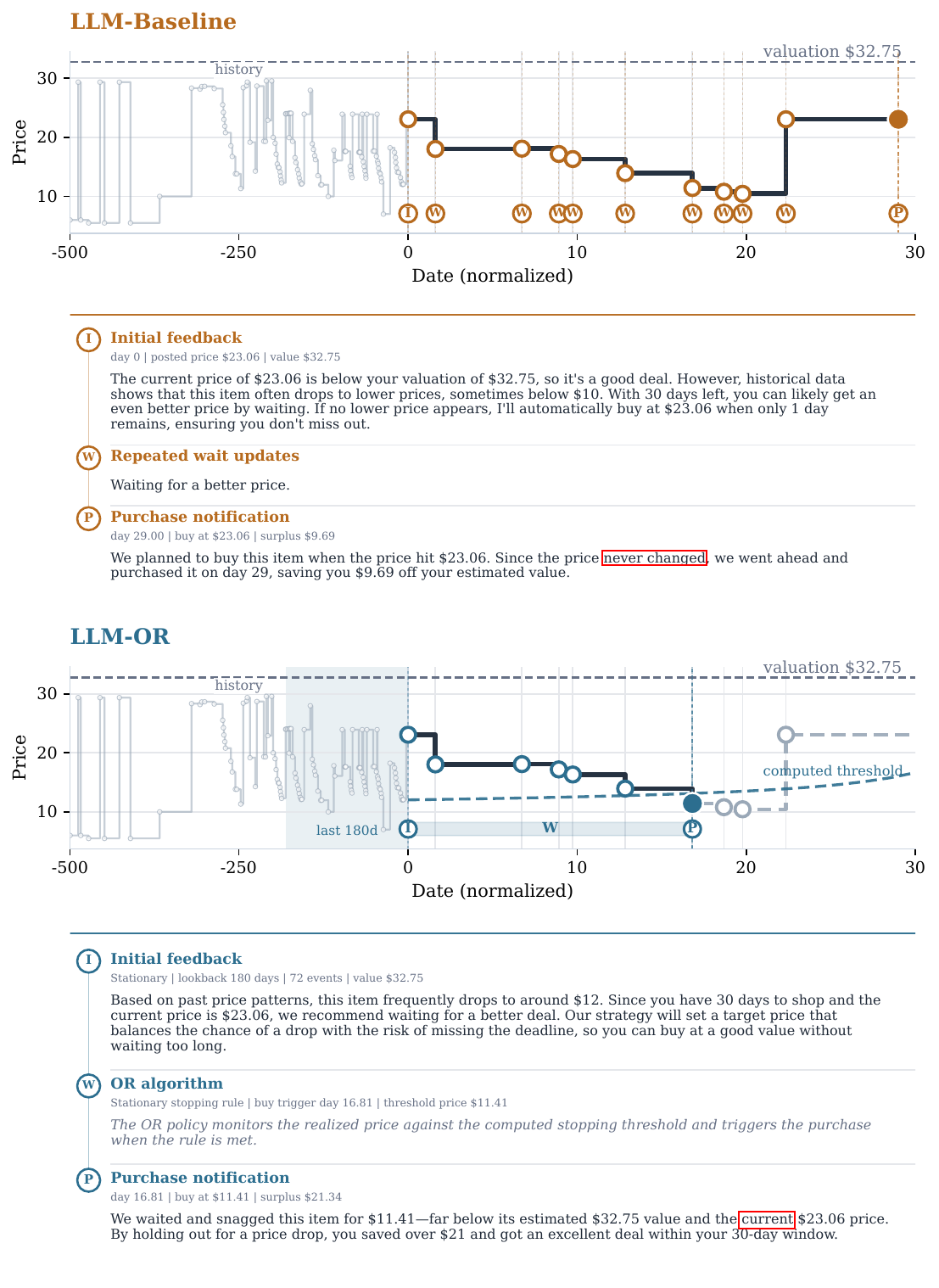}
  \begin{minipage}{\linewidth}
  \footnotesize
  \emph{Note:} The message excerpts are reproduced directly from the LLM output. The {\color{red}red} boxes mark two minor factual errors in the raw purchase notification: the first should indicate that the price is the same as the initial observed price, rather than that it never changed; the second should refer to the initial price at the start of monitoring, rather than to a current price of 23.06.
  \end{minipage}
\end{figure}

As part of the implementation diagnostic, we also record token usage. Across the 1,000 instances, LLM-OR uses 5.26 million total tokens (2.12 million input + 3.14 million output), compared with 21.80 million (19.26 million input + 2.53 million output) for LLM-baseline. This difference comes from using the OR policy rather than repeated LLM queries.
\section{Conclusion}

We develop a buyer-side operations framework for designing automated buying agents. We formulate purchase timing as a finite-horizon stopping problem and study how the optimal purchasing policy changes as the agent's information about future prices weakens. In the stationary price-adjustment benchmark, the optimal policy is a dynamic threshold characterized by an ordinary differential equation. Under Bayesian uncertainty about the price-adjustment distribution, the threshold becomes belief-dependent, and the information gap quantifies the value of knowing the true price-adjustment distribution. Under minimal information, randomized threshold policies provide competitive-ratio and minimax-regret guarantees. Together, these results show how the structure of the purchase policy changes as the agent moves from a calibrated stochastic model to learning and then to robust protection.

The numerical study complements the analytical results by applying these policies to real price histories. The comparison indicates that the OR policies proposed in this paper, especially the \polStationary policy, perform competitively on the evaluation instances. The LLM-OR implementation illustrates a complementary role for language models in buying-agent design: the LLM can select a model, choose calibration data, and generate user-facing explanations, while the selected OR policy retains control of the sequential purchase-timing decision.

Our analysis isolates the purchase-timing decision in a setting with a single item, a known valuation, and an exogenous price path. This abstraction separates the core timing problem from other market frictions and allows us to focus on how the agent's information about future prices shapes the purchasing policy. It also defines the limits of the current framework and points to several natural extensions. Relaxing the memoryless price-adjustment assumption yields richer autoregressive or contextual price dynamics. Relaxing the single-item assumption yields multi-item formulations that allow substitution across products. Finally, incorporating stock-out risk would make the purchase-timing decision account for both future prices and product availability.


\PutSingleSpacedBib
\end{bibunit}
\ECSwitch
\begin{bibunit}
\ECBibHyperlinks

\OneAndAHalfSpacedXI
\ECHead{\begin{center}E-Companion for \\``Strategic Buying Agents''\end{center}}
\makeatletter \providecommand*{\theHsection}{} \providecommand*{\theHsubsection}{} \providecommand*{\theHsubsubsection}{} \providecommand*{\theHtheorem}{} \providecommand*{\theHlemma}{} \providecommand*{\theHproposition}{} \providecommand*{\theHcorollary}{} \providecommand*{\theHassumption}{} \providecommand*{\theHdefinition}{} \providecommand*{\theHremark}{} \providecommand*{\theHexample}{} \renewcommand*{\theHsection}{EC.\arabic{section}} \renewcommand*{\theHsubsection}{EC.\arabic{section}.\arabic{subsection}} \renewcommand*{\theHsubsubsection}{EC.\arabic{section}.\arabic{subsection}.\arabic{subsubsection}} \renewcommand*{\theHtheorem}{EC.\arabic{section}.\arabic{theorem}} \renewcommand*{\theHlemma}{EC.\arabic{section}.\arabic{lemma}} \renewcommand*{\theHproposition}{EC.\arabic{section}.\arabic{proposition}} \renewcommand*{\theHcorollary}{EC.\arabic{section}.\arabic{corollary}} \renewcommand*{\theHassumption}{EC.\arabic{section}.\arabic{assumption}} \renewcommand*{\theHdefinition}{EC.\arabic{section}.\arabic{definition}} \renewcommand*{\theHremark}{EC.\arabic{section}.\arabic{remark}} \renewcommand*{\theHexample}{EC.\arabic{section}.\arabic{example}} \makeatother

\section{Details of OR Policies}\label{Ec:policy_calibration}

This section records how the candidate OR policies in \Cref{subsec:deployment-mechanisms} are calibrated from the pre-window history. For an evaluation instance, let \((t_i,p_i)_{i=1}^{n_h}\) denote the pre-window price observations, with times measured relative to the start of the shopping window, and let \(p_0\) denote the window-start price. The held-out shopping-window prices are not used for calibration, except that \(p_0\) is part of the online state faced by every policy at the beginning of the window.

\textbf{Stationary.} The Stationary policy implementation estimates the Poisson adjustment intensity $\widehat\lambda$ by the empirical frequency of price-adjustment timestamps over the historical span. The price-adjustment distribution \(H\) is replaced by the empirical distribution
\[
\widehat H=\frac{1}{n_h}\sum_{i=1}^{n_h}\delta_{p_i}.
\]
The policy then solves the stationary threshold ODE in \Cref{sec:stationary} with \((\widehat\lambda,\widehat H,v,T)\) and makes a purchase when the current price is no larger than the resulting time-to-go threshold.

\textbf{Bayesian.} The Bayesian policy implementation uses the same adjustment intensity estimate \(\widehat\lambda\), but replaces the unknown price-adjustment distribution by a finite Bayesian predictive model. The implementation discretizes prices into bins scaled to the pre-window price level and places a Dirichlet prior over the bin probabilities. Pre-window prices initialize the prior counts, and observed price adjustments during the shopping window update the posterior count vector. The dynamic program is then solved on this finite belief-state representation to obtain the belief-dependent threshold policy.

\textbf{Robust.} The Robust policy calibrates only price bounds. The lower bound is set by the smallest price observed in the pre-window history together with \(p_0\). For the upper bound, the implementation screens the pre-window history to remove isolated extremely high prices that are not representative of the feasible price range for the shopping window, and then uses the upper end of the screened history, enlarged if needed to include \(p_0\). Thus, the robust model is calibrated only on information available at the start of the shopping window. The resulting primitives \((\widehat p_L,p_0,\widehat p_U,v)\) are substituted into the randomized robust threshold distribution in \Cref{sec:robust}. The numerical evaluation reports the expected performance of this randomized policy by averaging over its threshold randomization.

\textbf{Trend-aware.} \polTrendAware fits the pre-window linear trend
\[
p_i=\alpha+\beta t_i+\varepsilon_i
\]
by least squares, forms residual prices \(x_i=p_i-(\alpha+\beta t_i)\), and uses the empirical residual distribution as the price-adjustment distribution around the deterministic trend. The fitted intercept is shifted to the shopping-window start, the adjustment intensity is set to \(\widehat\lambda\), and the threshold is computed in original price units by combining the deterministic trend with the residual continuation calculation. When the fitted trend is numerically negligible, or the pre-window history is too short to estimate it, the implementation falls back to \polStationary.

\section{Heuristic Selector}\label{Ec:heuristic_selector}

\polSelector in \Cref{subsec:deployment-mechanisms} maps each pre-window price history to one candidate model. It first screens the history for evidence about adjustment timing, level shifts, and deterministic trends, and then routes the instance to \polStationary, \polBayesian, \polTrendAware, or \polRobust according to the decision logic described in \Cref{subsec:deployment-mechanisms}. All screening cutoffs are fixed before evaluation and are held constant across instances. 

The timing screen partitions the pre-window observation span into at most five equal-length bins and computes Pearson's statistic
\[
X^2=\sum_{j=1}^{K}\frac{(N_j-\bar n)^2}{\bar n},
\]
where \(N_j\) is the number of observed price-adjustment timestamps in bin \(j\) and \(\bar n=n/K\). The statistic is compared with a \(\chi^2_{K-1}\) benchmark using the two-sided p-value \(p_\chi=2\min\{\Pr[\chi^2_{K-1}\ge X^2],\Pr[\chi^2_{K-1}\le X^2]\}\). The reported run uses \(\alpha_\chi=0.5\) and requires at least 10 pre-window timestamps for the timing screen. When fewer timestamps are available, the timing screen is skipped. The number of bins is reduced when needed to keep the expected cell count at least five.

The level-shift screen is Bai--Perron-style \citep{bai1998estimating,bai2003computation} in the following sense. For each candidate number of breaks, it fits a piecewise-constant mean model by dynamic programming, minimizing within-segment squared deviations subject to a minimum segment length. The implementation allows at most five breaks, requires each segment to contain at least three observations and, when feasible, at least 10\% of the pre-window sample, and selects the number of breaks by BIC. We use this screen only for routing diagnostics and do not use the estimated break dates for formal inference. The selector treats the pre-window price sequence as having unexplained level shifts when the selected number of breaks exceeds the tolerance \(\kappa=1\).

The trend screen fits a linear regression of pre-window price on time and flags a trend when the absolute \(t\)-statistic of the slope exceeds 2. If a trend is flagged, the selector applies the same level-shift screen to detrended residuals. Histories with unreliable timing, unexplained raw breaks, or unstable detrended residuals are routed to \polRobust; trending histories with stable residuals are routed to \polTrendAware; and stable non-trending histories are routed to \polStationary when the pre-window sample has at least \(n_0=15\) observations and to \polBayesian otherwise.

\section{Proofs}
\subsection{Proof of \Cref{lemma:mono_V}}
\begin{proof}{Proof.}
Fix $s\in[0,T]$ and let $0\le p_1\le p_2$. Couple the two systems on the same probability space so that they share the same future event times and, at each future event, the same post-adjustment price draw. Let $\sigma$ denote the first future event time.

For any admissible stopping time $\tau$, let $J(s,p;\tau)$ denote the expected payoff obtained from the initial state $(s,p)$ under $\tau$.

Let $\tau_2$ be an $\varepsilon$-optimal stopping time for the problem starting from $(s,p_2)$, where $\varepsilon>0$ is arbitrary. We construct an admissible stopping rule $\tau_1$ for the problem starting from $(s,p_1)$ as follows: before the first future event, the agent ignores the fact that the current price is lower and behaves exactly as under $\tau_2$; if $\tau_2<\sigma$, then set $\tau_1=\tau_2$; if $\tau_2\ge \sigma$, then after time $\sigma$ the two systems have the same time-to-go and the same posted price, so from that point onward $\tau_1$ follows exactly the same continuation rule as $\tau_2$.

Under this construction, if $\tau_2<\sigma$, then both agents stop before any price change, and hence the payoff satisfies \((v-p_1)^+\ge (v-p_2)^+\). If $\tau_2\ge \sigma$, then from time $\sigma$ onward the two systems are in the same state and generate the same continuation payoff. Therefore, the payoff under $\tau_1$ is almost surely no smaller than that under $\tau_2$. Taking expectations gives
\[
V(s,p_1)\ge J(s,p_1;\tau_1)\ge J(s,p_2;\tau_2)\ge V(s,p_2)-\varepsilon.
\]

Letting $\varepsilon\downarrow 0$ yields \(V(s,p_1)\ge V(s,p_2)\). Hence $V(s,p)$ is non-increasing in $p$. Using the same method, we also have $V(s,p_2)\ge V(s,p_1)- (p_2-p_1)$, hence $V(s,p)$ is 1-Lipschitz with respect to $p$.

Next fix $p$ and let $0\le s_1<s_2\le T$. Starting from state $(s_2,p)$, the agent can always ignore the last $s_2-s_1$ units of time-to-go (i.e., move the deadline forward) and then implement an optimal policy for state $(s_1,p)$. This shows \(V(s_2,p)\ge V(s_1,p)\). Moreover, the extra value from the additional interval of length $s_2-s_1$ can arise only if at least one event occurs during that interval, and the total payoff is bounded above by $v$. Hence
\[
0\le V(s_2,p)-V(s_1,p)
   \le v\,\mathbb{P}\left(N(s_2-s_1)\ge 1\right)
   = v\bigl(1-e^{-\lambda(s_2-s_1)}\bigr).
\]
Thus $V(\cdot,p)$ is non-decreasing and continuous in $s$. In particular, $m(s):=\mathbb{E}_H[V(s,P)]$ is also continuous in $s$.\hfill\qedsymbol
\end{proof}
\subsection{Proof of \Cref{prop:threshold}}
\begin{proof}{Proof.}
\medskip\noindent\textit{Part~\ref{prop:threshold_form}: Threshold structure.} For $p\in[0,v]$, define the option value of waiting by $W(s,p):=V(s,p)-(v-p)$. Because immediate purchase is always feasible, $W(s,p)\ge 0$ for all $p\in[0,v]$. In addition, for $0\le p_1\le p_2\le v$, we have $W(s,p_2)-W(s,p_1)=\bigl(V(s,p_2)-V(s,p_1)\bigr)+(p_2-p_1)\ge 0$, where the inequality follows from the $1$-Lipschitz property in \Cref{lemma:mono_V}. The same Lipschitz property gives continuity of $W(s,\cdot)$, so $W(s,\cdot)$ is continuous and non-decreasing on $[0,v]$.

At price $p=0$, the agent can secure a payoff $v$ by purchasing immediately, and no policy can generate more than $v$. Hence $V(s,0)=v$ and $W(s,0)=0$. By continuity and monotonicity of $W(s,\cdot)$, the set $\mathcal{S}_s:=\{p\in[0,v]:W(s,p)=0\}$ is a nonempty closed interval of the form $[0,b(s)]$ for some $b(s)\in[0,v]$. For $p>b(s)$ the inequality $W(s,p)>0$ is strict: if $W(s,p)=0$ for some $p>b(s)$, then $p\in\mathcal{S}_s$, contradicting $b(s)=\sup\mathcal{S}_s$. For $p>v$, immediate purchase yields $(v-p)^+=0$ while the option value of waiting is non-negative, so waiting is always weakly preferred and $b(s)\le v$ on the entire price space.

We adopt the convention of stopping at the threshold $p=b(s)$, so that at each state $(s,p)$ the optimal action is to purchase if and only if $p\le b(s)$. Equivalently, from elapsed time $T-s$, the optimal stopping rule is $\tau^\star=\inf\{t\in[T-s,T]:P_t\le b(T-t)\}$.

\medskip\noindent\textit{Part~\ref{prop:ODE}: Boundary identity and ODE.} Fix $s>0$. The variational inequality \eqref{eq:VI} gives the usual complementarity conditions for this obstacle problem. At a stopping point, where $V(s,p)=(v-p)^+$, the HJB residual of the immediate-purchase payoff must be non-positive. This residual is $\lambda m(s)-\lambda(v-p)^+$. Hence, any stopping point with $p\le v$ must satisfy $p\le v-m(s)$. In particular, since the boundary point is included in the stopping region by convention, $b(s)\le v-m(s)$. If, instead, no continuation point can lie strictly below $v-m(s)$. To see this, suppose that $p>b(s)$ and $p<v-m(s)$. Then $p\le v$ and $V(s,p)>v-p>m(s)$. At points where the continuation equation in \eqref{eq:VI} holds, $\partial_sV(s,p)=\lambda\bigl(m(s)-V(s,p)\bigr)<0$, which contradicts the monotonicity of $V$ in time-to-go established in \Cref{lemma:mono_V}. Thus, regular continuation points must satisfy $p\ge v-m(s)$, and the same inequality extends to all continuation points by continuity of $V$ and $m$. Taking $p\downarrow b(s)$ gives $b(s)\ge v-m(s)$. Combining the two inequalities yields the boundary identity
\begin{equation}\label{eq:boundary_identity_revised}
b(s)=v-m(s),\qquad s>0.
\end{equation}

We next derive the ODE. Since $V(\cdot,p)$ is Lipschitz in $s$ uniformly in $p$, differentiation under the expectation is valid for almost every $s>0$, and $m'(s)=\mathbb{E}_H[\partial_s V(s,P)]$. On the stopping region $\{P\le b(s)\}$, we have $V(s,P)=v-P$, so $\partial_s V(s,P)=0$. On the continuation region $\{P>b(s)\}$, the variational inequality \eqref{eq:VI} implies $\partial_s V(s,P)=\lambda\bigl(m(s)-V(s,P)\bigr)$ for a.e. $s>0$. Therefore $m'(s)=\lambda\,\mathbb{E}_H[(m(s)-V(s,P))\mathbf{1}_{\{P>b(s)\}}]$.

On $\{P\le b(s)\}$, the stopping rule is optimal, so $V(s,P)=v-P=m(s)+b(s)-P$, where the second equality uses \eqref{eq:boundary_identity_revised}. Hence $\mathbb{E}_H[V(s,P)\mathbf{1}_{\{P\le b(s)\}}]=m(s)\,H(b(s))+\mathbb{E}_H[(b(s)-P)^+]$. Since $m(s)=\mathbb{E}_H[V(s,P)]$, it follows that $\mathbb{E}_H[V(s,P)\mathbf{1}_{\{P>b(s)\}}]=m(s)(1-H(b(s)))-\mathbb{E}_H[(b(s)-P)^+]$. Substituting this expression into the expression for $m'(s)$ and using $b(s)=v-m(s)$ gives
\[
b'(s)
=
-\lambda\,\mathbb{E}_H\left[(b(s)-P)^+\right]
=
-\lambda\int_0^{b(s)}(b(s)-p)\,H(\mathrm{d}p).
\]

Finally, the terminal condition is immediate. At $s=0$, $V(0,p)=(v-p)^+$, so $m(0)=\mathbb{E}_H[(v-P)^+]$. Therefore,
\[
\lim_{s\downarrow 0}b(s)
=
v-m(0)
=
v-\mathbb{E}_H[(v-P)^+]
=
\mathbb{E}_H[\min\{v,P\}].
\]

At the exact deadline, the agent purchases if and only if $p\le v$, so $b(0)=v$.

\medskip\noindent\textit{Part~\ref{prop:mono_convex}: Monotonicity and convexity of $b$.} Define $G_H(x):=\mathbb{E}_H[(x-P)^+]$ for $x\ge 0$. Then $G_H(x)\ge 0$ for all $x$, and the ODE established above can be written as $b'(s)=-\lambda G_H(b(s))$ for a.e. $s>0$. Since $G_H\ge 0$, we have $b'(s)\le 0$ almost everywhere, so $b$ is non-increasing.

Moreover, $G_H$ is non-decreasing in $x$. Because $b$ is non-increasing in $s$, the composition $s\mapsto G_H(b(s))$ is non-increasing. Hence $b'(s)=-\lambda G_H(b(s))$ is non-decreasing almost everywhere. This implies that $b$ is convex on $(0,T]$. \hfill\qedsymbol
\end{proof}
\subsection{Proof of \Cref{prop:stationary_monotonicity}}
\begin{proof}{Proof of \Cref{prop:stationary_monotonicity}.}
For a price-adjustment distribution $H$, define $G_H(x):=\mathbb{E}_H[(x-P)^+]$. The map $G_H$ is non-negative, non-decreasing, and $1$-Lipschitz in $x$, so the ODE characterizing $b$ in \Cref{prop:threshold} has a unique absolutely continuous solution given the right-limit initial condition. At the exact deadline, $b(0)=v$ regardless of $\lambda$ and $H$, so Parts (i) and (ii) hold at $s=0$ with equality; Part (iii) holds at $s=0$ because $b_i(0)=v_i$ and $v_1\le v_2$. It remains to be verified that the ordering holds for $s>0$. For notational convenience, we write $b_i$ for the optimal threshold corresponding to parameter tuple~$i$. \textit{Part (i).} Fix $H$ and $v$. By \Cref{prop:threshold}, for each arrival rate $\lambda$, the threshold solves
\[
b'(s)=-\lambda\,G_H(b(s)), \qquad s>0,
\]
with the same right-limit initial condition
\[
\lim_{s\downarrow 0}b(s)=\mathbb{E}_H[\min\{v,P\}].
\]
Thus, changing $\lambda$ does not change the initial level of the threshold; it only changes the speed at which the same autonomous dynamics evolves. More precisely, let $\beta$ solve
\[
\beta'(u)=-G_H(\beta(u)), 
\qquad 
\beta(0)=\mathbb{E}_H[\min\{v,P\}].
\]
Then $b_i(s)=\beta(\lambda_i s)$ for $s>0$. Since $G_H\ge 0$, the function $\beta$ is non-increasing. Therefore, if $\lambda_1\le \lambda_2$, then for every $s>0$,
\[
b_{2}(s)=\beta(\lambda_2 s)\le \beta(\lambda_1 s)=b_{1}(s).
\]

\medskip\noindent\textit{Part (ii).} Fix $\lambda$ and $v$. From the ODE characterization, we have
\[
b_i'(s)=-\lambda\,G_{H_i}(b_i(s)), \qquad s>0,
\]
and
\[
\lim_{s\downarrow 0} b_i(s)=\mathbb{E}_{H_i}[\min\{v,P\}], \qquad i=1,2.
\]
By the expectation characterization of FSD \citep[Section~1.A.1]{shaked2007stochastic}, $H_1\preceq_{\mathrm{FSD}}H_2$ implies that, for every increasing function $\phi$ for which the expectations exist, one has $\mathbb{E}_{H_1}[\phi(P)]\le \mathbb{E}_{H_2}[\phi(P)]$. Applying this first to $\phi(p)=\min\{v,p\}$, which is increasing in $p$, gives
\[
\lim_{s\downarrow 0} b_1(s)=\mathbb{E}_{H_1}[\min\{v,P\}]
\le
\mathbb{E}_{H_2}[\min\{v,P\}]
=
\lim_{s\downarrow 0} b_2(s).
\]

Next, fix any $x\ge 0$. Since $p\mapsto (x-p)^+$ is decreasing in $p$, monotonicity under decreasing transformations and the expectation characterization of FSD imply $G_{H_1}(x)\ge G_{H_2}(x)$ \citep[Theorem~1.A.3(a) and Section~1.A.1]{shaked2007stochastic}. We now show that $b_1(s)\le b_2(s)$ for all $s$. If this were not the case, then by continuity, there exists a first time
\[
s_*:=\inf\{s>0: b_1(s)>b_2(s)\}.
\]
At this time, $b_1(s_*)=b_2(s_*)$. Moreover, for $s>s_*$ sufficiently close to $s_*$, one has $b_1(s)\ge b_2(s)$. On such an interval,
\[
b_1'(s)-b_2'(s)
=
-\lambda\,G_{H_1}(b_1(s))
+\lambda\,G_{H_2}(b_2(s)).
\]
Because $b_1(s)\ge b_2(s)$, $G_{H_2}$ is non-decreasing, and $G_{H_1}\ge G_{H_2}$ pointwise, we obtain $b_1'(s)-b_2'(s)\le 0$. This contradicts the first crossing of $b_1$ above $b_2$. Therefore $b_1(s)\le b_2(s)$ for all $s>0$.

\medskip\noindent\textit{Part (iii).} Fix $\lambda$ and $H$. The two thresholds satisfy the same ODE,
\[
b_i'(s)=-\lambda\,G_H(b_i(s)), \qquad s>0,
\]
but have different right-limit initial conditions,
\[
\lim_{s\downarrow 0} b_i(s)=\mathbb{E}_{H}[\min\{v_i,P\}], \qquad i=1,2.
\]
Since $v_1\le v_2$ and $v\mapsto \min\{v,p\}$ is non-decreasing for each fixed $p$, it follows that
\[
\lim_{s\downarrow 0} b_1(s)=\mathbb{E}_{H}[\min\{v_1,P\}]
\le
\mathbb{E}_{H}[\min\{v_2,P\}]
=
\lim_{s\downarrow 0} b_2(s).
\]

Suppose, for contradiction, that $b_1(s)>b_2(s)$ for some $s>0$, and let
\[
s_*:=\inf\{s>0: b_1(s)>b_2(s)\}.
\]
Then $b_1(s_*)=b_2(s_*)$, and for $s>s_*$ sufficiently close to $s_*$, one has $b_1(s)\ge b_2(s)$. On such an interval,
\[
b_1'(s)-b_2'(s)
=
-\lambda\,G_H(b_1(s))
+\lambda\,G_H(b_2(s))
\le 0,
\]
because $G_H$ is non-decreasing. This contradicts the first crossing of $b_1$ above $b_2$. Therefore $b_1(s)\le b_2(s)$ for all $s>0$. \hfill\qedsymbol
\end{proof}
\subsection{Proof of \Cref{prop:bayes_threshold}}
\begin{proof}{Proof of \Cref{prop:bayes_threshold}.}
Given a state $(p,\pi)$, define $M(s,\pi):=\E_{P\sim \bar{H}_\pi} \left[V^B\bigl(s,P,\Phi(\pi,P)\bigr)\right]$. Here, $M(s,\pi)$ is the continuation value immediately after a price adjustment when the time-to-go is $s$, before observing the realized new posted price. 

Because the adjustment process is Poisson with rate $\lambda$, the current posted price remains equal to $p$ until the first adjustment arrives. Moreover, no new information is revealed before that first adjustment. Hence, prior to the first adjustment, any admissible policy is equivalent to a two-step rule for some deterministic waiting duration $w\in[0,s]$. The policy waits for $w$ units of time. If no adjustment has occurred by then, it purchases at price $p$ when the time-to-go has fallen to $s-w$. If an adjustment occurs earlier, it switches to an optimal policy from the post-adjustment state.

Let $J_{s,w}(p,\pi)$ denote the value of this policy. Conditioning on the first adjustment time yields
  \[
    J_{s,w}(p,\pi)
    =
    e^{-\lambda w}(v-p)^+
    +
    \int_0^w \lambda e^{-\lambda z} M(s-z,\pi)\,dz .
  \]
Therefore,
  \[
    V^B(s,p,\pi)=\sup_{w\in[0,s]}J_{s,w}(p,\pi).
  \]

For each fixed $w$, the only term in $J_{s,w}(p,\pi)$ that depends on $p$ is $e^{-\lambda w}(v-p)^+$, which is non-increasing in $p$. Hence $p\mapsto J_{s,w}(p,\pi)$ is non-increasing for every $w$, and taking the supremum over $w$ shows that $p\mapsto V^B(s,p,\pi)$ is non-increasing.

Immediate purchase corresponds to $w=0$, for which
  \[
    J_{s,0}(p,\pi)=(v-p)^+.
  \]
For any fixed $w\in(0,s]$, write
  \[
    A_{s,w}(\pi):=\int_0^w \lambda e^{-\lambda z} M(s-z,\pi)\,dz,
  \]
which is independent of $p$. Then
  \[
    J_{s,w}(p,\pi)=e^{-\lambda w}(v-p)^+ + A_{s,w}(\pi).
  \]

If $p\le v$, then
  \[
    J_{s,0}(p,\pi)-J_{s,w}(p,\pi)
    =
    (1-e^{-\lambda w})(v-p)-A_{s,w}(\pi),
  \]
which is non-increasing in $p$. Therefore, if immediate purchase weakly dominates waiting time $w$ at some price $p_0$, then it also weakly dominates waiting time $w$ at every lower price $p\le p_0$.

Restrict attention first to prices $p\in[0,v]$. For each fixed $(s,\pi)$ and each $w>0$, the set $\{p\in[0,v]:J_{s,0}(p,\pi)\ge J_{s,w}(p,\pi)\}$ is therefore a lower interval. It is nonempty because immediate purchase at $p=0$ attains the maximal possible surplus $v$. Since immediate purchase is optimal exactly when it weakly dominates every waiting time $w\in(0,s]$, the set of prices in $[0,v]$ at which buying is optimal is the intersection of lower intervals and hence is itself a lower interval. Let $b(s,\pi)\in[0,v]$ denote its upper endpoint.

For prices $p>v$, immediate purchase yields zero surplus. Waiting is weakly optimal because it also yields a nonnegative payoff, and ties can be resolved in favor of purchasing. Thus, one can select an optimal policy that buys for $p\le b(s,\pi)$ and waits for $p>b(s,\pi)$. \hfill\qedsymbol
\end{proof}

\subsection{Proof of \Cref{prop:bayes_param_monotonicity}}
\begin{proof}{Proof of \Cref{prop:bayes_param_monotonicity}.}
We give the argument briefly because it follows the same comparison logic as \Cref{prop:stationary_monotonicity}. Fix a belief $\pi$ and valuation $v$. The arrival rate affects the Bayesian dynamic program only through the Poisson process. After the time change $\tilde{s}=\lambda s$, the problem with rate $\lambda$ and time-to-go $s$ is equivalent to a unit-rate problem with effective time-to-go $\tilde{s}$, with the same posterior update rule and predictive price-adjustment distributions. Hence
\[
  b(s,\pi;\lambda,v)=\tilde b(\lambda s,\pi;v),
\]
where $\tilde b$ denotes the threshold in the unit-rate formulation.

It remains to note that the Bayesian stopping boundary is non-increasing in effective time-to-go. A longer time horizon weakly expands the set of admissible waiting policies while leaving immediate purchase feasible. Thus, the option value of waiting is weakly larger at every current price. Since the stopping region is a lower interval by \Cref{prop:bayes_threshold}, the upper endpoint of this interval cannot increase with the time-to-go. Therefore, if $\lambda_1\le\lambda_2$, then
\[
  b(s,\pi;\lambda_2,v)
  =
  \tilde b(\lambda_2 s,\pi;v)
  \le
  \tilde b(\lambda_1 s,\pi;v)
  =
  b(s,\pi;\lambda_1,v).
\]

For valuation monotonicity, fix $\lambda$ and $\pi$, and let $v_1\le v_2$.
Write $V_i$ and $b_i$ for the Bayesian value function and threshold under valuation $v_i$, and set $\Delta:=v_2-v_1$. The posterior dynamics, adjustment process, and admissible stopping rules are identical under the two valuations; only the purchase payoff changes. For any admissible stopping rule~$\tau$, the pathwise bound \((v_2-P_\tau)^+\le (v_1-P_\tau)^++\Delta\) implies, after taking expectations and optimizing over stopping rules, that
\[
  V_2(s,p,\pi)\le V_1(s,p,\pi)+\Delta .
\]
Now take any price \(p\le b_1(s,\pi)\). By \Cref{prop:bayes_threshold}, purchase
is optimal under valuation \(v_1\); since \(b_1(s,\pi)\le v_1\), this gives
\(V_1(s,p,\pi)=v_1-p\). The preceding inequality then yields
\(V_2(s,p,\pi)\le v_2-p\). Immediate purchase is feasible under valuation \(v_2\), so \(V_2(s,p,\pi)\ge v_2-p\). Hence equality holds, and \(p\) is also in the stopping region under valuation \(v_2\). Therefore the stopping region under \(v_1\) is contained in the stopping region under \(v_2\), which implies
\[
  b(s,\pi;\lambda,v_1)\le b(s,\pi;\lambda,v_2)
  \qquad\text{for all }s.
\]
\hfill\qedsymbol
\end{proof}

\subsection{Proof of \Cref{lem:belief_order}}
\begin{proof}{Proof of \Cref{lem:belief_order}.}
\medskip\noindent\textit{Part~\ref{lem:bo_predictive}.} By \Cref{asm:mlrp} and the fact that MLR order implies FSD \citep[Theorem~1.C.1]{shaked2007stochastic}, $\theta_1\le \theta_2$ implies $H_{\theta_1}\preceq_{\mathrm{FSD}}H_{\theta_2}$. Let $\phi$ be any increasing function for which the expectations below exist, and define $m_\phi(\theta):=\mathbb E_{P\sim H_\theta}[\phi(P)]$. By the expectation characterization of FSD \citep[Section~1.A.1]{shaked2007stochastic}, the preceding implication makes $m_\phi$ increasing in $\theta$. Therefore, by the same characterization, if $\pi_1\preceq_{\mathrm{FSD}}\pi_2$, then \( \mathbb E_{P\sim\bar H_{\pi_1}}[\phi(P)] =\int_\Theta m_\phi(\theta)\,\pi_1(d\theta) \le \int_\Theta m_\phi(\theta)\,\pi_2(d\theta) =\mathbb E_{P\sim\bar H_{\pi_2}}[\phi(P)]. \) This is exactly $\bar H_{\pi_1}\preceq_{\mathrm{FSD}}\bar H_{\pi_2}$.

\medskip\noindent\textit{Part~\ref{lem:bo_update}.} Let $q_i$ denote densities of $\pi_i$ in the MLR definition, and fix an observed price $p$ for which both posteriors are defined. The posterior density is proportional to $h_\theta(p)q_i(\theta)$. For $\theta_1<\theta_2$, the MLR ordering $\pi_1\preceq_{\mathrm{MLR}}\pi_2$ gives $q_2(\theta_2)q_1(\theta_1)\ge q_2(\theta_1)q_1(\theta_2)$. Multiplying both sides by the common nonnegative factor $h_{\theta_2}(p)h_{\theta_1}(p)$ and dividing by the positive product $Z_1(p)Z_2(p)$ yields
\[
\frac{h_{\theta_2}(p)q_2(\theta_2)}{Z_2(p)}
\frac{h_{\theta_1}(p)q_1(\theta_1)}{Z_1(p)}
\ge
\frac{h_{\theta_1}(p)q_2(\theta_1)}{Z_2(p)}
\frac{h_{\theta_2}(p)q_1(\theta_2)}{Z_1(p)},
\]
where $Z_i(p)$ is the normalizing constant in Bayes' rule for $\Phi(\pi_i,p)$. This is the cross-product condition for $\Phi(\pi_1,p)\preceq_{\mathrm{MLR}}\Phi(\pi_2,p)$.

\medskip\noindent\textit{Part~\ref{lem:bo_signal}.} Fix prior densities $q_1,q_2$ with $\pi_1\preceq_{\mathrm{MLR}}\pi_2$ and prices $p_1\le p_2$ for which the two posteriors are defined. For $\theta_1<\theta_2$, Assumption~\ref{asm:mlrp} gives
\[
  h_{\theta_2}(p_2)h_{\theta_1}(p_1)
  \ge
  h_{\theta_2}(p_1)h_{\theta_1}(p_2).
\]
The prior ordering gives
\[
  q_2(\theta_2)q_1(\theta_1)
  \ge
  q_2(\theta_1)q_1(\theta_2).
\]
Multiplying the two inequalities and dividing by the positive posterior normalizing constants yields
\[
\frac{h_{\theta_2}(p_2)q_2(\theta_2)}{Z_2(p_2)}
\frac{h_{\theta_1}(p_1)q_1(\theta_1)}{Z_1(p_1)}
\ge
\frac{h_{\theta_1}(p_2)q_2(\theta_1)}{Z_2(p_2)}
\frac{h_{\theta_2}(p_1)q_1(\theta_2)}{Z_1(p_1)}.
\]
Here $Z_i(p_i)$ is the normalizing constant for $\Phi(\pi_i,p_i)$. The displayed inequality is the cross-product condition for $\Phi(\pi_1,p_1)\preceq_{\mathrm{MLR}}\Phi(\pi_2,p_2)$. Taking $\pi_1=\pi_2=\pi$ gives the stated monotonicity in the observed signal. This proves all three claims. \hfill\qedsymbol
\end{proof}

\subsection{Proof of \Cref{prop:bayes_belief_monotonicity}}
\begin{proof}{Proof of \Cref{prop:bayes_belief_monotonicity}.}
We use a finite-adjustment approximation and then pass to the original Poisson model. For each $n\ge 0$, let $V_n(s,p,\pi)$ be the value when the agent can use at most $n$ future price adjustments before the deadline. If no future adjustment can be used, the agent must decide immediately, so $V_0(s,p,\pi):=(v-p)^+$. For $n+1$ available adjustments,
\[
V_{n+1}(s,p,\pi)
=
\sup_{0\leq t\leq s}\left\{
e^{-\lambda t}(v-p)^+
+
\int_0^{t} \lambda e^{-\lambda u}\,
\E_{P\sim \bar H_\pi}
\!\left[
V_n\bigl(s-u,P,\Phi(\pi,P)\bigr)
\right]\dd u
\right\}.
\]

The sequence is non-decreasing in $n$, and the Poisson process generates only finitely many adjustments before the deadline almost surely. Hence $V_n(s,p,\pi)\uparrow V^B(s,p,\pi)$ pointwise as $n\to\infty$.

The key step is an induction showing that two monotonicity properties are preserved by the dynamic program. For every $n$,

\begin{enumerate}[label=(\alph*),ref=(\alph*)]
\item\label{prop:bayes_belief_ind_price} for each fixed $(s,\pi)$, the map $p\mapsto V_n(s,p,\pi)$ is non-increasing;
\item\label{prop:bayes_belief_ind_belief} if $\pi_1\preceq_{\mathrm{MLR}}\pi_2$, then
\[
V_n(s,p,\pi_1)\ge V_n(s,p,\pi_2)
\qquad\text{for all }(s,p).
\]
\end{enumerate}

The base case is immediate because $V_0(s,p,\pi)=(v-p)^+$ is non-increasing in $p$ and does not depend on $\pi$. Assume the two properties hold for some $n\ge0$.

To prove price monotonicity for $V_{n+1}$, fix $(s,\pi)$ and write, for each $t\in[0,s]$,
\[
J_n(t;s,p,\pi)
:=
e^{-\lambda t}(v-p)^+
+
\int_0^t \lambda e^{-\lambda u}\,
\E_{P\sim \bar H_\pi}\!\left[
V_n\bigl(s-u,P,\Phi(\pi,P)\bigr)
\right]\dd u .
\]
Since $V_{n+1}(s,p,\pi)=\sup_{0\le t\le s}J_n(t;s,p,\pi)$ and the current price enters $J_n$ only through the non-increasing term $e^{-\lambda t}(v-p)^+$, the supremum is also non-increasing in $p$. This proves Part~\ref{prop:bayes_belief_ind_price} for $n+1$.

It remains to prove belief monotonicity. Fix $\pi_1\preceq_{\mathrm{MLR}}\pi_2$. For each $u\in[0,s]$, let
\[
D_i:=\left\{p:\int_\Theta h_\theta(p)\pi_i(d\theta)>0\right\},
\qquad i=1,2,
\]
be the positive predictive support under belief $\pi_i$. On $D_i$, define
\[
g_i(u,p):=
V_n\bigl(s-u,p,\Phi(\pi_i,p)\bigr),
\qquad i=1,2.
\]

The continuation comparison has two moving parts: the predictive distribution is worse under $\pi_2$, and the posterior after a given signal is also worse. We map these effects to a common price domain using the following envelope. For each fixed $u$, the function $g_2(u,\cdot)$ admits a bounded non-increasing extension $\widetilde g_2(u,\cdot)$ to the full price space such that
\[
\widetilde g_2(u,p)=g_2(u,p)\quad \bar H_{\pi_2}\text{-a.s.},
\qquad
g_1(u,p)\ge \widetilde g_2(u,p)\quad \bar H_{\pi_1}\text{-a.s.}
\]
Indeed, $g_2(u,\cdot)$ is non-increasing on $D_2$: if $p_1,p_2\in D_2$ and $p_1\le p_2$, then Part~\ref{lem:bo_signal} gives $\Phi(\pi_2,p_1)\preceq_{\mathrm{MLR}}\Phi(\pi_2,p_2)$, and the induction hypothesis applied first to beliefs and then to prices gives $g_2(u,p_1)\ge g_2(u,p_2)$. Define
\[
\widetilde g_2(u,p)
:=
\begin{cases}
g_2(u,p),
& p\in D_2,\\[0.6ex]
\displaystyle \sup\left\{g_2(u,q):q\in D_2,\ q\ge p\right\},
& p\notin D_2 \text{ and } \{q\in D_2:q\ge p\}\neq\emptyset,\\[1.2ex]
0,
& \{q\in D_2:q\ge p\}=\emptyset.
\end{cases}
\]
This construction preserves $g_2$ on $D_2$, fills support gaps by the upper envelope generated by $\pi_2$-possible prices, and assigns the payoff lower bound above the $\pi_2$ predictive support. Hence $\widetilde g_2$ is bounded, non-increasing, and equal to $g_2$ $\bar H_{\pi_2}$-a.s. Moreover, if $p\in D_1$ and $q\in D_2$ with $q\ge p$, then Part~\ref{lem:bo_signal} gives $\Phi(\pi_1,p)\preceq_{\mathrm{MLR}}\Phi(\pi_2,q)$, so the induction hypothesis yields
\[
g_1(u,p)
=
V_n\bigl(s-u,p,\Phi(\pi_1,p)\bigr)
\ge
V_n\bigl(s-u,p,\Phi(\pi_2,q)\bigr)
\ge
V_n\bigl(s-u,q,\Phi(\pi_2,q)\bigr)
=
g_2(u,q).
\]
Taking the supremum over such $q$ proves the desired envelope inequality in the first two cases of the definition; in the last case, $\widetilde g_2(u,p)=0$ and the inequality follows from nonnegativity.

We can now compare the continuation expectations. Since $\pi_1\preceq_{\mathrm{MLR}}\pi_2$ implies $\pi_1\preceq_{\mathrm{FSD}}\pi_2$ \citep[Theorem~1.C.1]{shaked2007stochastic}, \Cref{lem:bo_predictive} yields
\(
\bar H_{\pi_1}\preceq_{\mathrm{FSD}}\bar H_{\pi_2}
\).

For each fixed $u$, we have
\[
\E_{P\sim \bar H_{\pi_1}}[g_1(u,P)]
\ge
\E_{P\sim \bar H_{\pi_1}}[\widetilde g_2(u,P)]
\ge
\E_{P\sim \bar H_{\pi_2}}[\widetilde g_2(u,P)]
=
\E_{P\sim \bar H_{\pi_2}}[g_2(u,P)].
\]
The first inequality follows from $g_1(u,p)\ge \widetilde g_2(u,p)$ $\bar H_{\pi_1}$-a.s. The second follows from $\bar H_{\pi_1}\preceq_{\mathrm{FSD}}\bar H_{\pi_2}$ and the fact that $\widetilde g_2(u,\cdot)$ is non-increasing, using monotonicity under decreasing transformations and the expectation characterization of FSD \citep[Theorem~1.A.3(a) and Section~1.A.1]{shaked2007stochastic}. The equality follows from $\widetilde g_2=g_2$ $\bar H_{\pi_2}$-a.s.

Thus, for every waiting time $t$,
\(
J_n(t;s,p,\pi_1)\ge J_n(t;s,p,\pi_2)
\) for all $t\in[0,s]$. Taking the supremum over $t$ gives
\[
V_{n+1}(s,p,\pi_1)\ge V_{n+1}(s,p,\pi_2).
\]
Thus Part~\ref{prop:bayes_belief_ind_belief} also holds for $n+1$.

By induction, both properties hold for all $n\ge 0$. Passing to the pointwise limit
\(
V_n(s,p,\pi)\uparrow V^B(s,p,\pi)
\)
preserves the belief comparison, so
\[
V^B(s,p,\pi_1)\ge V^B(s,p,\pi_2)
\qquad\text{for all }(s,p).
\]

Finally, translate the value comparison into a threshold comparison. By \Cref{prop:bayes_threshold}, for each $(s,\pi)$ the stopping region is of the form $\{p:p\le b(s,\pi)\}$. If $p$ is in the stopping region under $\pi_1$, then $V^B(s,p,\pi_1)=(v-p)^+$. The value comparison gives $V^B(s,p,\pi_2)\le V^B(s,p,\pi_1)$, while immediate purchase is feasible under $\pi_2$, so $V^B(s,p,\pi_2)\ge (v-p)^+$. Hence $V^B(s,p,\pi_2)=(v-p)^+$, and $p$ is also in the stopping region under $\pi_2$. Therefore the stopping region under $\pi_1$ is contained in the stopping region under $\pi_2$, and
\[
b(s,\pi_1)\le b(s,\pi_2)
\qquad\text{for all }s.
\]
This proves both claims. \hfill\qedsymbol
\end{proof}

\subsection{Proofs for Information Gap}

We organize the argument in two layers. The short-window bound follows directly from a deadline-purchase comparison. For longer windows, we compare the oracle with a feasible learn-then-act policy: the loss first splits into a learning-delay term and an acting-phase term; the learning-delay term is controlled by the lower tail of the price distribution near its support endpoint; and the acting-phase term is controlled by posterior concentration together with Lipschitz continuity of the oracle threshold. The final proof chooses the learning length and confidence level to balance these bounds.

\subsubsection{Short-Window Bound.}

\begin{proof}{Proof of \Cref{lem:short_window_gap}.}
The lower bound is the value of information. For the upper bound, compare the oracle with the feasible Bayesian policy that waits until the deadline and then purchases if the posted price is smaller than or equal to the consumer's valuation. Let $N_T$ denote the number of adjustments during the shopping window. On $\{N_T=0\}$, the price never changes, so the deadline policy attains the same payoff as the oracle, $(v-P_T)^+$. On $\{N_T\ge1\}$, the oracle's advantage over this feasible policy is at most $v$. Since the Bayesian value is at least the payoff from the deadline policy,
\[
  \mathcal G(T,\hat\pi)
  \le
  v\,\mathbb P(N_T\ge1)
  =
  v(1-e^{-\lambda T})
  \le
  v\lambda T,
\]
where the last inequality uses $1-e^{-x}\le x$. \hfill\qedsymbol
\end{proof}

\subsubsection{ Long-Window Setup: Posterior Confidence Set and LTA Policy.}

For the long-window bound, we first construct the posterior confidence set used by the learn-then-act policy.
For a realized price history $x_{1:n}$, the posterior effective support is
\[
  \Theta_n(x_{1:n})
  :=
  \{\theta\in\Theta:\hat\pi(\theta)\prod_{i=1}^n h_\theta(x_i)>0\}.
\]
The posterior-regularity condition in \Cref{ass:regular_family} implies that, after any $n$ observed prices, the posterior is $(a_0{+}a_1 n)$-strongly log-concave on its effective support, which is an interval. Let $\bar\theta_n$ denote the posterior mean and, for $\delta\in(0,1)$, define
\[
  d_n(\delta):=
  \sqrt{\frac{2\log(2/\delta)}{a_0+a_1 n}},
  \qquad
  I_n(\delta):=
  [\bar\theta_n-d_n(\delta),\bar\theta_n+d_n(\delta)]\cap\Theta.
\]

The proof of \Cref{prop:info_gap_confidence} compares the oracle's expected payoff with that of a feasible \emph{learn-then-act} (LTA) policy. Fix a split time $\ell\in(0,T)$. We use elapsed-time notation: $P_0$ is the entry price and $P_\ell$ is the price at the end of the learning phase. The LTA policy operates in two phases: 
\begin{enumerate}[label=(\alph*)]
  \item \emph{Learning phase} (elapsed time $[0,\ell]$): starting from the
posterior $\Phi(\hat\pi,P_0)$  after observing the entry price, observe prices without purchasing. After $\ell$ units of time, the agent has seen $n_\ell:=1+N_\ell$ prices, where $N_\ell\sim\mathrm{Poisson}(\lambda\ell)$, and formed the posterior mean $\bar\theta_{n_\ell}$.
  \item \emph{Acting phase} (elapsed time $[\ell,T]$): construct the
posterior confidence interval $I_{n_\ell}(\delta)$ and follow the upper-envelope threshold
    \[
      \bar b_{I_{n_\ell}(\delta)}(s)
      :=
      \sup_{\eta\in I_{n_\ell}(\delta)} b^\eta(s),
      \qquad s\in[0,T-\ell].
    \]
At elapsed time $t\in[\ell,T]$, the policy stops whenever the current price satisfies $P_t\le \bar b_{I_{n_\ell}(\delta)}(T-t)$. 
\end{enumerate}

Since $V^{\mathrm B}$ is the value of the optimal Bayesian policy from the entry state $(T,P_0,\Phi(\hat\pi,P_0))$, it is at least as large as the expected payoff of any feasible policy from the same state, including the LTA policy. Hence,  
\begin{equation*}
  \mathcal G(T,\hat\pi)
  \;\le\;
  \E\bigl[V^{\tilde{\theta}}(T,P_0)\bigr] 
  -
  \E\bigl[\text{LTA payoff}\bigr].
\end{equation*}
We decompose the right-hand side into two terms:
\begin{equation}\label{eq:gap_two_terms}
  \mathcal G(T,\hat\pi)
  \;\le\;
  \underbrace{%
    \E\bigl[V^{\tilde{\theta}}(T,P_0)-V^{\tilde{\theta}}(T{-}\ell,P_\ell)\bigr]}
  _{\mathcal E(\ell)}
  \;+\;
  \underbrace{%
    \E\bigl[V^{\tilde{\theta}}(T{-}\ell,P_\ell)
    -J_{T-\ell}^{\tilde{\theta}}\bigl(P_\ell,\,
    \tau_{\bar b_{I_{n_\ell}(\delta)}}\bigr)\bigr]
    }
  _{\mathcal R(\ell)},
\end{equation}
where $J_s^\theta(p,\tau)$ is the expected payoff of stopping rule~$\tau$ under $\theta$-dynamics from a state with time-to-go $s$ and current price $p$, and $\tau_{\bar b_{I_{n_\ell}(\delta)}}$ denotes the stopping rule using the confidence-envelope threshold.  The first term~$\mathcal E(\ell)$ is the \emph{learning-delay gap}; the second term~$\mathcal R(\ell)$ is the \emph{acting-phase regret}.

\subsubsection{Learning-Delay Bound.}

\begin{lemma}[{\sc Learning Delay}]\label{lem:early_purchase_rate}
Under Assumption~\ref{ass:regular_family}(i), there exists $C_{\mathrm E}<\infty$ such that for every $\ell\in[0,T)$,
  \[
    \mathcal E(\ell)
    \le
    C_{\mathrm E}(1+T-\ell)^{-1/\alpha}.
  \]
\end{lemma}

\begin{proof}{Proof of \Cref{lem:early_purchase_rate}.}
Fix $\ell<T$. Conditional on $\tilde{\theta}=\theta$, both $P_0$ and $P_\ell$ have marginal~$H_\theta$ (the initial price is drawn from~$H_\theta$, and at any later time the current price is either unchanged or the most recent adjustment, both distributed as~$H_\theta$). Hence, 
  \begin{align*}
\mathcal E(\ell) &= \E\bigl[m^{\tilde{\theta}}(T)-m^{\tilde{\theta}}(T{-}\ell)\bigr] = \E\bigl[b^{\tilde{\theta}}(T{-}\ell)-b^{\tilde{\theta}}(T)\bigr],
  \end{align*}
using $m^\theta(s)=\E_{H_\theta}[V^\theta(s,P)]=v-b^\theta(s)$ from \Cref{prop:threshold}.

The threshold remains above the lower endpoint $p_L^\theta$. Indeed, $b^\theta(0^+)=\E_{H_\theta}[\min\{v,P\}]\ge p_L^\theta$, and the ODE has zero drift at and below $p_L^\theta$ because $H_\theta([0,p_L^\theta))=0$. Thus, the solution cannot cross below $p_L^\theta$. Define $y_\theta(s):=b^\theta(s)-p_L^{\theta}\ge0$. Let $c>0$, $\alpha>0$, and $\varepsilon>0$ be the common constants in
Assumption~\ref{ass:regular_family}, the ODE from \Cref{prop:threshold} gives, for a.e. $s>0$,
  \[
    y_\theta'(s)
    =
    -\lambda\,\E_{H_\theta}\bigl[(p_L^{\theta}{+}y_\theta(s)-P)^+\bigr]
    \le
-\lambda\,\int_0^{\min\{y,\varepsilon\}} c u^\alpha\,du
    =
    -\frac{\lambda c}{\alpha+1}\,
    \min\bigl\{y_\theta(s)^{\alpha+1},\,\varepsilon^{\alpha+1}\bigr\}.
  \]
The inequality follows by writing $\E_{H_\theta}[(p_L^\theta+y-P)^+]=\int_0^y H_\theta([p_L^\theta,p_L^\theta+u])\,du$ and applying the lower-tail condition for $u\le\varepsilon$. Let $K:=\lambda c/(\alpha+1)$. When $y_\theta(s)>\varepsilon$, the preceding inequality implies $y_\theta'(s)\le -K\varepsilon^{\alpha+1}$. Since $p_L^\theta\le b^\theta(s)\le v$, the hitting time of $[0,\varepsilon]$ is bounded above by the common constant $\tau_0:=(v-\varepsilon)^+/(K\varepsilon^{\alpha+1})$. Let $\tau_\theta$ denote this hitting time. For $s\ge\tau_\theta$, comparison with the solution of $z'=-Kz^{\alpha+1}$ starting from $z(\tau_\theta)=\varepsilon$ gives
  \[
    y_\theta(s)
    \le
    \left(\varepsilon^{-\alpha}+\alpha K(s-\tau_\theta)\right)^{-1/\alpha}.
  \]
If $s\ge 2\tau_0+1$, then $s-\tau_\theta\ge (1+s)/2$, so $y_\theta(s)$ is bounded by $(2/(\alpha K))^{1/\alpha}(1+s)^{-1/\alpha}$. If $s<2\tau_0+1$, then $y_\theta(s)\le v\le v(2\tau_0+2)^{1/\alpha}(1+s)^{-1/\alpha}$. Taking
  \[
    C_{\mathrm E}\ge
    \max\left\{
      \left(\frac{2}{\alpha K}\right)^{1/\alpha},
      v(2\tau_0+2)^{1/\alpha}
    \right\}
  \]
yields $y_\theta(s)\le C_{\mathrm E}(1+s)^{-1/\alpha}$ uniformly in~$\theta$ and $s$. Since $b^\theta(T)\ge p_L^\theta$,
  \[
    b^\theta(T{-}\ell)-b^\theta(T)
    \leq b^\theta(T-\ell)-p_L^{\theta}
    =
    y_\theta(T-\ell)
    \leq C_{\mathrm E}(1+T-\ell)^{-1/\alpha}.
  \]
Taking expectations over $\Theta$ completes the proof. \hfill\qedsymbol
\end{proof}

\subsubsection{Posterior Concentration Bound.}

\begin{lemma}[{\sc Brascamp--Lieb Variance Bound}]\label{lem:brascamp_lieb_variance}
Let $\mu$ be a probability measure on an interval $I$ with density proportional to $\exp\{-\varphi(\theta)\}$ on $I$. If $\varphi''(\theta)\ge \kappa>0$ on the relative interior of $I$, then, for every absolutely continuous function $f$,
\[
  \operatorname{Var}_{\mu}(f(\theta))
  \le
  \frac{1}{\kappa}\int_I |f'(\theta)|^2\,\mu(d\theta).
\]
In particular, $\operatorname{Var}_{\mu}(\theta)\le 1/\kappa$.
\end{lemma}
\begin{proof}{Proof.}
This is the one-dimensional Brascamp--Lieb variance inequality \citep{brascamp1976extensions}; see also the treatment in \citet{bobkov2000brunn}. The final statement follows from taking $f(\theta)=\theta$. \hfill\qedsymbol
\end{proof}

\begin{lemma}\label{lem:posterior_interval}
Under \Cref{ass:regular_family}, for every $\ell\ge0$ and $\delta\in(0,1)$, the interval $I_{n_\ell}(\delta)$ satisfies \( \Pbb \left[\tilde{\theta}\in I_{n_\ell}(\delta)\right]\ge 1-\delta, \) and there exists $C_I<\infty$ such that
\[
  \mathbb E[\operatorname{diam}(I_{n_\ell}(\delta))]
  \le
  C_I\sqrt{\frac{\log(2/\delta)}{1+\ell}}.
\]
\end{lemma}
\begin{proof}{Proof of \Cref{lem:posterior_interval}.}
Fix a realized history with $n$ observed prices $x_1,\ldots,x_n$. Its posterior effective support is
\[
  \Theta_n(x_{1:n})
  =
  \left\{\theta\in\Theta:\hat\pi(\theta)\prod_{i=1}^n h_\theta(x_i)>0\right\}.
\]
The posterior-regularity part of Assumption~\ref{ass:regular_family} implies that this set is an interval and that the posterior density $q_n$, defined by $\pi_n(d\theta\mid x_{1:n})=q_n(\theta)d\theta$, is strongly log-concave on its relative interior:
\[
  -\partial_\theta^2\log q_n(\theta)
  \ge
  a_0+a_1 n.
\]
Let $\kappa_n:=a_0+a_1n$ and write $q_n(\theta)\propto\exp\{-\varphi_n(\theta)\}$, so that $\varphi_n''(\theta)\ge\kappa_n$.

This gives sub-Gaussian posterior tails around the posterior mean. For $t\ge0$, let
\[
  M_n(t):=\int_{\Theta_n(x_{1:n})}\exp\{t u\}q_n(u)\,du
\]
be the posterior moment generating function, conditional on $x_{1:n}$. The compactness of $\Theta$ makes $M_n(t)$ finite. The exponentially tilted posterior density
\[
  \tilde q_{n,t}(\theta)
  :=
  \frac{\exp\{t\theta\}q_n(\theta)}{M_n(t)}
\]
has negative log-density $\varphi_n(\theta)-t\theta$ up to an additive constant, so the same curvature lower bound applies. By \Cref{lem:brascamp_lieb_variance}, $\operatorname{Var}_{\tilde q_{n,t}}(\theta)\le1/\kappa_n$. For the centered log-moment generating function
\[
  \Lambda_n(t):=\log M_n(t)-t\bar\theta_n,
\]
we have $\Lambda_n(0)=\Lambda_n'(0)=0$ and
\[
  \Lambda_n''(t)
  =
  \operatorname{Var}_{\tilde q_{n,t}}(\theta)
  \le
  \frac{1}{\kappa_n}.
\]
Integrating twice gives $\Lambda_n(t)\le t^2/(2\kappa_n)$ for all $t\ge0$. Therefore, for any $z>0$, Chernoff's inequality gives
\[
  \pi_n(\theta-\bar\theta_n>z\mid x_{1:n})
  \le
  \inf_{t>0}\exp\{-tz+\Lambda_n(t)\}
  \le
  \inf_{t>0}\exp\left\{-tz+\frac{t^2}{2\kappa_n}\right\}
  =
  \exp\left(-\frac{\kappa_n z^2}{2}\right).
\]
Applying the same argument to the tilt proportional to $\exp\{-t\theta\}q_n(\theta)$ gives the same lower-tail bound. A union bound then yields
\[
  \pi_n\!\left(
    |\theta-\bar\theta_n|>z
    \,\middle|\,x_{1:n}
  \right)
  \le
  2\exp\!\left(-\frac{\kappa_n z^2}{2}\right).
\]
With $z=d_n(\delta)=\sqrt{2\log(2/\delta)/\kappa_n}$, the right-hand side equals $\delta$, so $\pi_n(I_n(\delta)\mid x_{1:n})\ge1-\delta$ for every realized history and every fixed $n$. Let $\mathcal H_\ell$ be the random history observed during the learning phase. Averaging the conditional coverage bound over $\mathcal H_\ell$ gives
\[
  \Pbb(\tilde{\theta}\in I_{n_\ell}(\delta))
  =
  \E\!\left[
    \Pbb(\tilde{\theta}\in I_{n_\ell}(\delta)\mid \mathcal H_\ell)
  \right]
  \ge
  1-\delta .
\]

It remains to convert this concentration radius into an expected diameter bound. Since $I_n(\delta)$ is the intersection of $\Theta$ with an interval of radius $d_n(\delta)$,
\[
  \operatorname{diam}(I_{n_\ell}(\delta))
  \le
  2d_{n_\ell}(\delta)
  =
  2\sqrt{\frac{2\log(2/\delta)}{a_0+a_1 n_\ell}}.
\]
Thus it suffices to control $\mathbb E[(a_0+a_1 n_\ell)^{-1/2}]$. Since $n_\ell=1+N_\ell$ with $N_\ell\sim\mathrm{Poisson}(\lambda\ell)$, define $A_\ell:=\{N_\ell\ge\lambda\ell/2\}$. Splitting on $A_\ell$,
\[
  \mathbb E[(a_0+a_1 n_\ell)^{-1/2}]
  \le
  \frac{1}{\sqrt{a_0+a_1(1+\lambda\ell/2)}}
  +
  \frac{\Pbb(A_\ell^c)}{\sqrt{a_0+a_1}}.
\]
The Poisson Chernoff bound gives $\Pbb(A_\ell^c)\le\exp(-\lambda\ell/8)$. The first term in the preceding display is at most $c^{-1/2}(1+\ell)^{-1/2}$ with $c:=\min\{a_0+a_1,a_1\lambda/2\}$, and the exponential term is also bounded by a constant multiple of $(1+\ell)^{-1/2}$. Hence, for a constant $C$ depending only on $(a_0,a_1,\lambda)$,
\[
  \mathbb E[(a_0+a_1 n_\ell)^{-1/2}]
  \le
  C(1+\ell)^{-1/2}.
\]
Combining this bound with the radius inequality gives
\[
  \mathbb E[\operatorname{diam}(I_{n_\ell}(\delta))]
  \le
  2\sqrt{2}\,C
  \sqrt{\frac{\log(2/\delta)}{1+\ell}}.
\]
Taking $C_I:=2\sqrt{2}\,C$ yields the desired result. \hfill\qedsymbol
\end{proof}

\subsubsection{Threshold and Acting-Regret Bounds.}

\begin{lemma}[{\sc Lipschitz Continuity of Oracle Thresholds}]
\label{lem:threshold_lipschitz}
Suppose \Cref{ass:regular_family} holds. Then, for all $\theta,\theta'\in\Theta$,
\[
  \sup_{s>0}|b^\theta(s)-b^{\theta'}(s)|
  \le
  W_\infty(H_\theta,H_{\theta'})
  \le
  L|\theta-\theta'|.
\]
\end{lemma}

\begin{proof}{Proof.}
Fix $\theta,\theta'\in\Theta$, and write $H:=H_\theta$ and $G:=H_{\theta'}$. Define
\[
  F_H(x):=\lambda\,\E_{P\sim H}[(x-P)^+],
  \qquad
  F_G(x):=\lambda\,\E_{Q\sim G}[(x-Q)^+].
\]
By \Cref{prop:threshold}, the corresponding oracle thresholds satisfy $b_H'(s)=-F_H(b_H(s))$ and $b_G'(s)=-F_G(b_G(s))$ for a.e. $s>0$, with right-limit initial conditions $b_H(0^+)=\E_{P\sim H}[\min\{v,P\}]$ and $b_G(0^+)=\E_{Q\sim G}[\min\{v,Q\}]$.

Let $d:=W_\infty(H,G)$. For any $\varepsilon>0$, choose a coupling $(P,Q)$ with marginals $H$ and $G$ such that $|P-Q|\le d+\varepsilon$ almost surely. Then $Q\ge P-(d+\varepsilon)$, and hence, for every $x$, $ (x-Q)^+ \le (x+d+\varepsilon-P)^+ $. Taking expectations gives $ F_G(x) \le F_H(x+d+\varepsilon) $.

The same coupling also gives
\[
  b_H(0^+)-b_G(0^+)
  \le
  \E[|\min\{v,P\}-\min\{v,Q\}|]
  \le
  d+\varepsilon,
\]
because $p\mapsto\min\{v,p\}$ is 1-Lipschitz.

Now define $y_\varepsilon(s):=b_G(s)+d+\varepsilon$. The previous inequality implies
\[
  y_\varepsilon'(s)
  =
  -F_G(b_G(s))
  \ge
  -F_H(b_G(s)+d+\varepsilon)
  =
  -F_H(y_\varepsilon(s))
\]
for a.e. $s>0$, while $y_\varepsilon(0^+)\ge b_H(0^+)$. The function $F_H$ is non-decreasing in its argument, so $-F_H$ is non-increasing. Therefore the scalar comparison argument applies: if $z(s):=b_H(s)-y_\varepsilon(s)$, then on the set $\{z(s)>0\}$,
\[
  z'(s)
  \le
  -F_H(b_H(s))+F_H(y_\varepsilon(s))
  \le
  0.
\]
Since $z(0^+)\le0$, the positive part $z^+(s)$ cannot increase from zero, and thus $b_H(s)\le y_\varepsilon(s)$ for all $s>0$. Hence $b_H(s)\le b_G(s)+d+\varepsilon$. Letting $\varepsilon\downarrow0$ gives $b_H(s)\le b_G(s)+d$ for all $s>0$.

Repeating the same argument with $H$ and $G$ interchanged gives $b_G(s)\le b_H(s)+d$ for all $s>0$. Therefore
\[
  \sup_{s>0}|b_H(s)-b_G(s)|
  \le
  W_\infty(H,G).
\]
Assumption~\ref{ass:regular_family}(ii) then yields $\sup_{s>0}|b^\theta(s)-b^{\theta'}(s)|\le L|\theta-\theta'|$.

\end{proof}

\begin{lemma}[{\sc Acting Regret}]
\label{lem:acting_regret_confidence}
Suppose \Cref{ass:regular_family} holds. There exists $C_R<\infty$ such that, for every learning length $\ell\in[0,T]$ and confidence level $\delta\in(0,1)$,
\[
  \mathcal R(\ell)
  \le
  C_R\sqrt{\frac{\log(2/\delta)}{1+\ell}}
  + v\delta.
\]
\end{lemma}

\begin{proof}{Proof of \Cref{lem:acting_regret_confidence}.}
Fix the realized confidence interval $I:=I_{n_\ell}(\delta)$ and let $\tau_I$ be the stopping rule generated by the upper-envelope threshold $\bar b_I$. On the event $\{\tilde{\theta}\in I\}$, condition on a realization $\tilde{\theta}=\theta$. Then $\theta\in I$, so $\bar b_I(s)\ge b^\theta(s)$ for every $s\in[0,T-\ell]$, hence $\tau_I\le \tau_\theta^*$ pathwise. Therefore, up to time $\tau_I$, both the confidence-envelope policy and the $\theta$-oracle take the same action, namely to wait.

Consider an acting-phase initial state $(s,p)$ with $s>0$; the case $s=0$ is trivial. Let $(U,P)$ denote the time-to-go horizon and current price at the random state where $\tau_I$ stops. The terminal condition gives \(V^\theta(0,P)=(v-P)^+\). Thus, if $U=0$, both the confidence-envelope policy and the oracle face the same terminal decision, and the local loss from stopping at $(U,P)$ is zero. It remains to consider the case $U>0$. Since the confidence-envelope policy stops only when $P\le \bar b_I(U)$ and $\bar b_I(U)\le v$, its stopping payoff is $v-P$.

If $P\le b^\theta(U)$, then the oracle threshold policy also stops at $(U,P)$, so $V^\theta(U,P)=v-P$ and the local loss is zero. If $P>b^\theta(U)$, threshold value matching gives
\[
  V^\theta(U,b^\theta(U))=v-b^\theta(U).
\]
Because $p\mapsto V^\theta(U,p)$ is non-increasing,
\[
  V^\theta(U,P)
  \le
  V^\theta(U,b^\theta(U))
  =
  v-b^\theta(U).
\]
Therefore,
\[
  V^\theta(U,P)-(v-P)
  \le
  P-b^\theta(U)
  \le
  \bar b_I(U)-b^\theta(U)
  \le
  \sup_{u\in(0,s]}(\bar b_I(u)-b^\theta(u)).
\]
Taking conditional expectations over the stopping state gives, on $\{\theta\in I\}$,
\[
  V^\theta(s,p)-J_s^\theta(p,\tau_I)
  \le
  \sup_{u\in(0,s]}(\bar b_I(u)-b^\theta(u)).
\]
By \Cref{lem:threshold_lipschitz}, for every $\eta\in I$ and every $u>0$,
\[
  b^\eta(u)-b^\theta(u)
  \le
  |b^\eta(u)-b^\theta(u)|
  \le
  L|\eta-\theta|.
\]
On $\{\theta\in I\}$, it follows that
\[
  \sup_{u\in(0,s]}(\bar b_I(u)-b^\theta(u))
  \le
  L\,\operatorname{diam}(I).
\]
On the complement event $\{\tilde{\theta}\notin I\}$, the regret is trivially bounded by $v$. Taking expectations gives
\[
  \mathcal R(\ell)
  \le
  L\,\mathbb E[\operatorname{diam}(I_{n_\ell}(\delta))]
  + v\,\mathbb P(\tilde{\theta}\notin I_{n_\ell}(\delta)).
\]
By \Cref{lem:posterior_interval}, the first term is bounded by $LC_I\sqrt{\log(2/\delta)/(1+\ell)}$ and the second probability is at most $\delta$. Taking $C_R:=LC_I$ proves the claim. \hfill\qedsymbol
\end{proof}

\subsubsection{Proof of the Long-Window Bound.}

\begin{proof}{Proof of \Cref{prop:info_gap_confidence}.}
Set $\ell=T/2$ and $\delta=1/(eT)$. By \eqref{eq:gap_two_terms} and \Cref{lem:early_purchase_rate,lem:acting_regret_confidence}, for $T\ge1$,
\[
  \mathcal G(T,\hat\pi)
  \le
  C_{\mathrm E}(1+T/2)^{-1/\alpha}
  +
  C_R\sqrt{\frac{\log(2eT)}{1+T/2}}
  +
  \frac{v}{eT}.
\]
Absorbing constants yields the stated bound. \hfill\qedsymbol
\end{proof}

\subsubsection{Example: Bounded Exponential-Tilt Price-Adjustment Distribution.}
\label{sec:example_exponential_tilt}

We give a simple bounded-support family satisfying both the monotone-likelihood-ratio condition and the regularity conditions used in the long-window information-gap bound. Fix $\bar p>0$ and let
\[
  h_\theta(p)
  =
  \frac{\theta e^{\theta p}}{e^{\theta\bar p}-1},
  \qquad
  0\le p\le \bar p,
  \qquad
  \theta\in[\underline\theta,\bar\theta]\subset(0,\infty).
\]
Higher values of $\theta$ tilt probability mass toward higher prices. Indeed, for $\theta_2>\theta_1$, the likelihood ratio $h_{\theta_2}(p)/h_{\theta_1}(p)$ is proportional to $e^{(\theta_2-\theta_1)p}$ and is increasing in $p$, so \Cref{asm:mlrp} holds.

The family also satisfies \Cref{ass:regular_family}. First, with $p_L^\theta=0$,
\[
  H_\theta([0,x])
  =
  \frac{e^{\theta x}-1}{e^{\theta\bar p}-1}
  \ge
  \frac{\underline\theta}{e^{\bar\theta\bar p}-1}\,x,
  \qquad
  0\le x\le \bar p,
\]
so the lower-tail condition holds with $\alpha=1$. Second, the quantile function is
\[
  Q_\theta(u)
  =
  \frac{1}{\theta}
  \log\left(1+u(e^{\theta\bar p}-1)\right),
  \qquad u\in[0,1].
\]
Because $(\theta,u)\mapsto Q_\theta(u)$ is continuously differentiable on the compact set $[\underline\theta,\bar\theta]\times[0,1]$, there exists $L<\infty$ such that
\[
  W_\infty(H_\theta,H_{\theta'})
  =
  \sup_{u\in[0,1]}|Q_\theta(u)-Q_{\theta'}(u)|
  \le
  L|\theta-\theta'|.
\]
Third, $\Theta(x)=[\underline\theta,\bar\theta]$ for $x\in[0,\bar p]$, and
\[
  -\partial_\theta^2\log h_\theta(x)
  =
  \frac{1}{\theta^2}
  -
  \frac{\bar p^2 e^{-\theta\bar p}}{(1-e^{-\theta\bar p})^2}
  >
  0.
\]
The last expression is continuous and strictly positive on the compact $\Theta$, so it is bounded below by some $a_1>0$. Finally, a truncated Gamma prior on $\Theta$ with density proportional to $\theta^{\alpha_0-1}e^{-\beta_0\theta}$ and $\alpha_0>1$ satisfies
\[
  -\frac{d^2}{d\theta^2}\log\hat\pi(\theta)
  =
  \frac{\alpha_0-1}{\theta^2}
  \ge
  \frac{\alpha_0-1}{\bar\theta^2}
  =:a_0.
\]
Thus \Cref{prop:info_gap_confidence} applies with $\alpha=1$, giving $\mathcal G(T,\hat\pi)=O(T^{-1}+\sqrt{\log T/T})=O(\sqrt{\log T/T})$.

\subsection{Adversarial upper bound for \Cref{prop:optimal_cr}}
We prove \Cref{lem:flash_sale_upper_main} using the sale-termination rule in \eqref{eq:flash_sale_end}.

\begin{proof}{Proof of \Cref{lem:flash_sale_upper_main}.}
We first record the reduction behind the upper bound. By Yao's minimax principle \citep{yao1977probabilistic}, it is enough to construct a flash-sale process under which every deterministic online policy satisfies \(\mathbb E[\ALG]/\mathbb E[\OPT]\le\rho\). Equivalently, if a randomized online algorithm guaranteed competitive ratio $\rho$ on every price path, then under any distribution over price paths its expected payoff would be at least $\rho\,\mathbb E[\OPT]$; but a randomized policy is a mixture of deterministic policies, so an upper bound that holds for every deterministic policy under the constructed distribution also holds for every randomized policy.

If \(p_0=p_L\), the initial price already attains the lowest feasible price. Hence, buying immediately matches the offline oracle on every feasible path, and the competitive ratio is \(1\). Thus, the upper bound is immediate in this boundary case. In the remainder of the proof, assume \(p_0>p_L\).

The first step is to convert the random end-sale time into the markdown reached by the price path before the sale expires. Fix one such distribution and write
\[
q:=\hat p(\tilde t),
\qquad
M:=q-p_L.
\]
Define the realized markdown depth from the reference price \(q\) by
\[
X:=q-\hat p(\sigma)\in[0,M].
\]
For any $x\in[0,M]$, let
\[
t(x):=\tilde t+\frac{t_2-t_1}{p_0-p_L}\,x,
\]
so that $\hat p(t(x))=q-x$. By construction,
\[
\mathbb P(X\ge x)
=
\mathbb P(\sigma\ge t(x))
=
\frac{A}{A+x},
\qquad x\in[0,M].
\]
For \(x<M\), this follows directly from the distribution of \(\sigma\); at \(x=M\), it uses the atom at \(t_2\). Thus the sale-termination rule \eqref{eq:flash_sale_end} is equivalently a distribution over the realized markdown depth \(X\) with survival function
\[
\mathbb P(X\ge x)=\frac{A}{A+x},
\qquad x\in[0,M].
\]
Therefore,
\[
\mathbb E[X]
=
\int_0^M \mathbb P(X\ge x)\,dx
=
\int_0^M \frac{A}{A+x}\,dx
=
A\log\!\left(\frac{A+M}{A}\right).
\]

It remains to bound the payoff of an arbitrary deterministic online policy. Since $\sigma\ge\tilde t$ almost surely and $\hat p$ is strictly decreasing, buying before $\tilde t$ is weakly dominated by waiting until $\tilde t$. After $\tilde t$, while the sale is still active, the policy observes no new uncertainty beyond the deterministic decline of \(\hat p\): the only random event is the sale ending. Therefore, a deterministic policy is characterized, during the active sale, by the first sale price at which it would buy. We write this target price as \(q-y\), or equivalently, use the target markdown \(y\in[0,M]\). If a policy would never buy during the sale, we may upper-bound its payoff by assigning it the target \(y=M\), which only makes the policy weakly better for the upper-bound calculation. When \(v>p_U\), if the target is missed, we also upper bound the policy's continuation payoff by allowing it to buy at the reversion price \(p_U\).

We now consider the four parameter choices separately.

\medskip \noindent \textbf{Case (i): $p_L<v\le p_0$.} Fix $\varepsilon\in(0,v-p_L)$. Here $A=\varepsilon$ and $q=v-\varepsilon$, so $M=v-\varepsilon-p_L$. If the algorithm chooses target markdown $y\in[0,M]$, then it buys at price $q-y$ on the event $\{X\ge y\}$ and earns surplus
\[
v-(q-y)=\varepsilon+y.
\]
If $X<y$, then the sale ends before the target is reached, and buying later at $p_U$ is not beneficial because $v\le p_0\le p_U$; hence, the payoff is $0$. Thus
\[
\mathbb E[\ALG]
=
\frac{\varepsilon}{\varepsilon+y}(\varepsilon+y)
=
\varepsilon.
\]
Since this is independent of $y$, every deterministic online algorithm has expected payoff at most $\varepsilon$.

The offline oracle observes $\sigma$ and buys just before the sale ends, hence its payoff is
\[
v-\hat p(\sigma^-)=v-\hat p(\sigma)=\varepsilon+X,
\]
where the middle equality follows from the continuity of \(\hat p\). Therefore
\[
\mathbb E[\OPT]
=
\varepsilon+\mathbb E[X]
=
\varepsilon+\varepsilon
\log\!\left(\frac{\varepsilon+v-\varepsilon-p_L}{\varepsilon}\right)
=
\varepsilon\left(1+\log\!\left(\frac{v-p_L}{\varepsilon}\right)\right).
\]
Hence
\[
\frac{\mathbb E[\ALG]}{\mathbb E[\OPT]}
\le
\frac{1}{1+\log\!\left(\frac{v-p_L}{\varepsilon}\right)}.
\]
Letting $\varepsilon\downarrow0$ gives a zero upper bound for this valuation regime.

\medskip \noindent \textbf{Case (ii): $p_0<v\le p_U$.} Now $A=v-p_0$ and $q=p_0$, so $M=p_0-p_L$. If the algorithm chooses markdown $y\in[0,M]$, then on $\{X\ge y\}$ it buys at price $p_0-y$ and earns
\[
v-(p_0-y)=A+y.
\]
If $X<y$, then the sale ends too early and buying later at $p_U$ yields a non-positive surplus, since $v\le p_U$. Thus
\[
\mathbb E[\ALG]
=
\frac{A}{A+y}(A+y)
=
A
=
v-p_0.
\]

The oracle again buys just before termination, so its payoff is
\[
v-\hat p(\sigma^-)=v-\hat p(\sigma)=A+X.
\]
Hence
\[
\mathbb E[\OPT]
=
A+A\log\!\left(\frac{A+p_0-p_L}{A}\right)
=
(v-p_0)\left(1+\log\!\left(\frac{v-p_L}{v-p_0}\right)\right).
\]
Therefore
\[
\frac{\mathbb E[\ALG]}{\mathbb E[\OPT]}
\le
\frac{1}{1+\log\!\left(\frac{v-p_L}{v-p_0}\right)}.
\]

\medskip \noindent \textbf{Case (iii): $v>p_U$ and $x_0<p_U-p_0$.} Here $A=p_U-p_0$ and $q=p_0$, so $M=p_0-p_L$. If the algorithm chooses markdown $y\in[0,M]$, then on $\{X\ge y\}$ it buys during the sale at price $p_0-y$ and earns
\[
v-(p_0-y)=v-p_U+A+y.
\]
If $X<y$, the sale ends first. For the purpose of an upper bound, we allow the algorithm to buy at price $p_U$ and obtain a surplus of $v-p_U$. Thus
\[
\mathbb E[\ALG]
\le 
(v-p_U)+\frac{A}{A+y}(A+y)
=
v-p_U+A
=
v-p_0.
\]

The oracle buys just before termination, so its payoff is
\[
v-\hat p(\sigma^-)=v-\hat p(\sigma)=v-p_U+A+X=v-p_0+X.
\]
Therefore
\[
\mathbb E[\OPT]
=
v-p_0+A\log\!\left(\frac{A+p_0-p_L}{A}\right)
=
v-p_0+(p_U-p_0)\log\!\left(\frac{p_U-p_L}{p_U-p_0}\right).
\]
Hence
\[
\frac{\mathbb E[\ALG]}{\mathbb E[\OPT]}
\le
\frac{v-p_0}
{\,v-p_0+(p_U-p_0)\log\!\left(\frac{p_U-p_L}{p_U-p_0}\right)}.
\]

\medskip \noindent \textbf{Case (iv): $v>p_U$ and $x_0\ge p_U-p_0$.} The equation defining $x_0$ has a unique positive solution because $x/(v-p_U)+\log x$ is strictly increasing on $(0,\infty)$. Moreover, $x_0<p_U-p_L$, since the left-hand side evaluated at $p_U-p_L$ exceeds $\log(p_U-p_L)-1$. Thus $q=p_U-x_0$ lies above $p_L$; when $x_0\ge p_U-p_0$, it also satisfies $q\le p_0$, so the corresponding $\tilde t$ is well defined on the flash-sale path. Now $A=x_0$ and $q=p_U-x_0$, so
\[
M=q-p_L=p_U-p_L-x_0.
\]
If the algorithm chooses markdown $y\in[0,M]$, then on $\{X\ge y\}$ it buys during the sale at price $q-y$ and earns
\[
v-(q-y)=v-p_U+x_0+y.
\]
If $X<y$, the sale ends first. For the purpose of an upper bound, we allow the algorithm to buy at $p_U$, obtaining surplus $v-p_U$. Hence
\[
\mathbb E[\ALG]
\le 
(v-p_U)+\frac{x_0}{x_0+y}(x_0+y)
=
v-p_U+x_0.
\]

The oracle buys just before termination, and its payoff is \(v-\hat p(\sigma^-)=v-\hat p(\sigma)=v-p_U+x_0+X\). Thus
\[
\mathbb E[\OPT]
=
v-p_U+x_0+\mathbb E[X]
=
v-p_U+x_0+x_0\log\!\left(\frac{x_0+M}{x_0}\right).
\]
Since $x_0+M=p_U-p_L$, this becomes
\[
\mathbb E[\OPT]
=
v-p_U+x_0+x_0\log\!\left(\frac{p_U-p_L}{x_0}\right).
\]
Using the defining equation of $x_0$,
\[
\log\!\left(\frac{p_U-p_L}{x_0}\right)
=
1+\frac{x_0}{v-p_U},
\]
and therefore
\[
\mathbb E[\OPT]
=
v-p_U+x_0+x_0\left(1+\frac{x_0}{v-p_U}\right)
=
\frac{(v-p_U+x_0)^2}{v-p_U}.
\]
It follows that
\[
\frac{\mathbb E[\ALG]}{\mathbb E[\OPT]}
\le
\frac{v-p_U+x_0}{(v-p_U+x_0)^2/(v-p_U)}
=
\frac{v-p_U}{v-p_U+x_0}.
\]

Combining the four cases proves \(\mathrm{CR}(\mathcal A)\le\rho(v)\) for every online policy \(\mathcal A\).\hfill\qedsymbol
\end{proof}
\subsection{Proof of \Cref{lem:matching_threshold_policy}}\label{sec:proof:lem:matching_threshold_policy}
\begin{proof}{Proof of \Cref{lem:matching_threshold_policy}.}
The case \(v\le p_0\) is immediate because the target ratio is \(0\). Consider \(v>p_0\), and fix an arbitrary feasible price path \(p(\cdot)\). Let
\[
p^\star:=\inf_{t\in[0,T]}p(t).
\]
Since \(p(0)=p_0\), we have \(p^\star\le p_0\). Because \(v>p_0\), the offline value is \(\OPT(p)=v-p^\star\). If the threshold \(b\) is reached before the deadline, the purchase price is at most \(b\). If the threshold is not reached before the deadline and \(v>p_U\), the terminal valuation rule purchases at \(T\) and yields a payoff of at least \(v-p_U\). Thus, it suffices to lower-bound the randomized policy's payoff as a function of \(p^\star\). Boundary events where \(b=p^\star\) do not affect the integrals below; the only atom is at \(p_0\), and that threshold is reached at the initial price.

\medskip \noindent \textbf{Case (i): $p_0<v\le p_U$.} The CDF is valid because
\[
F_b(p_0-)
=
\rho\log\!\left(\frac{v-p_L}{v-p_0}\right)
=
1-\rho
<1.
\]
Hence there is an atom of size $\rho$ at $p_0$, and on $[p_L,p_0)$ the distribution has density
\[
f_b(z)=F_b'(z)=\frac{\rho}{v-z}.
\]
Since \(v\le p_U\), any payoff from the terminal valuation rule can only improve the lower bound. When \(b\ge p^\star\), the policy obtains realized surplus at least \(v-b\): either the threshold is reached before the deadline, or the terminal valuation rule buys at \(T\) if the relevant crossing occurs only at the endpoint. Therefore
\[
\mathbb E[\ALG\mid p^\star]
\ge
\rho(v-p_0)+\int_{p^\star}^{p_0}(v-z)\frac{\rho}{v-z}\,dz
=
\rho(v-p_0)+\rho(p_0-p^\star)
=
\rho(v-p^\star).
\]
Thus
\[
\frac{\mathbb E[\ALG\mid p^\star]}{\OPT(p)}\ge\rho
\]
for every feasible price path.

\medskip \noindent \textbf{Case (ii): $v>p_U$ and $x_0<p_U-p_0$.} Let
\[
D:=v-p_U,
\qquad
C:=p_U-p_0,
\qquad
L:=\log\!\left(\frac{p_U-p_L}{p_U-p_0}\right).
\]
Then
\[
\rho=\frac{D+C}{D+C+CL}.
\]
We first check that the CDF is valid. Since $x_0<C$ and the function
\[
h(x):=\frac{x}{v-p_U}+\log x
\]
is strictly increasing on $(0,\infty)$, the defining equation
\[
h(x_0)=\log(p_U-p_L)-1
\]
implies
\[
h(C)>h(x_0),
\]
that is,
\[
\frac{C}{v-p_U}+\log C>\log(p_U-p_L)-1.
\]
Rearranging gives
\[
L<1+\frac{C}{D}.
\]
Hence
\[
F_b(p_0-)=\rho L
=
\frac{(D+C)L}{D+C+CL}
<1.
\]
So the CDF indeed has an atom of size $1-\rho L$ at $p_0$, and density
\[
f_b(z)=\frac{\rho}{p_U-z}
\qquad\text{on }[p_L,p_0).
\]

If the threshold is reached before the deadline, the agent pays at most \(b\) and gets a payoff of at least \(v-b=D+(p_U-b)\). If the threshold is not reached before the deadline, the terminal valuation rule buys at \(T\) and yields a payoff of at least \(D\). Therefore
\[
\mathbb E[\ALG\mid p^\star]
\ge
D+\mathbb E\!\left[(p_U-b)\mathbf 1\{b\ge p^\star\}\right].
\]
Using the atom at $p_0$ and the density on $[p_L,p_0)$,
\[
\mathbb E[\ALG\mid p^\star]
\ge
D+(1-\rho L)(p_U-p_0)+\int_{p^\star}^{p_0}(p_U-z)\frac{\rho}{p_U-z}\,dz.
\]
That is,
\[
\mathbb E[\ALG\mid p^\star]
\ge
D+(1-\rho L)C+\rho(p_0-p^\star).
\]
Since
\[
1-\rho L
=
1-\frac{D+C}{D+C+CL}L,
\]
we obtain
\[
D+(1-\rho L)C
=
D+C-\rho CL
=
\rho(D+C)
=
\rho(v-p_0).
\]
Therefore
\[
\mathbb E[\ALG\mid p^\star]
\ge
\rho(v-p_0)+\rho(p_0-p^\star)
=
\rho(v-p^\star).
\]
Thus
\[
\frac{\mathbb E[\ALG\mid p^\star]}{\OPT(p)}\ge\rho
\]
for every feasible price path.

\medskip \noindent \textbf{Case (iii): $v>p_U$ and $x_0\ge p_U-p_0$.} Let
\[
D:=v-p_U,
\qquad
q:=p_U-x_0,
\qquad
\rho:=\frac{D}{D+x_0}.
\]
The cutoff \(q\) is admissible. The case condition gives \(q\le p_0\). Moreover, \(x\mapsto x/D+\log x\) is strictly increasing, and at \(x=p_U-p_L\) its value is larger than \(\log(p_U-p_L)-1\). The defining equation for \(x_0\) therefore implies \(x_0<p_U-p_L\), so \(q>p_L\). By the defining equation of $x_0$,
\[
\log\!\left(\frac{p_U-p_L}{x_0}\right)=1+\frac{x_0}{D},
\]
and hence
\[
\rho\log\!\left(\frac{p_U-p_L}{x_0}\right)
=
\frac{D}{D+x_0}\left(1+\frac{x_0}{D}\right)
=
1.
\]
Therefore, the CDF is valid and continuous at $q$, with density
\[
f_b(z)=\frac{\rho}{p_U-z}
\qquad\text{on }[p_L,q).
\]

If \(p^\star\le q\), then the event \(\{b\ge p^\star\}\) is the same as \(b\in[p^\star,q]\), up to a boundary event. The terminal valuation rule contributes at least \(D\) whenever the threshold is not reached before the deadline, and a reached threshold contributes the additional term \(p_U-b\). Hence
\[
\mathbb E[\ALG\mid p^\star]
\ge
D+\int_{p^\star}^{q}(p_U-z)\frac{\rho}{p_U-z}\,dz
=
D+\rho(q-p^\star).
\]
Since $q=p_U-x_0$ and $D=\rho(D+x_0)$,
\[
\mathbb E[\ALG\mid p^\star]
\ge
D+\rho(p_U-x_0-p^\star)
=
\rho(D+p_U-p^\star)
=
\rho(v-p^\star).
\]

If \(p^\star>q\), then no threshold in the support is reached before the deadline, so the terminal valuation rule yields payoff at least \(D\):
\[
\mathbb E[\ALG\mid p^\star]\ge D.
\]
Because \(p^\star>q=p_U-x_0\), we have \(p_U-p^\star<x_0\), and thus
\[
v-p^\star=D+(p_U-p^\star)<D+x_0.
\]
Therefore
\[
\mathbb E[\ALG\mid p^\star]
\ge
D
>
\frac{D}{D+x_0}(v-p^\star)
=
\rho(v-p^\star).
\]

Combining the two subcases yields
\[
\mathbb E[\ALG\mid p^\star]\ge \rho(v-p^\star)
\]
for every feasible price path. Since \(\OPT(p)=v-p^\star\), the policy achieves competitive ratio at least
\[
\rho=\frac{v-p_U}{v-p_U+x_0}.
\] 
This completes the proof. \hfill\qedsymbol
\end{proof}
\subsection{Proofs for Minimax Regret}
\begin{proof}{Proof of \Cref{lem:regret_lower_bound}.}
It suffices, by Yao's principle, to construct a distribution over feasible price paths under which every deterministic online policy has expected regret at least \(\bar R(v)\). Throughout the proof, \(a\), \(\gamma\), and \(\Delta\) are the quantities defined immediately before \(\bar R(v)\); let \(\sigma\) be defined as in \Cref{lem:regret_lower_bound}. By construction, \(p_L\le \gamma\le p_0\), \(\Delta>0\), and
\begin{equation}\label{eq:Lemma_6_1}
\Delta\log\frac{a-p_L}{\Delta}=\bar R(v).    
\end{equation}

Indeed, if \(\gamma=p_0\), then \(\Delta=a-p_0\); if \(\gamma<p_0\), then \(\Delta=(a-p_L)/e\), so equation \eqref{eq:Lemma_6_1} equals \((a-p_L)/e\).

Let \(X:=\hat p(\sigma)\) be the terminal sale price. For \(x\in(p_L,\gamma)\), let \(t(x)\) be the unique time satisfying \(\hat p(t(x))=x\). The end-sale distribution gives
\[
    \mathbb P(X\le x)
    =
    \mathbb P(\sigma\ge t(x))
    =
    \frac{\Delta}{a-x},
    \qquad p_L<x<\gamma.
\]
At \(x=p_L\), the same CDF value is induced by the atom at \(t_2\). Thus \(X\) has CDF
\[
G(x):=
\begin{cases}
0, & x<p_L,\\[0.6ex]
\dfrac{\Delta}{a-x}, & p_L\le x<\gamma,\\[1.2ex]
1, & x\ge \gamma.
\end{cases}
\]
Consider any deterministic online policy. Before the jump, all paths share the same declining price history. Hence, the policy's first purchase along this common history can be represented by a target price \(y\): the first price at which it would buy if the decline continued all the way to \(p_L\). If the policy would not buy before \(p_L\), set \(y=0<p_L\).

If \(y\ge \gamma\), every path reaches \(\gamma\) before reversion and the policy buys before reversion at a price at least \(\gamma\). Its online payoff is therefore at most \(v-\gamma\), whereas the offline payoff is \(v-X\). Its regret is at least \(\gamma-X\), and hence its expected regret is at least
\[
    \mathbb E[\gamma-X]
    =
    \int_{p_L}^{\gamma} G(x)\,dx
    =
    \Delta\log\frac{a-p_L}{\Delta}
    =
    \bar R(v).
\]

If \(p_L\le y<\gamma\), then the policy buys before reversion on \(\{X\le y\}\) and misses the target on \(\{X>y\}\). In the first event, regret is at least \(y-X\). In the second event, even granting the policy the best possible continuation payoff after reversion, regret is at least \(a-X\): when \(v\le p_U\), this is \(v-X\), and when \(v>p_U\), it is \(p_U-X\). Thus, its expected regret is at least
\[
    r(y)
    :=
    \int_{p_L}^{y}(y-x)\,dG(x)
    +
    \int_{y}^{\gamma}(a-x)\,dG(x).
\]
The CDF has an atom at \(p_L\) and density \(g(x)=\Delta/(a-x)^2\) on \((p_L,\gamma)\). For \(y\in(p_L,\gamma)\), the atom contributes to the first integral and is already captured by \(G(y)\). Hence
\[
    r'(y)=G(y)-(a-y)g(y)=0,
\]
so \(r(y)\) is constant on \((p_L,\gamma)\). Taking the limit as \(y\uparrow \gamma\) gives \(r(y)=\mathbb E[\gamma-X]=\bar R(v)\). Finally, if \(y<p_L\), the policy does not buy before reversion on any path, and its regret is at least \(a-X\ge \gamma-X\) pathwise. Its expected regret is therefore also at least \(\bar R(v)\).

Thus, every deterministic policy has expected regret at least \(\bar R(v)\) under the constructed distribution. A randomized online policy is a mixture of deterministic policies, so its expected regret under the same path distribution is also at least \(\bar R(v)\). Consequently, for every online policy \(\mathcal A\), there exists a feasible path with \(\Reg(\mathcal A,p)\ge\bar R(v)\), proving the claim. \hfill\qedsymbol
\end{proof}

\begin{proof}{Proof of \Cref{lem:regret_achievable}.}
For a feasible price path, write
\[
    p^*:=\inf_{t\in[0,T]}p(t).
\]
Since \(p(0)=p_0\), we have \(p^*\le p_0\). Boundary events at which the path attains exactly the realized threshold do not affect the integrals below.

If a threshold \(b\) is reached before the deadline, the policy purchases at a price no higher than \(b\), so its regret relative to the offline minimum is at most \(b-p^*\). If the threshold is not reached before the deadline, the terminal valuation rule yields a nonnegative payoff when \(v\le p_U\), so the regret is at most \(v-p^*\). When \(v>p_U\), the terminal valuation rule purchases at the deadline; because the deadline price is no higher than \(p_U\), the regret is at most \(p_U-p^*\). Thus the missed-threshold penalty is governed by \(a-p^*\); in Case~(i), this reduces to \(v-p^*\).

\medskip \noindent\textbf{Case~(i): $p_L < v \le p_0$.} Set \(\bar p:=v-(v-p_L)/e\), draw \(b\) with density
\[
    f(b)=\frac{1}{v-b}\mathbf 1_{\{p_L\le b\le \bar p\}},
\]
This is a probability density because
\[
    \int_{p_L}^{\bar p}\frac{db}{v-b}
    =
    \log\frac{v-p_L}{v-\bar p}
    =1.
\]
If \(p^*\le\bar p\), then with \(F(x):=\int_{p_L}^x f(b)\,db\),
\[
    \Reg(\mathcal A,p)
    \le
    \int_{p^*}^{\bar p}(b-p^*)\frac{db}{v-b}
    +F(p^*)(v-p^*).
\]
Since \(F(p^*)=\log((v-p_L)/(v-p^*))\), the right-hand side equals
\[
    -(\bar p-p^*)+(v-p^*)\log\frac{v-p_L}{v-\bar p}
    =
    v-\bar p
    =
    \frac{v-p_L}{e}.
\]
If \(p^*>\bar p\), then no threshold in the support is reached and the regret is at most \((v-p^*)^+<v-\bar p\). Hence \(\sup_{p(\cdot)}\Reg(\mathcal A,p)\le (v-p_L)/e\).

\medskip \noindent\textbf{Case~(ii): $p_0 < v \le p_U$.} Let \(L_v:=\log((v-p_L)/(v-p_0))\). If \(L_v\le1\), draw \(b\) with density \(f(b)=1/(v-b)\) on \([p_L,p_0)\) and put an atom \(q_0:=1-L_v\) at \(p_0\). If \(p^*\in[p_L,p_0]\), then
\[
    \Reg(\mathcal A,p)
    \le
    q_0(p_0-p^*)
    +\int_{p^*}^{p_0}(b-p^*)\frac{db}{v-b}
    +F(p^*)(v-p^*),
\]
where \(F(p^*)=\int_{p_L}^{p^*}db/(v-b)\). The derivative of the right-hand side with respect to \(p^*\) is zero, so the bound is constant on \([p_L,p_0]\). Evaluating at \(p^*=p_0\) gives
\[
    \Reg(\mathcal A,p)
    \le
    (1-q_0)(v-p_0)
    =
    (v-p_0)\log\frac{v-p_L}{v-p_0}.
\]
If \(L_v>1\), set \(\bar p:=v-(v-p_L)/e\). Then \(\bar p<p_0\), and the truncated density \(f(b)=1/(v-b)\) on \([p_L,\bar p]\) is a probability density. The same calculation as in Case~(i) gives \(\sup_{p(\cdot)}\Reg(\mathcal A,p)\le (v-p_L)/e\). The two subcases imply \(\sup_{p(\cdot)}\Reg(\mathcal A,p)\le\bar R(v)\).

\medskip \noindent\textbf{Case~(iii): $v > p_U$.} The same construction applies with \(p_U\) replacing \(v\), together with the terminal valuation rule if the threshold is not reached before the deadline. Since \(v>p_U\), this terminal rule always purchases at \(T\). Let \(L_U:=\log((p_U-p_L)/(p_U-p_0))\), with the convention that \(L_U=+\infty\) when \(p_U=p_0\). If \(L_U\le1\), draw \(b\) with density \(f(b)=1/(p_U-b)\) on \([p_L,p_0)\) and put an atom \(q_0:=1-L_U\) at \(p_0\). Since a missed threshold costs at most \(p_U-p^*\), the same derivative calculation gives
\[
    \Reg(\mathcal A,p)
    \le
    (p_U-p_0)\log\frac{p_U-p_L}{p_U-p_0}.
\]
If \(L_U>1\), use the truncated density \(f(b)=1/(p_U-b)\) on \([p_L,\bar p_U]\), where \(\bar p_U:=p_U-(p_U-p_L)/e\). Then \(\sup_{p(\cdot)}\Reg(\mathcal A,p)\le (p_U-p_L)/e\). Hence \(\sup_{p(\cdot)}\Reg(\mathcal A,p)\le\bar R(v)\).  \hfill\qedsymbol
\end{proof}

\section{Decreasing Price Envelopes}
\label{ec:robust_trend_envelopes}

This subsection explains the reduction used in \Cref{subsec:robust_trend_envelopes}. Let \(p_L(\cdot)\) and \(p_U(\cdot)\) be non-increasing envelopes. Define the effective primitives
\[
  p_L^{\mathrm{eff}}:=p_L(T),\qquad
  p_U^{\mathrm{eff}}:=p_U(T),\qquad
  p_0^{\mathrm{eff}}:=\min\{p_0,p_U(T)\}.
\]

For every feasible path, the ex-post minimum \(p^\star:=\inf_{t\in[0,T]}p(t)\) satisfies \(p^\star\ge p_L^{\mathrm{eff}}\). In addition, because the agent observes \(p_0\) at the beginning and can wait until the deadline, the best realized price is no higher than \(p_0^{\mathrm{eff}}\); equivalently, \(p^\star\le p_0^{\mathrm{eff}}\) whenever the effective interval is nondegenerate. If a randomized threshold is not reached before the deadline, the deadline price is at most \(p_U^{\mathrm{eff}}\), so the missed-threshold continuation bound used in the constant-bound proof remains valid with \(p_U\) replaced by \(p_U^{\mathrm{eff}}\).

The achievability arguments for \Cref{lem:matching_threshold_policy,lem:regret_achievable} therefore go through verbatim after replacing
\[
  (p_L,p_U,p_0)
  \quad\text{by}\quad
  (p_L^{\mathrm{eff}},p_U^{\mathrm{eff}},p_0^{\mathrm{eff}}).
\]
The randomized threshold distributions depend only on the lower bound for \(p^\star\), the largest relevant upper bound on the best realized price, and the upper bound on the deadline price when the threshold is not reached before the deadline. These are precisely \(p_L^{\mathrm{eff}}\), \(p_0^{\mathrm{eff}}\), and \(p_U^{\mathrm{eff}}\), respectively.

The adversarial optimality constructions can also be embedded in the decreasing-trend class. Fix a small terminal interval \([T-\eta,T]\). Because the envelopes are non-increasing and continuous at \(T\), their values on this interval can be made arbitrarily close to the terminal band \([p_L^{\mathrm{eff}},p_U^{\mathrm{eff}}]\) by taking \(\eta\) small. The adversary first chooses a feasible decreasing path that reaches the effective initial level \(p_0^{\mathrm{eff}}\) at the beginning of the markdown experiment. From that point, the path follows a decreasing trajectory from \(p_0^{\mathrm{eff}}\) toward \(p_L^{\mathrm{eff}}\). The random interruption time is the end-sale time \(\sigma\) from the flash-sale construction: at \(\sigma\), the markdown expires, the price jumps upward to the upper-envelope path, and then follows \(p(t)=p_U(t)\) until the deadline. Thus, the only uncertainty faced by the online rule is the realized markdown depth before the end of the sale. Up to an arbitrarily small envelope-approximation error, the same randomization over this markdown depth equalizes deterministic target-threshold rules exactly as in \Cref{lem:flash_sale_upper_main} and \Cref{lem:regret_lower_bound}. Hence, the formulas remain the optimal guarantees for the decreasing-trend problem under the effective primitives \((p_L^{\mathrm{eff}},p_U^{\mathrm{eff}},p_0^{\mathrm{eff}})\).

\PutSingleSpacedBib
\end{bibunit}

\end{document}